\documentclass[11pt]{amsart}
\usepackage{graphicx}  
\usepackage{subfig}  
\usepackage{amssymb, amsmath, amsthm}
\usepackage{color}
\usepackage{float}
\usepackage{enumerate}
\usepackage{graphics}
\usepackage{epstopdf}
\usepackage{tikz}
\usetikzlibrary{arrows,decorations.markings}
\usepackage{cite}
\usepackage{epsfig}
\usepackage{amsmath}
\usepackage{amssymb}
\usepackage{amscd}
\usepackage{graphicx}
\usepackage{mathrsfs}
\usepackage{needspace}
\usepackage[T1]{fontenc}
\usepackage[colorlinks=true,linkcolor=blue,citecolor=red]{hyperref}

\makeatletter
\newcommand{\equationlabels}[1]{%
  \begingroup
  \edef\@currentlabel{\theequation}%
  \edef\@currentHref{equation.\theHequation}%
  #1%
  \endgroup}
\makeatother
\numberwithin{equation}{section}
\newtheorem{theorem}{Theorem}[section]

\newtheorem{lemma}[theorem]{Lemma}
\newtheorem{prop}{Proposition}[section]
\newtheorem{assumption}[theorem]{Assumption}

\newtheorem{dbar-RHP}[theorem]{$\bar{\partial}$-RH problem}
\theoremstyle{definition}

\newtheorem{re}{Remark}[section]
\newtheorem{rhp}{RH Problem}
\newtheorem{rhpa}{RH Problem}[section]

\newcommand{\R}{\mathbb{R}}


\begin{document}
	
	\title[Long-time asymptotics for the Sawada-Kotera equation]{On the Cauchy Problem for the Sawada-Kotera Equation: Soliton resolution conjecture, Painlev\'e\ transcendents and asymptotic stability}
	
	\author[Z. K. Huang]{Zheng-Kang Huang}
	
 \author[S. F. Tian]{Shou-Fu Tian$^{*,\dagger}$}

	\address{Zheng-Kang Huang and Shou-Fu Tian\newline
School of Mathematics, China University of Mining and Technology, Xuzhou 221116, People's Republic of China}
\email{sftian@cumt.edu.cn, shoufu2006@126.com}
	
	\thanks{$^{*}$Corresponding author(sftian@cumt.edu.cn, shoufu2006@126.com).\\
\hspace*{3ex}$^\dag$Contributed equally as the first author.}
	
\begin{abstract}
We derive long-time asymptotics in different space-time regions for the Sawada-Kotera (SK) equation with admissible data involving continuous spectrum and finite discrete spectrum. Our result proves the soliton resolution conjecture, yields the asymptotics governed by the F-XVIII Painlev\'e transcendent, and establishes the asymptotic stability of \(N\)-soliton solutions. More importantly, this extends the long-time asymptotic theory to the coexistence of solitons and radiation, whereas previous results were restricted to the purely continuous spectrum. One key ingredient in our argument is a new meromorphic Riemann-Hilbert framework that incorporates the discrete spectrum through complete six-point pole orbits while preserving cofactor analyticity and the singular structure at the spectral origin. A second ingredient is the combination of an orbitwise, region-dependent pole reduction mechanism and the \(\bar{\partial}\)-nonlinear steepest descent method, which allows us to select the solitons contributing in each region and separate them from the continuous radiation. Our framework also resolves the coalescing stationary points at the singular spectral origin through the regular modified SK problem and the Miura map.
\end{abstract}

\maketitle
\tableofcontents
\section{Introduction}
The inverse spectral approach to integrable Cauchy problems associates the nonlinear evolution with a scalar differential operator. For rapidly decreasing coefficients, a convenient normalized form is
\[
\mathscr H_m=\partial_x^m+\sum_{j=0}^{m-2}a_j(x,t)\partial_x^j,\qquad m\ge2,
\]
where $a_j(\cdot,t)$ belong to the Schwartz class $\mathcal S(\mathbb R)$. Direct and inverse scattering for such operators on the line were developed in \cite{b2,b3}, and the treatment of arbitrary spectral singularities was given in \cite{d2}. The second-order case $\mathscr H_2=\partial_x^2+a_0$ is the Schr\"odinger spectral problem underlying the Korteweg-de Vries (KdV) equation, whose inverse scattering theory originated with Gardner, Greene, Kruskal, and Miura \cite{g1} and was further developed by Deift and Trubowitz \cite{d1}.

For the third-order operator $\mathscr H_3=\partial_x^3+a_1\partial_x+a_0$, different constraints on the coefficients lead to different integrable systems \cite{k1,w3}. For instance, the choice $a_1=2f$ and $a_0=f_x+g$ gives the ``good'' Boussinesq system \cite{c14},
\[
f_t=g_x,\qquad 3g_t+f_{xxx}+8ff_x=0.
\]
The reduction $a_1=6u$ and $a_0=0$ yields the Sawada-Kotera (SK) equation as a fifth-order Lax flow. This equation was introduced by Sawada and Kotera and independently derived by Caudrey, Dodd, and Gibbon \cite{c1,s1}. In this work, we study its Cauchy problem with decaying initial data,
\begin{equation}\label{1.1}
\begin{cases}
u_t+u_{xxxxx}+30(uu_{xxx}+u_xu_{xx})+180u^2u_x=0,\\
u(x,0)=u_0(x)\in\mathcal S(\mathbb R).
\end{cases}
\end{equation}
The choice $a_1=6U$ and $a_0=3U_x$ gives another fifth-order flow, the Kaup-Kupershmidt (KK) equation \cite{k1},
\[
U_t+U_{xxxxx}+30UU_{xxx}+75U_xU_{xx}+180U^2U_x=0.
\]
Fordy and Gibbons showed that these two fifth-order equations are connected by Miura transformations to a common modified equation \cite{f1}. In the present normalization, the modified Sawada-Kotera (mSK) equation is
\begin{equation}\label{1.2}
w_t+w_{xxxxx}-\bigl(5w_xw_{xx}+5ww_x^2+5w^2w_{xx}-w^5\bigr)_x=0,
\end{equation}
and its transformation to SK is
\begin{equation}\label{1.3}
u=\frac16(w_x-w^2).
\end{equation}
The corresponding KK variable is obtained from the same modified field by $U=-(2w_x+w^2)/6$.

The SK and KK equations also arise as integrable normal forms for long, weakly nonlinear waves in liquids containing gas bubbles. In these media, pressure disturbances are coupled to bubble motion, whose inertia contributes to dispersion. Retaining corrections beyond the leading long-wave approximation introduces higher-order nonlinear and dispersive terms \cite{k2,k3}. When dissipation is negligible, suitable near-identity transformations and changes of variables render these models asymptotically equivalent to the SK or KK equation. This connection makes the SK equation a tractable model for studying the balance between nonlinear wave interactions and fifth-order dispersion. Its soliton solutions provide explicit examples of persistent coherent waves \cite{s1}. For localized initial disturbances, this raises questions about the separation of coherent waves from dispersive radiation, the position shifts caused by interactions, and the response of the soliton component to perturbations. These dynamical questions naturally lead to long-time asymptotics, soliton resolution, and stability. The integrable structure of SK makes them accessible through spectral analysis.

In this spectral framework, the inverse scattering transform (IST) reconstructs the nonlinear wave from its scattering data and provides a basis for asymptotic analysis \cite{a2}. It was extended to the nonlinear Schr\"odinger equation by Zakharov and Shabat \cite{z2} and systematized by Ablowitz, Kaup, Newell, and Segur \cite{a1}. Direct and inverse scattering formulations for shallow-water systems were subsequently developed in \cite{b7,b8,c6,c7,c8,c9,c10}. The continuous scattering data describe radiation, while discrete eigenvalues and norming constants describe solitons. Their simultaneous presence therefore requires a long-time analysis that separates both contributions in the prescribed space-time region.

The inverse problem is commonly formulated as a matrix Riemann-Hilbert (RH) problem. Following the Beals-Coifman formulation \cite{b1}, Deift and Zhou introduced the nonlinear steepest descent method for oscillatory RH problems \cite{d3}, with further developments in \cite{c4,d4,g3,l2}. Its principal ingredients are matrix factorizations, scalar conjugations, contour deformations, local models at stationary points, and small-norm estimates. The $\bar\partial$ extension of this method, introduced by McLaughlin and Miller \cite{m1,m2}, replaces analytic continuations by nonanalytic extensions controlled through a $\bar\partial$-problem. The resulting mixed problem separates into a pure RH model and a $\bar\partial$ remainder. This approach has been used to study long-time asymptotics, soliton resolution, and asymptotic stability for many integrable equations \cite{b4,c5,c13,l3,l4,t1,y1}.

Examples with $2\times2$ Lax pairs include the focusing nonlinear Schr\"odinger, Camassa-Holm, and focusing Kundu-Eckhaus equations \cite{b5,b6,l6,r1,w1,x1}. Higher-rank problems involve more spectral sectors and additional symmetry constraints. Representative $3\times3$ examples include the Degasperis-Procesi and Novikov equations \cite{b7,b8}, the Yajima-Oikawa system \cite{z1}, the Sasa-Satsuma equation \cite{h1,l5,x2}, the ``good'' Boussinesq equation \cite{c2,c3}, and the Tzitz\'eica equation \cite{h2}. The spin-1 Gross-Pitaevskii equation provides a $4\times4$ example \cite{g2,t1}.

For the SK hierarchy, Wang and Zhu developed a $3\times3$ scattering formulation for the modified Sawada-Kotera (mSK) equation \cite{w2} and derived the long-time asymptotics of the SK and mSK equations by the Deift-Zhou method \cite{w3}. Their analysis describes the decay, dispersive-wave, Painlev\'e, and rapid-decay regions under the assumption of no discrete spectrum. In particular, the self-similar asymptotics are governed by the fourth-order Painlev\'e-type F-XVIII equation. Reflectionless SK solitons have also been constructed by inverse scattering and Marchenko-type methods \cite{a3,a4}. These results naturally raise the question of how to obtain long-time asymptotics when continuous and discrete spectra coexist, and how to establish soliton resolution and the asymptotic stability of the resulting soliton component.

In this work, we formulate a meromorphic $3\times3$ RH problem with complete six-point pole orbits and apply the $\bar\partial$-nonlinear steepest descent method. The analysis yields asymptotic formulas in four space-time regions, with the contributing solitons selected by their velocities and the continuous spectrum entering through radiation and transmission factors. Under the scalar closure and matrix realization conditions stated in the body, the reflectionless component decomposes into separated one-solitons with explicit position shifts. Combining this decomposition with perturbative scattering estimates gives regional asymptotic stability.

In the Painlev\'e region, the stationary points coalesce at the singular spectral origin. We use the regular mSK equation~\eqref{1.2} together with the Miura transformation~\eqref{1.3}.
A related Miura correspondence appears in \cite{c2}, and analogous Painlev\'e mechanisms occur in \cite{c4,c5,h1,l6}. With $s=x/(5t)^{1/5}$, the mSK profile satisfies the F-XVIII equation in Cosgrove's classification \cite{c11,c12}, and \eqref{1.3} transfers its asymptotics to SK. The corresponding mSK results are obtained by analogous deformations.

\textbf{Our work is arranged as follows:}
The remainder of the work is organized as follows. Section~\ref{sec2} develops the scattering data and the meromorphic RH problems. Sections~\ref{sec3}-\ref{sec5} analyze the positive-velocity, Painlev\'e, and negative-velocity regions. Section~\ref{sec6} gives the moving-to-frozen comparison, soliton resolution, and asymptotic stability.

\section{Inverse scattering transform}\label{sec2}
\subsection{Jost functions}
The SK Lax pair is $\Phi_x=L\Phi$, $\Phi_t=Z\Phi$, where
\[
L=\begin{pmatrix}0&1&0\\0&0&1\\k^3&-6u&0\end{pmatrix},
\]
and $Z$ is the temporal SK Lax matrix in \cite{w3}.

Introduce the following gauge transformation
\[
\Phi (x,t;k)=G(k)\Psi (x,t;k),\qquad G(k)=\left( \begin{matrix}
	\omega&		\omega ^2&		1\\
	\omega ^2k&		\omega k&		k\\
	k^2&		k^2&		k^2\\
\end{matrix} \right) ,\qquad \omega =e^{2\pi i/3}.
\]
Then
\[
G^{-1}LG=k\Lambda +Q,\qquad \Lambda =\mathrm{diag}\mathrm{(}\omega ,\omega ^2,1),\qquad Q=k^{-1}Q_1,\quad Q_1=-2u\,er^{\mathrm{T}},
\]
where \(e=(1,1,1)^{\mathrm T}\), \(r=(\omega^2,\omega,1)^{\mathrm T}\), and \(r^{\mathrm T}e=0\). The temporal part is similarly $G^{-1}ZG=9k^5\Lambda^2+P(x,t;k)$, with $P\to0$ as $|x|\to\infty$. Thus
\[
\Psi_x=(k\Lambda+Q)\Psi,\qquad \Psi_t=(9k^5\Lambda^2+P)\Psi.
\]
Writing $\Psi=J\mathrm e^{k\Lambda x+9k^5\Lambda^2t}$ gives
\[
\begin{cases}
	J_x-[k\Lambda ,J]=QJ,\\
	J_t-[9k^5\Lambda ^2,J]=PJ.\\
\end{cases}
\]
The Jost solutions are uniquely determined by
\[
J_\pm(x;k)=I_3-\int_x^{\pm\infty}e^{(x-y)\widehat{k\Lambda}}\bigl(Q(y;k)J_\pm(y;k)\bigr)\,dy,\qquad e^{\widehat{k\Lambda}}A=e^{k\Lambda}Ae^{-k\Lambda}.
\]
Set
\[
\Sigma =\mathbb{R} \cup \omega \mathbb{R} \cup \omega ^2\mathbb{R} ,\qquad D_n=\left\{ k:\frac{(n-1)\pi}{3}<\mathrm{arg}k<\frac{n\pi}{3} \right\} ,\quad n=1,...,6.
\]
Each half-ray is oriented from \(0\) to infinity, with its \(+\) side on the left.

Put $S=D_3\cup D_4$. Volterra theory gives $\det J_\pm=1$. The columns of $J_+$ are analytic in $\omega^2S,\omega S,S$, and those of $J_-$ in $-\omega^2S,-\omega S,-S$. Away from \(0\), they extend with the required boundary \(k\)-derivatives, are smooth in \(x\), and tend rapidly to \(I_3\) at the normalizing end. Moreover,
\[
J_{\pm}(x;k)=\mathcal{A} J_{\pm}(x;\omega k)\mathcal{A} ^{-1}=\mathcal{B} \overline{J_{\pm}(x;\bar{k})}\mathcal{B} ,\quad \mathcal{A} =\left( \begin{matrix}
	0&		0&		1\\
	1&		0&		0\\
	0&		1&		0\\
\end{matrix} \right) ,\quad \mathcal{B} =\left( \begin{matrix}
	0&		1&		0\\
	1&		0&		0\\
	0&		0&		1\\
\end{matrix} \right) .
\]
For compactly supported $u_0$, both Jost solutions extend analytically to $\mathbb C\setminus\{0\}$.

Coefficient comparison gives differentiated expansions on closed subsectors. In particular,
\[
J_{\pm}=I+\frac{J_{\pm}^{(1)}}{k}+O(k^{-2}),\qquad J_{+}^{(1)}=2\int_x^{\infty}{u(y)\,dy\,\mathrm{diag}\mathrm{(}\omega ^2,\omega ,1).}
\]

At \(k=0\), the Jost functions for \(u_0\in\mathcal S(\mathbb R)\) admit, columnwise on their analytic sides,
\[
J_{\pm}(x;k)=\frac{J_{\pm ,0}^{(-1)}(x)}{k}+\sum_{j=0}^m{J_{\pm ,0}^{(j)}(x)k^j}+\mathcal{R} _{0,\pm}^{(m)}(x,k),
\]
where, for \(0\leq n\leq m+1\),
\[
\bigl\| \partial _{k}^{n}\mathcal{R} _{0,\pm}^{(m)}(x,k) \bigr\| \le g_{m,n,\pm}(x)|k|^{m+1-n},
\]
with \(\|Q_1\|g_{m,n,\pm}\in L^1(\mathbb R)\), and likewise for \(J_\pm^A\). The identity \(r^{\mathrm T}e=0\) cancels all higher singular terms, and
\[
J_{\pm ,0}^{(-1)}=a_{\pm}(x)er^{\mathrm{T}},\qquad (J_{\pm}^{A})_{0}^{(-1)}=b_{\pm}(x)re^{\mathrm{T}}.
\]
Thus both Jost systems have at most a simple pole at \(0\). Also, $Q_1J_{\pm,0}^{(-1)}=0$, $\partial_xJ_{\pm,0}^{(-1)}=Q_1J_{\pm,0}^{(0)}$, and normalization gives $J_{\pm,0}^{(-1)}\to0$, $J_{\pm,0}^{(0)}\to I_3$ at the corresponding end.

\subsection{The scattering matrix}
The scattering matrix is defined by
\[
s(k) = I - \int_{\mathbb{R}} \mathrm{e}^{-x\widehat{k\Lambda}} \bigl( Q(x;k) J_{+}(x;k) \bigr) \, dx .
\]
When the initial potential \(u_{0}(x)\) has compact support, \(s(k)\) links the two Jost solutions through
\[
J_{+}(x;k) = J_{-}(x;k) \, \mathrm{e}^{x\widehat{k\Lambda}} \, s(k),\qquad k\in\mathbb{C}\setminus\{0\}.
\]
Its entries have the domains
\[
s(k)\in \left( \begin{matrix}
	\omega ^2\overline{S}&		\mathbb{R} _+&		\omega \mathbb{R} _+\\
	\mathbb{R} _+&		\omega \overline{S}&		\omega ^2\mathbb{R} _+\\
	\omega \mathbb{R} _+&		\omega ^2\mathbb{R} _+&		\overline{S}\\
\end{matrix} \right) \setminus \{0\}.
\]
The diagonal entries are analytic in the indicated interiors and continuous to their boundaries away from \(0\), while the off-diagonal entries are ray-smooth. For every \(N\),
\[
s(k)=I-\sum_{j=1}^N\frac{s_j}{k^j}+O(k^{-N-1}), \qquad k\to\infty,
\]
whereas
\[
s(k)=\frac{c^{(-1)}}k er^{\mathrm T}+s^{(0)}+s^{(1)}k+\cdots,\qquad k\to0,
\]
with \(c^{(-1)}\) determined by \(u_0\). Finally,
\[
s(k)=\mathcal A s(\omega k)\mathcal A^{-1}=\mathcal B\overline{s(\bar k)}\mathcal B .
\]

For \(X^A=(X^{-1})^{\mathrm T}\), the cofactor Jost system satisfies
\[
(J^A)_x+[k\Lambda ,J^A]=-Q^{\mathrm{T}}J^A,\qquad s^A(k)=I_3+\int_{\mathbb{R}}{e^{x\widehat{k\Lambda }}(Q^{\mathrm{T}}J_{+}^{A})(x;k)\,dx.}
\]
The reversed exponential sign gives the complementary column domains
\[
J_{+}^{[1]}, J_{-}^{A,[1]}: D_1\cup D_2,\qquad J_{-}^{[3]}, J_{+}^{A,[3]}: D_6\cup D_1,
\]
where \(X^{[j]}\) denotes the \(j\)th column of \(X\).

Moreover, \(s^A_{11},s^A_{22},s^A_{33}\) are analytic, respectively, in \(D_4\cup D_5,D_2\cup D_3,D_6\cup D_1\).

Define the reflection coefficients by
\[
\begin{cases}
	r_1(k)\doteq \frac{s_{12}(k)}{s_{11}(k)},&		k\in \mathbb{R} _+,\\
	r_2(k)\doteq \frac{s_{12}^{A}(k)}{s_{11}^{A}(k)},&		k\in \mathbb{R} _-.\\
\end{cases}
\]
For a fixed \(M>1\), set
\[
	\begin{aligned}
		K_M&=[M,\infty),& I_M&=[M^{-1},M],& J_M&=(-\infty,-M^{-1}].
	\end{aligned}
\]

\begin{assumption}\label{ass:1.1}
	For the fixed \(M>1\) used below, the following conditions are invoked only in the sectors indicated.
	\begin{enumerate}
		\item[(a)] \emph{Generic scattering data.}
		\(\lim_{k\to0}ks_{11}(k)\ne0\) and
		\(\lim_{k\to0}ks_{11}^{A}(k)\ne0\). The nonzero discrete spectrum is finite and normal-simple. Its representative zeros are exclusive, generate pairwise disjoint complete orbits separated from \(\Sigma\cup\{0\}\), and exhaust the nonzero discrete spectrum. The analytic denominators have no common zeros, and there are no spectral singularities on \(\Sigma\setminus\{0\}\).

		\item[(b)] \emph{Positive-sector regularity.}
		The moving and frozen outer models in Sectors~I-II are uniformly bounded, solvable, and compatible with scalar SK reconstruction.

		\item[(c)] \emph{Negative-sector outer model.}
		The balanced residue system selected by \(J_M\) has a
		uniform determinant gap for \(t\ge t_0\) and \(\xi\in J_M\).
	\end{enumerate}
\end{assumption}

\subsection{A basic RH problem}
For each sector $D_n$ ($n=1,\dots,6$) and \(k\ne0\), define \(M_n(x;k)\) by
\[
(M_n)_{ij}(x;k)=\delta_{ij}+\int_{\gamma_{ij}^n}\bigl(e^{(x-y)\widehat{k\Lambda}}(QM_n)(y;k)\bigr)_{ij}\,dy,\qquad i,j=1,2,3.
\]
Here \(\gamma_{ij}^n\) runs from the normalizing infinity to \(x\), chosen by exponential decay, and \(\gamma_{ii}^n=(+\infty,x)\). One-sided continuation defines boundary values. Analytic Fredholm theory makes \(M_n\) meromorphic in \(D_n\setminus\{0\}\), continuous on the remaining boundary, and \(x\)-smooth. Assume finitely many simple poles \(\mathcal Z\subset\mathbb C\setminus(\Sigma\cup\{0\})\) and no spectral singularities on \(\Sigma\setminus\{0\}\). Then \(\det M_n=1\) and
\begin{equation}\label{2.1}
	M_n(x;k)=\mathcal A M_{n+2}(x;\omega k)\mathcal A^{-1}=\mathcal B\overline{M_{7-n}(x;\bar k)}\mathcal B .
\end{equation}

The Volterra recursion gives differentiated large-\(k\) expansions uniformly on closed subsectors. In particular,
\[
M_n(x;k)=I+\frac{J_+^{(1)}(x)}k+O(k^{-2}).
\]
For compactly supported data,
\[
M_n=J_-e^{x\widehat{k\Lambda }}S_n=J_+e^{x\widehat{k\Lambda }}T_n,\qquad s=S_nT_{n}^{-1},
\]
and on the ray separating \(D_{n-1}\) and \(D_n\),
\[
v_n(x,t;k)=e^{\widehat{k\Lambda x+9k^5\Lambda^2t}}S_{n-1}^{-1}(k)S_n(k).
\]
For Schwartz data, define instead
\[
v_n(k)=e^{-x\widehat{k\Lambda }}M_{n-1,-}^{-1}(x;k)M_{n,+}(x;k),\qquad k\in \partial D_n\setminus \{0\}.
\]
The Lax equation makes this expression \(x\)-independent. Compact-support approximation gives the same space-time jumps and tangential derivatives on \(\Sigma\), without off-ray continuation. Thus \(M=M_n\) in \(D_n\) satisfies \(M_+=M_-v\) on \(\Sigma\setminus(\mathcal Z\cup\{0\})\).

The sectorial Volterra construction, cofactor pairings, and origin expansions give the following denominator structure. Put \(p_i=s_{ii}\) and \(q_i=s^A_{ii}\).
\begin{lemma}\label{lem:2.1}
	Let \(u_0\in\mathcal S(\mathbb R)\). On the analytic sides of the six rays, as \(k\to0\),
	\[
	p_3(k)=\frac ck+p_0+O(k),\qquad q_3(k)=\frac dk+q_0+O(k),
	\]
	where \(c,d,p_0,q_0\in\mathbb R\), and
	\[
	p_1(k)=p_3(\omega k),\quad p_2(k)=p_3(\omega^2k),\qquad q_1(k)=q_3(\omega k),\quad q_2(k)=q_3(\omega^2k).
	\]
	The one-sided remainders and their first \(k\)-derivatives are uniformly \(O(k)\) and \(O(1)\). If Assumption~\ref{ass:1.1}{\rm(a)} holds, then \(cd\ne0\).
\end{lemma}

\begin{proof}
	The rank-one principal parts and \(Q_1J_{\pm,0}^{(-1)}=Q_1^{\mathrm T}(J_\pm^A)_0^{(-1)}=0\) exclude \(k^{-2}\) terms and give the differentiated remainders. The cyclic identities follow from the \(\mathbb Z_3\) reduction, Schwarz symmetry gives reality, and Assumption~\ref{ass:1.1}{\rm(a)} gives \(cd\ne0\).
\end{proof}

Here \(q_i\) is defined by the analytic cofactor scattering problem, and equals the complementary principal minor of \(s\) for compactly supported data. The analytic denominators are \(p_1,p_3,p_2\) in \(D_1\cup D_2\), \(D_3\cup D_4\), \(D_5\cup D_6\), respectively, and \(q_3,q_2,q_1\) in \(D_6\cup D_1\), \(D_2\cup D_3\), \(D_4\cup D_5\).

In \(D_1\),
\begin{equation}\label{2.2}
	M_{1}^{[1]}=J_{+}^{[1]},\qquad p_1M_{1}^{[2]}=-J_{-}^{A,[1]}\times J_{+}^{A,[3]},\qquad q_3M_{1}^{[3]}=J_{-}^{[3]},
\end{equation}
while in \(D_4\),
\begin{equation}\label{2.3}
	q_1M_4^{[1]}=J_-^{[1]},\qquad p_3M_4^{[2]}=J_-^{A,[3]}\times J_+^{A,[1]},\qquad M_4^{[3]}=J_+^{[3]}.
\end{equation}
The rank-one structure cancels the \(k^{-2}\) cross-product terms. Assumption~\ref{ass:1.1}{\rm(a)} gives \(p_1^{-1},q_3^{-1}=O(k)\), so \(M_1^{[2]},M_1^{[3]}=O(1)\), while \(M_1^{[1]}\) may have a simple pole. The other sectors follow by \eqref{2.1}. The endpoint expansions give \(r_1(0)=\omega^2\) and \(r_2(0)=1\).

More explicitly, the nonzero discrete spectrum in Assumption~\ref{ass:1.1}{\rm(a)} consists of the two families of representative exclusive zeros
	\[
	\begin{aligned}
		\mathcal Z_1
		&=\{\kappa_j\in D_1:p_1(\kappa_j)=0,\
		p_1'(\kappa_j)\ne0,\ q_3(\kappa_j)\ne0\}_{j=1}^{N_1},\\
		\mathcal Z_4
		&=\{\eta_\ell\in D_4:q_1(\eta_\ell)=0,\
		q_1'(\eta_\ell)\ne0,\ p_3(\eta_\ell)\ne0\}_{\ell=1}^{N_4}.
	\end{aligned}
	\]
	Their complete reduction orbits are pairwise disjoint, are separated from \(\Sigma\cup\{0\}\), and exhaust the nonzero discrete spectrum. The analytic denominators have no common zeros, and there are no spectral singularities on \(\Sigma\setminus\{0\}\).

The norming constants are defined by
\begin{equation}\label{2.4}
	\begin{aligned}
		\left.p_1(k)M_1^{[2]}(k)\right|_{k=\kappa_j}
		&=-p_1'(\kappa_j)c_j e^{-\theta_{21}(\kappa_j)}
		M_1^{[1]}(\kappa_j),\\
		\left.q_1(k)M_4^{[1]}(k)\right|_{k=\eta_\ell}
		&=q_1'(\eta_\ell)d_\ell e^{\theta_{21}(\eta_\ell)}
		M_4^{[2]}(\eta_\ell).
	\end{aligned}
\end{equation}
The exclusive-zero conditions and the pairing identities imply \(c_j,d_\ell\ne0\). For general Schwartz data, \eqref{2.4} is taken as the definition of the norming constants.

The full discrete spectrum is
\begin{equation}\label{2.5}
	\mathcal Z=\bigcup_{j=1}^{N_1}\{\omega^m\kappa_j,\omega^m\bar\kappa_j:m=0,1,2\}\cup\bigcup_{\ell=1}^{N_4}\{\omega^m\eta_\ell,\omega^m\bar\eta_\ell:m=0,1,2\}.
\end{equation}
The spectral contour, the sector decomposition, and the two types of complete reduction orbits in \eqref{2.5} are illustrated in Figure~\ref{fig:SK-spectrum-orbits}.

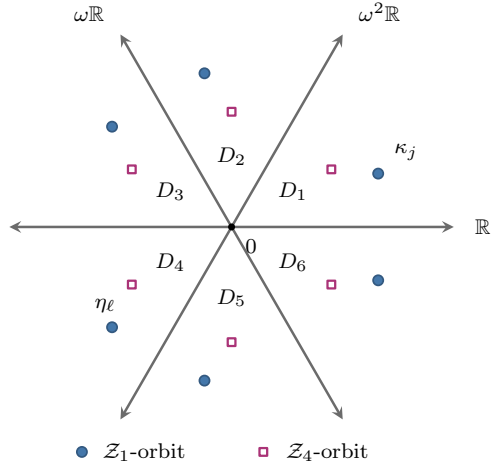
\begin{figure}[htbp]
	\centering
	
	\definecolor{SKSpecBlue}{HTML}{4477AA}
	\definecolor{SKSpecPurple}{HTML}{AA3377}
	\definecolor{SKSpecGray}{HTML}{777777}
	
	\begin{tikzpicture}[
		scale=1.20,
		>=stealth,
		line cap=round,
		line join=round
		]
		\tikzset{
			spectralray/.style={
				SKSpecGray!90!black,
				line width=0.95pt
			},
			orbitone/.style={
				circle,
				fill=SKSpecBlue,
				draw=SKSpecBlue!75!black,
				line width=0.65pt,
				inner sep=1.35pt
			},
			orbittwo/.style={
				rectangle,
				fill=white,
				draw=SKSpecPurple,
				line width=0.85pt,
				inner sep=1.45pt
			},
			sectorlabel/.style={
				font=\scriptsize
			},
			smalllabel/.style={
				font=\scriptsize
			}
		}
		
		\foreach \ang in {0,60,120,180,240,300}{
			\draw[spectralray,->]
			(0,0)--(\ang:2.45);
		}
		
		\node[
		smalllabel,
		right,
		xshift=1.5mm
		]
		at (0:2.45)
		{$\mathbb R$};
		
		\node[
		smalllabel,
		above right,
		xshift=0.8mm,
		yshift=0.5mm
		]
		at (60:2.45)
		{$\omega^2\mathbb R$};
		
		\node[
		smalllabel,
		above left,
		xshift=-0.8mm,
		yshift=0.5mm
		]
		at (120:2.45)
		{$\omega\mathbb R$};
		
		\fill (0,0) circle (1.1pt);
		
		\node[
		smalllabel,
		below right,
		xshift=0.5mm,
		yshift=-0.3mm
		]
		at (0,0)
		{$0$};
		
		\node[sectorlabel] at (30:0.78) {$D_1$};
		\node[sectorlabel] at (90:0.78) {$D_2$};
		\node[sectorlabel] at (150:0.78) {$D_3$};
		\node[sectorlabel] at (210:0.78) {$D_4$};
		\node[sectorlabel] at (270:0.78) {$D_5$};
		\node[sectorlabel] at (330:0.78) {$D_6$};
		
		\foreach \ang in {20,140,260,340,100,220}{
			\node[orbitone]
			at (\ang:1.72) {};
		}
		
		\node[
		smalllabel,
		above right,
		xshift=0.8mm,
		yshift=0.5mm
		]
		at (20:1.72)
		{$\kappa_j$};
		
		\foreach \ang in {210,330,90,150,270,30}{
			\node[orbittwo]
			at (\ang:1.27) {};
		}
		
		\node[
		smalllabel,
		below left,
		xshift=-0.8mm,
		yshift=-0.5mm
		]
		at (210:1.27)
		{$\eta_\ell$};
		
		\node[orbitone]
		at (-1.65,-2.48) {};
		
		\node[smalllabel,right]
		at (-1.53,-2.48)
		{$\mathcal Z_1$-orbit};
		
		\node[orbittwo]
		at (0.35,-2.48) {};
		
		\node[smalllabel,right]
		at (0.49,-2.48)
		{$\mathcal Z_4$-orbit};
		
	\end{tikzpicture}
	
	\caption{The spectral contour
		$\Sigma=\mathbb R\cup\omega\mathbb R\cup\omega^2\mathbb R$,
		the sectors $D_n$, and representative complete six-point orbits
		generated by $\kappa_j\in\mathcal Z_1$ and $\eta_\ell\in\mathcal Z_4$.}
	\label{fig:SK-spectrum-orbits}
\end{figure}
If $\operatorname*{Res}_{k=\lambda}M(k)=\lim_{k\to\lambda}M(k)C_\lambda$, the reductions give
\begin{equation}\label{2.6}
	C_{\omega ^m\lambda}=\omega ^m\mathcal{A} ^{-m}C_{\lambda}\mathcal{A} ^m,\qquad C_{\bar{\lambda}}=\mathcal{B} \overline{C_{\lambda}}\mathcal{B} ,\qquad m=0,1,2.
\end{equation}
Thus it suffices to prescribe
\[
C_{\kappa _j}=-c_je^{-\theta _{21}(\kappa _j)}E_{12},\qquad C_{\eta _{\ell}}=d_{\ell}e^{\theta _{21}(\eta _{\ell})}E_{21},
\]
and all remaining residue matrices are generated by \eqref{2.6}.

For the perturbative stability statement below, use \(\|f\|_{L^{1,1}}=\int_{\mathbb R}(1+|x|)|f(x)|\,dx\). Let \(\mathfrak R(\Sigma)\) denote the six open half-rays of \(\Sigma\).
For \(a>0\), let \(\mathcal Y_a\) be the space of ray-wise \(C^1\) reflection data \(r=(r_\Gamma)_{\Gamma\in\mathfrak R(\Sigma)}\) with finite norm
\begin{equation}\label{2.7}
	\|r\|_{\mathcal Y_a}
	=\sum_{\Gamma\in\mathfrak R(\Sigma)}
	\sup_{\substack{k\in\Gamma\\ |k|\ge a}}
	(1+|k|)\bigl(|r_\Gamma(k)|
	+|\partial_\tau r_\Gamma(k)|\bigr).
\end{equation}
\begin{lemma}\label{lem:2.2}
Let \(u_*\in\mathcal S(\mathbb R)\) be real-valued and satisfy the nonzero-spectral conditions in Assumption~\ref{ass:1.1}{\rm(a)}. Choose \(r_0>0\) so that \(|k|\le 2r_0\) contains no representative
zero of \(u_*\). For real \(u_0\in\mathcal S(\mathbb R)\), set
\[
\varepsilon=\|u_0-u_*\|_{L^{1,1}}.
\]
There exists \(\varepsilon_0=\varepsilon_0(u_*,r_0)>0\) such that, if \(\varepsilon\le\varepsilon_0\), each representative zero of \(u_*\) has a unique normal-simple continuation for \(u_0\), and the
corresponding complete orbits remain exclusive and pairwise disjoint. Moreover,
\[
\max_j \bigl( |\kappa _j(u_0)-\kappa _j(u_*)|+|c_j(u_0)-c_j(u_*)| \bigr) +\max_{\ell} \bigl( |\eta _{\ell}(u_0)-\eta _{\ell}(u_*)|+|d_{\ell}(u_0)-d_{\ell}(u_*)| \bigr) \le C\varepsilon ,
\]
where empty maxima are omitted. No additional zero or spectral singularity occurs for \(|k|\ge r_0\). Furthermore, for every \(a\ge r_0\),
\[
\|r[u_0]-r[u_*]\|_{\mathcal Y_a}\le C_a\varepsilon.
\]
In particular, if \(u_*\) is reflectionless, then
\begin{equation}\label{2.8}
\|r[u_0]\|_{\mathcal Y_a}\le C_a\varepsilon, \qquad a\ge r_0.
\end{equation}
\end{lemma}

\begin{proof}
The Volterra equations for the Jost and cofactor solutions, together with the resolvent identity, imply that for \(|k|\ge r_0\) the relevant Jost columns and their first \(k\)-derivatives depend locally
Lipschitz continuously on the potential in \(L^{1,1}\). Hence, on their respective analytic domains,
\[
|f[u_0](k)-f[u_*](k)|+|\partial _kf[u_0](k)-\partial _kf[u_*](k)|\le \frac{C_{r_0}\varepsilon}{1+|k|},\qquad f\in \{p_i,q_i,s_{12},s_{12}^{A}\}.
\]

By normal simplicity and exclusivity of the reference zeros, Rouch\'e's theorem gives exactly one zero of the perturbed denominator near each reference representative zero. The preceding Lipschitz
bound yields
\[
|\kappa_j(u_0)-\kappa_j(u_*)|+|\eta_\ell(u_0)-\eta_\ell(u_*)|\le C\varepsilon.
\]
The residue formulas~\eqref{2.4}, together with the same local bounds and Cauchy estimates, imply
\[
|c_j(u_0)-c_j(u_*)|+|d_\ell(u_0)-d_\ell(u_*)|\le C\varepsilon.
\]
The reduction symmetries then preserve the corresponding complete exclusive orbits.

Finally, outside the zero disks, the analytic denominators of the reference problem are uniformly bounded away from zero on the compact region
\[
\{k:r_0\le |k|\le R\},
\]
while their large-\(k\) asymptotics give the same conclusion for \(|k|\ge R\). Thus no additional zero or spectral singularity occurs for \(|k|\ge r_0\). Using
\[
r_1=\frac{s_{12}}{p_1}, \qquad r_2=\frac{s_{12}^{A}}{q_1},
\]
and their reduction images gives
\[
\parallel r[u_0]-r[u_*]\parallel _{\mathcal{Y} _a}\le C_a\varepsilon ,\qquad a\ge r_0.
\]
If \(u_*\) is reflectionless, this reduces to~\eqref{2.8}.
\end{proof}

The mSK data are the tilded analogues above, with the same simplicity, exclusivity, orbit, and residue conditions. Expressing \(S_{n-1}^{-1}S_n\) through \(s,s^A\) and applying the space-time conjugation gives the six jumps
\begin{equation}\label{2.9}
\begin{aligned}
	v_1&=\left( \begin{matrix}
		1&		-r_1(k)\,\mathrm{e}^{-\theta _{21}}&		0\\
		\overline{r}_1(k)\,\mathrm{e}^{\theta _{21}}&		1-|r_1(k)|^2&		0\\
		0&		0&		1\\
	\end{matrix} \right) ,v_2=\left( \begin{matrix}
		1&		0&		0\\
		0&		1-|r_2(\omega k)|^2&		-\overline{r}_2(\omega k)\,\mathrm{e}^{-\theta _{32}}\\
		0&		r_2(\omega k)\,\mathrm{e}^{\theta _{32}}&		1\\
	\end{matrix} \right) ,\\
	v_3&=\left( \begin{matrix}
		1-|r_1(\omega ^2k)|^2&		0&		\overline{r}_1(\omega ^2k)\,\mathrm{e}^{-\theta _{31}}\\
		0&		1&		0\\
		-r_1(\omega ^2k)\,\mathrm{e}^{\theta _{31}}&		0&		1\\
	\end{matrix} \right) ,v_4=\left( \begin{matrix}
		1-|r_2(k)|^2&		-\overline{r}_2(k)\,\mathrm{e}^{-\theta _{21}}&		0\\
		r_2(k)\,\mathrm{e}^{\theta _{21}}&		1&		0\\
		0&		0&		1\\
	\end{matrix} \right) ,\\
	v_5&=\left( \begin{matrix}
		1&		0&		0\\
		0&		1&		-r_1(\omega k)\,\mathrm{e}^{-\theta _{32}}\\
		0&		\overline{r}_1(\omega k)\,\mathrm{e}^{\theta _{32}}&		1-|r_1(\omega k)|^2\\
	\end{matrix} \right) ,v_6=\left( \begin{matrix}
		1&		0&		r_2(\omega ^2k)\,\mathrm{e}^{-\theta _{31}}\\
		0&		1&		0\\
		-\overline{r}_2(\omega ^2k)\,\mathrm{e}^{\theta _{31}}&		0&		1-|r_2(\omega ^2k)|^2\\
	\end{matrix} \right) ,\\
\end{aligned}
\end{equation}
where $\theta_{ij}(x,t;k)=(l_i-l_j)x+(z_i-z_j)t$ $(1\leq i,j\leq3)$ with $(l_1,l_2,l_3)=(\omega k,\omega^2k,k)$ and $(z_1,z_2,z_3)=(9\omega^2k^5,9\omega k^5,9k^5)$.

The SK inverse problem is the following.
\begin{rhp}\label{rhp1}
	Find a matrix-valued function \(M(x,t;k)=M_n(x,t;k)\) for \(k\in D_n\) (\(n=1,\dots,6\)) with the following properties:
	\begin{enumerate}
		\item[(i)]
		\(M(x,t;k)\) is meromorphic in $\mathbb C\setminus(\Sigma\cup\{0\})$, with simple poles precisely at the points of \(\mathcal Z\) in \eqref{2.5}.
		
		\item[(ii)]
		The boundary values satisfy
		\[
		M_+(x,t;k) = M_-(x,t;k)\,v(x,t;k),\qquad k\in\Sigma \setminus \{0\},
		\]
		where the jump matrices \(v_n(x,t;k)\) are given by \eqref{2.9}.
		
		\item[(iii)]
		The function satisfies
		\[
		M(x,t;k)=\mathcal{A}M(x,t;\omega k)\mathcal{A}^{-1}=\mathcal{B}\,\overline{M(x,t;\bar k)}\,\mathcal{B}.
		\]
		
		\item[(iv)]
		As \(k\to\infty\),
		\[
		M(x,t;k)=I_3+\frac{M_1(x,t)}{k}+\mathcal O(k^{-2}),\qquad M_1(x,t)=\alpha(x,t)\operatorname{diag}(\omega^2,\omega,1),
		\]
		where \(\alpha(x,t)\in\mathbb R\).
		
		\item[(v)]
		For every fixed integer \(p\ge0\), as \(k\to0\) in \(D_n\),
		\[
		M_n(x,t;k)=
		\sum_{l=-1}^pM_{0,n}^{(l)}(x,t)k^l
		+O(k^{p+1}).
		\]
		In \(D_1\),
		\[
		M_{0,1}^{(-1)}(x,t)=a_+(x,t)
		\begin{pmatrix}
			\omega ^2&0&0\\
			\omega ^2&0&0\\
			\omega ^2&0&0
		\end{pmatrix},
		\]
		where \(a_+\) is real valued and decays rapidly as \(x\to+\infty\), the remaining sectorial coefficients are generated by
		\[
		M_{0,n}^{(-1)}=\omega^2\mathcal A M_{0,n+2}^{(-1)}\mathcal A^{-1}=\mathcal B\overline{M_{0,7-n}^{(-1)}}\mathcal B.
		\]
		
		\item[(vi)] At the two independent representative sets,
		\[
		\underset{k=\kappa _j}{\mathrm{Res}}M(k)=\lim_{k\rightarrow \kappa _j} M(k)\left( \begin{matrix}
			0&		-c_je^{-\theta _{21}(\kappa _j)}&		0\\
			0&		0&		0\\
			0&		0&		0\\
		\end{matrix} \right),\qquad j=1,\ldots,N_1,
		\]
		
		\[
		\underset{k=\eta _\ell}{\mathrm{Res}}M(k)=\lim_{k\rightarrow \eta _\ell} M(k)\left( \begin{matrix}
			0&		0&		0\\
			d_\ell e^{\theta _{21}(\eta _\ell)}&		0&		0\\
			0&		0&		0\\
		\end{matrix} \right),\qquad \ell=1,\ldots,N_4,
		\]
		and all remaining residue conditions are generated by \eqref{2.6}.
	\end{enumerate}
\end{rhp}
Set
\[
\ell^{\mathrm T}=(\omega,\omega^2,1),\qquad e_3=(0,0,1)^{\mathrm T},\qquad N=\ell^{\mathrm T}M .
\]
The diagonal condition in item~{\rm(4)} of RH Problem~\ref{rhp1} defines the reconstructed potential by
\begin{equation}\label{2.10}
	u(x,t)=-\frac{1}{2}\partial _x\lim_{k\rightarrow \infty} k\bigl( N_3(x,t;k)-1 \bigr) =-\frac{1}{2}\partial _x\bigl( \ell ^{\mathrm{T}}M_1(x,t)e_3 \bigr) =-\frac{1}{2}\partial _x\alpha (x,t).
\end{equation}
This row formula survives all inverse transformations. Since \(\ell^{\mathrm T}M_{0,1}^{(-1)}=a_+(\omega+\omega^2+1)(\omega^2,0,0)=0\), the reductions make \(N\) bounded at \(0\) in every sector.

\begin{prop}
	For each fixed \((x,t)\) and fixed scattering data, RH Problem~\ref{rhp1} has at most one solution, and every solution satisfies \(\det M\equiv1\). For data generated by the global solution \(u\) in Theorem~\ref{th1}, the problem is solvable and \eqref{2.10} reconstructs \(u\).
\end{prop}

\begin{proof}
	Every residue factor is rank-one nilpotent, and locally
	\[
	M=H_\lambda\left(I_3+\frac{C_\lambda}{k-\lambda}\right),
	\qquad C_\lambda^2=0,
	\]
	with \(H_\lambda\) holomorphic. Thus \(\det M\) has neither jumps nor nonzero poles. Its possible simple principal part at \(0\) is excluded by \(\det M(k)=\det M(\omega k)\), so Liouville's theorem gives \(\det M\equiv1\).
	
	For two solutions, the same factorization and the origin form give \(Q=M\widehat M^{-1}=I_3+P/k\). At infinity, \(P=(\alpha-\widehat\alpha)\operatorname{diag}(\omega^2,\omega,1)\), while \(\det Q=1\) forces \(\alpha=\widehat\alpha\). Hence uniqueness. The sectorial Fredholm solution for \(u(\cdot,t)\) has the stated data and \(\alpha=2\int_x^\infty u(y,t)\,dy\), so \eqref{2.10} reconstructs \(u\).
\end{proof}

For mSK the Jost solutions are regular at \(0\). Replacing \(r_j\) by their tilded analogues gives the regular RH problem below.

\begin{rhp}\label{rhp2}
Find $m(x,t;k)$ meromorphic in $\mathbb C\setminus\Sigma$, with simple poles at $\widetilde{\mathcal Z}$, such that
\[
m_+=m_-\widetilde v,\qquad
m(k)=\mathcal A m(\omega k)\mathcal A^{-1}
=\mathcal B\overline{m(\bar k)}\mathcal B.
\]
The jumps $\widetilde v_n$ are obtained from \eqref{2.9} by replacing $r_j$ with $\widetilde r_j$. The representative residue matrices are
\[
\widetilde C_{\widetilde\kappa_j}
=-\widetilde c_j e^{-\theta_{21}(\widetilde\kappa_j)}E_{12},\qquad
\widetilde C_{\widetilde\eta_\ell}
=\widetilde d_\ell e^{\theta_{21}(\widetilde\eta_\ell)}E_{21},
\]
and $\operatorname*{Res}_{k=\lambda}m=\lim_{k\to\lambda}m\widetilde C_\lambda$, with the remaining residues generated by \eqref{2.6}. At infinity, $m=I_3+m_1/k+O(k^{-2})$, where
\[
m_1=\frac w3\begin{pmatrix}0&\omega&1\\\omega^2&0&1\\\omega^2&\omega&0\end{pmatrix}
+\frac13\int_\infty^x w(y,t)^2\,dy\,
\operatorname{diag}(\omega^2,\omega,1).
\]
At zero, $m_n$ has an expansion $m_n=\sum_{l=0}^p m_{0,n}^{(l)}k^l+O(k^{p+1})$ for every $p\ge0$, with
\[
m_{0,n}^{(l)}=\omega^l\mathcal A m_{0,n+2}^{(l)}\mathcal A^{-1}
=\mathcal B\overline{m_{0,7-n}^{(l)}}\mathcal B.
\]
\end{rhp}
The mSK potential is recovered by
\[
w(x,t) = 3\lim_{k\to\infty} k\,\bigl(m(x,t;k)\bigr)_{13}.
\]

\section{Long-time asymptotics in Sectors I and II}\label{sec3}
Fix $M>1$, and let $K_M$, $I_M$, and $J_M$ be the velocity sets defined in Section~\ref{sec2}. The following theorem describes the SK asymptotics in the four regions illustrated in Figure~\ref{fig:SK-asymptotic-regions}. Its proof is given in this section and Sections~\ref{sec4}-\ref{sec5}, together with the moving-to-frozen comparison in Section~\ref{sec6}.

\begin{theorem}[Long-time asymptotics of the SK equation]\label{th1}
Let \(u\) be a global classical real-valued solution of \eqref{1.1} with \(u_0\in\mathcal S(\mathbb R)\), and suppose that Assumption~\ref{ass:1.1} holds. Then, as \(t\to+\infty\),
	\begin{enumerate}
		\item[(i)] \emph{Sector I.}
		Uniformly for \(\xi=x/t\ge M\),
		\begin{equation}\label{3.1}
			u(x,t)=u_{\rm sol}(x,t;\widehat\sigma_M)+t^{-1/2}h_M+O(t^{-3/4}),
		\end{equation}
		where \(u_{\rm sol}(x,t;\widehat\sigma_M)\), \(\widehat\sigma_M\), and
		\(h_M\) are defined in \eqref{3.12}, \eqref{3.14}, and
		\eqref{3.35}, respectively.
		
		\item[(ii)] \emph{Sector II.}
		Uniformly for \(\xi\in I_M\),
		\[
		u(x,t)=u_{\rm sol}(x,t;\widehat\sigma(I_M))+t^{-1/2}h+O(t^{-3/4}),
		\]
		where \(u_{\rm sol}(x,t;\widehat\sigma(I_M))\) is defined by
		\eqref{3.12}, and \(h\) is given by \eqref{3.34}.
		
		\item[(iii)] \emph{Sector III.}
		Let \(u=(w_x-w^2)/6\), where \(w\) is a global mSK solution with Schwartz data satisfying Assumption~\ref{ass:1.2}{\rm(a)}, \eqref{4.8}, and \(\kappa_M<1\) in \eqref{4.14}. Uniformly for \(s=x/(5t)^{1/5}\in[-M,M]\),
		\[
		u(x,t)=\frac16(5t)^{-2/5}\{p'(s)-p(s)^2\}+O((5t)^{-3/5}),
		\]
where \(p\) is the real profile defined by \eqref{4.23} and satisfying
the F-XVIII equation \eqref{4.25}.
		
		\item[(iv)] \emph{Sector IV.}
		For every \(N\ge1\), there exists \(T_N>0\) such that uniformly for
		\(t\ge T_N\) and \(\xi\in J_M\),
		\[
		u(x,t)=u_J^{\rm out}(x,t)+O\bigl((|x|+t)^{-N}\bigr).
		\]
		Here \(u_J^{\rm out}\) is the rational outer solution on \(J_M\) defined by \eqref{5.15}.
	\end{enumerate}
\end{theorem}

\begin{re}
	The asymptotic formulas above do not provide uniform connection formulas in the intermediate gaps
	\[
	-\frac{t}{M}<x<-M(5t)^{1/5},\qquad M(5t)^{1/5}<x<\frac{t}{M}.
	\]
\end{re}

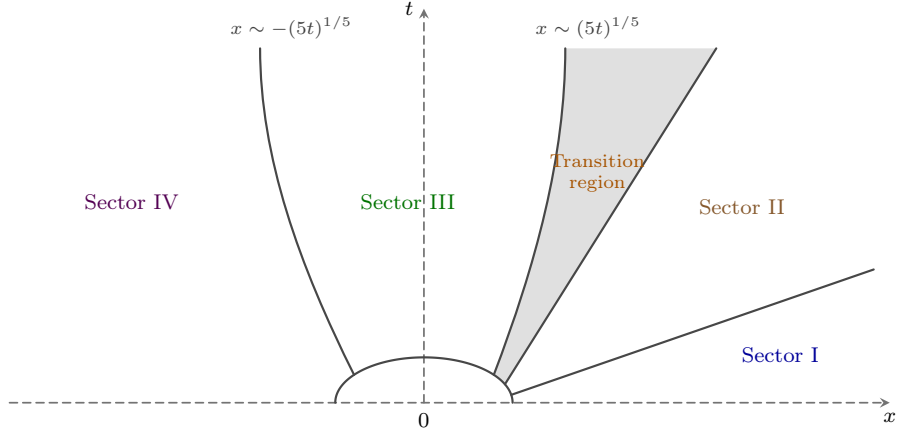
\begin{figure}[htbp]
	\centering
	\begin{tikzpicture}[
		x=0.85cm,
		y=0.75cm,
		>=stealth,
		every node/.style={font=\scriptsize},
		line cap=round,
		line join=round
		]
		
		\def\a{1.38}
		\def\b{0.80}
		
		\coordinate (Lcap) at (-\a,0);
		\coordinate (Rcap) at ( \a,0);
		
		\coordinate (Lcurve) at ({\a*cos(142)},{\b*sin(142)});
		\coordinate (Rcurve) at ({\a*cos(38)},{\b*sin(38)});
		
		\coordinate (Pvone) at ({\a*cos(24)},{\b*sin(24)});
		\coordinate (Pvtwo) at ({\a*cos(10)},{\b*sin(10)});
		
		\coordinate (Ltop) at (-2.55,6.25);
		\coordinate (Rtop) at ( 2.20,6.25);
		\coordinate (Ttop) at (4.00,6.25);
		\coordinate (VoneT) at (4.55,6.25);
		\coordinate (VtwoT) at (7.00,2.35);
		
		\fill[black!12]
		(Rtop)
		-- (VoneT)
		-- (Pvone)
		plot[
		domain=24:38,
		samples=50,
		smooth,
		variable=\ang
		]
		({\a*cos(\ang)},{\b*sin(\ang)})
		.. controls
		(
		{\a*cos(38)+0.58*(2.20-\a*cos(38))},
		{\b*sin(38)+(1/3)*(6.25-\b*sin(38))}
		)
		and
		(
		{2.20},
		{6.25-(1/3)*(6.25-\b*sin(38))}
		)
		.. (Rtop)
		-- cycle;
		
		\draw[->,densely dashed,black!55,line width=0.65pt]
		(-6.45,0)--(7.25,0)
		node[below,text=black] {$x$};
		
		\draw[->,densely dashed,black!55,line width=0.65pt]
		(0,0)--(0,6.95)
		node[left,text=black] {$t$};
		
		\draw[black!72,line width=0.82pt]
		plot[
		domain=180:0,
		samples=140,
		smooth,
		variable=\ang
		]
		({\a*cos(\ang)},{\b*sin(\ang)});
		
		\draw[black!72,line width=0.82pt]
		(Lcurve)
		.. controls
		(
		{\a*cos(142)+0.58*(-2.55-\a*cos(142))},
		{\b*sin(142)+(1/3)*(6.25-\b*sin(142))}
		)
		and
		(
		{-2.55},
		{6.25-(1/3)*(6.25-\b*sin(142))}
		)
		.. (Ltop);
		
		\draw[black!72,line width=0.82pt]
		(Rcurve)
		.. controls
		(
		{\a*cos(38)+0.58*(2.20-\a*cos(38))},
		{\b*sin(38)+(1/3)*(6.25-\b*sin(38))}
		)
		and
		(
		{2.20},
		{6.25-(1/3)*(6.25-\b*sin(38))}
		)
		.. (Rtop);
		
		\draw[black!72,line width=0.82pt]
		(Pvone)--(VoneT);
		
		\draw[black!72,line width=0.82pt]
		(Pvtwo)--(VtwoT);
		
		\node[
		above right,
		font=\tiny,
		text=black!75,
		xshift=-15pt,
		yshift=0pt
		] at (Ltop)
		{$x\sim -(5t)^{1/5}$};
		
		\node[
		above right,
		font=\tiny,
		text=black!75,
		xshift=-15pt,
		yshift=0pt
		] at (Rtop)
		{$x\sim (5t)^{1/5}$};
		
		\node[align=center,text=violet!65!black] at (-4.55,3.55)
		{Sector IV};
		
		\node[align=center,text=green!45!black] at (-0.25,3.55)
		{Sector III};
		
		\node[
		align=center,
		font=\tiny,
		text=orange!65!black
		] at (2.70,4.05)
		{Transition\\region};
		
		\node[align=center,text=brown!70!black] at (4.95,3.45)
		{Sector II};
		
		\node[align=center,text=blue!60!black] at (5.55,0.85)
		{Sector I};
		
		\node[below,text=black] at (0,0) {$0$};
		
	\end{tikzpicture}
	
	\caption{The asymptotic sectors I-IV for the SK equation in the $(x,t)$-half plane.}
	\label{fig:SK-asymptotic-regions}
\end{figure}

Work first with $\xi=x/t\in K\Subset(0,\infty)$. Then $0<\kappa _0\le k_0=(\xi/45)^{1/4}\le\kappa _1$, so $\mathcal K_0=\{\pm\omega^mk_0:m=0,1,2\}$ stays away from \(0\). The cofactor identities and absence of spectral singularities give
\[
1-|r_1(k)|^2=\frac{q_3(k)}{p_1(k)p_2(k)}>0\quad(k>0),\qquad 1-|r_2(k)|^2=\frac{p_3(k)}{q_1(k)q_2(k)}>0\quad(k<0).
\]
We write
\[
	\theta _{ij}(x,t;k)=t\bigl[ (l_i-l_j)\xi +(z_i-z_j) \bigr] \equiv t\Phi _{ij}(\xi ;k).
\]
For $k=a+ib$,
\[
\operatorname{Re}\Phi_{21}(\xi;k)=\sqrt3\,b\bigl(\xi-45a^4+90a^2b^2-9b^4\bigr).
\]

By the two symmetries it suffices to treat the real-axis sector. Compact positive-velocity uniformity is later extended by the variable \(\zeta=t/x\) and overlap.

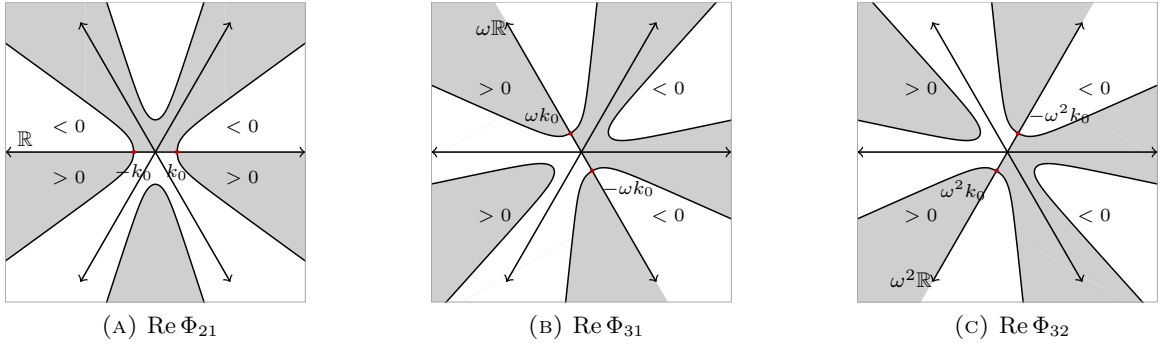
\begin{figure}[htbp]
	\centering
	
	\def\SKK{0.46}
	\def\SKB{4.60}
	\def\SKX{1.494}
	\def\SKV{3.20}
	
	\newcommand{\FillPhiTwentyOnePositive}[1]{%
		\fill[#1]
		(-\SKX,{sqrt(5*(\SKX)^2+sqrt(20*(\SKX)^4+5*(\SKK)^4))})
		-- plot[smooth,domain=-\SKX:\SKX,samples=140]
		(\x,{sqrt(5*(\x)^2+sqrt(20*(\x)^4+5*(\SKK)^4))})
		-- plot[smooth,domain=\SKX:\SKK,samples=90]
		(\x,{sqrt(5*(\x)^2-sqrt(20*(\x)^4+5*(\SKK)^4))})
		-- (-\SKK,0)
		-- plot[smooth,domain=-\SKK:-\SKX,samples=90]
		(\x,{sqrt(5*(\x)^2-sqrt(20*(\x)^4+5*(\SKK)^4))})
		-- cycle;
		
		\fill[#1]
		(-\SKB,\SKB)--(-\SKX,\SKB)
		-- (-\SKX,{sqrt(5*(\SKX)^2-sqrt(20*(\SKX)^4+5*(\SKK)^4))})
		-- plot[smooth,domain=-\SKX:-\SKB,samples=120]
		(\x,{sqrt(5*(\x)^2-sqrt(20*(\x)^4+5*(\SKK)^4))})
		-- cycle;
		
		\fill[#1]
		(\SKX,\SKB)--(\SKB,\SKB)
		-- (\SKB,{sqrt(5*(\SKB)^2-sqrt(20*(\SKB)^4+5*(\SKK)^4))})
		-- plot[smooth,domain=\SKB:\SKX,samples=120]
		(\x,{sqrt(5*(\x)^2-sqrt(20*(\x)^4+5*(\SKK)^4))})
		-- cycle;
		
		\fill[#1]
		(-\SKX,-\SKB)--(\SKX,-\SKB)
		-- plot[smooth,domain=\SKX:-\SKX,samples=140]
		(\x,{-sqrt(5*(\x)^2+sqrt(20*(\x)^4+5*(\SKK)^4))})
		-- cycle;
		
		\fill[#1]
		(\SKK,0)--(\SKB,0)
		-- (\SKB,{-sqrt(5*(\SKB)^2-sqrt(20*(\SKB)^4+5*(\SKK)^4))})
		-- plot[smooth,domain=\SKB:\SKK,samples=120]
		(\x,{-sqrt(5*(\x)^2-sqrt(20*(\x)^4+5*(\SKK)^4))})
		-- cycle;
		
		\fill[#1]
		(-\SKK,0)--(-\SKB,0)
		-- (-\SKB,{-sqrt(5*(\SKB)^2-sqrt(20*(\SKB)^4+5*(\SKK)^4))})
		-- plot[smooth,domain=-\SKB:-\SKK,samples=120]
		(\x,{-sqrt(5*(\x)^2-sqrt(20*(\x)^4+5*(\SKK)^4))})
		-- cycle;
	}
	
	\newcommand{\DrawPhiTwentyOneZeroCurves}{%
		\draw[black,line width=0.56pt,smooth,
		domain=-\SKX:\SKX,samples=150]
		plot (\x,{sqrt(5*(\x)^2+sqrt(20*(\x)^4+5*(\SKK)^4))});
		
		\draw[black,line width=0.56pt,smooth,
		domain=-\SKX:\SKX,samples=150]
		plot (\x,{-sqrt(5*(\x)^2+sqrt(20*(\x)^4+5*(\SKK)^4))});
		
		\draw[black,line width=0.56pt,smooth,
		domain=\SKK:\SKB,samples=160]
		plot (\x,{sqrt(5*(\x)^2-sqrt(20*(\x)^4+5*(\SKK)^4))});
		
		\draw[black,line width=0.56pt,smooth,
		domain=\SKK:\SKB,samples=160]
		plot (\x,{-sqrt(5*(\x)^2-sqrt(20*(\x)^4+5*(\SKK)^4))});
		
		\draw[black,line width=0.56pt,smooth,
		domain=-\SKB:-\SKK,samples=160]
		plot (\x,{sqrt(5*(\x)^2-sqrt(20*(\x)^4+5*(\SKK)^4))});
		
		\draw[black,line width=0.56pt,smooth,
		domain=-\SKB:-\SKK,samples=160]
		plot (\x,{-sqrt(5*(\x)^2-sqrt(20*(\x)^4+5*(\SKK)^4))});
	}
	
	\newcommand{\DrawSKContour}{%
		\draw[<->,black,line width=0.58pt]
		(-\SKV,0)--(\SKV,0);
		\draw[<->,black,line width=0.58pt]
		(240:\SKV)--(60:\SKV);
		\draw[<->,black,line width=0.58pt]
		(300:\SKV)--(120:\SKV);
	}
	
	\tikzset{
		contourlabel/.style={
			font=\scriptsize,
			inner sep=0.8pt
		}
	}
	
	\subfloat[$\operatorname{Re}\Phi_{21}$]{%
		\begin{minipage}[t]{0.32\textwidth}
		\centering
		\begin{tikzpicture}[scale=0.62,line cap=round,line join=round]
			\clip (-\SKV,-\SKV) rectangle (\SKV,\SKV);
			
			\FillPhiTwentyOnePositive{black!19}
			\DrawPhiTwentyOneZeroCurves
			\DrawSKContour
			
			\draw[black!35,line width=0.30pt]
			(-\SKV,-\SKV) rectangle (\SKV,\SKV);
			
			\fill[red!80!black] (\SKK,0) circle (1.35pt);
			\fill[red!80!black] (-\SKK,0) circle (1.35pt);
			
			\node[font=\tiny,below] at (\SKK,-0.05) {$k_0$};
			\node[font=\tiny,below] at (-\SKK,-0.05) {$-k_0$};
			
			\node[contourlabel,anchor=south west]
			at (-3.00,0.08) {$\mathbb R$};
			
			\node[font=\tiny] at (-1.85, 0.55) {$<0$};
			\node[font=\tiny] at ( 1.85, 0.55) {$<0$};
			\node[font=\tiny] at (-1.85,-0.55) {$>0$};
			\node[font=\tiny] at ( 1.85,-0.55) {$>0$};
		\end{tikzpicture}
\end{minipage}
}
	\hfill
	\subfloat[$\operatorname{Re}\Phi_{31}$]{%
	\begin{minipage}[t]{0.32\textwidth}
		\centering
		\begin{tikzpicture}[scale=0.62,line cap=round,line join=round]
			\clip (-\SKV,-\SKV) rectangle (\SKV,\SKV);
			
			\fill[black!19]
			(-\SKV,-\SKV) rectangle (\SKV,\SKV);
			
			\begin{scope}[rotate=120]
				\FillPhiTwentyOnePositive{white}
				\DrawPhiTwentyOneZeroCurves
			\end{scope}
			
			\DrawSKContour
			
			\draw[black!35,line width=0.30pt]
			(-\SKV,-\SKV) rectangle (\SKV,\SKV);
			
			\fill[red!80!black] (120:\SKK) circle (1.35pt);
			\fill[red!80!black] (300:\SKK) circle (1.35pt);
			
			\node[font=\tiny,above left]
			at (120:\SKK) {$\omega k_0$};
			\node[font=\tiny,below right]
			at (300:\SKK) {$-\omega k_0$};
			
			\node[contourlabel,anchor=south west]
			at (-2.30,2.45) {$\omega\mathbb R$};
			
			\node[font=\tiny] at (-1.85, 1.35) {$>0$};
			\node[font=\tiny] at ( 1.85, 1.35) {$<0$};
			\node[font=\tiny] at (-1.85,-1.35) {$>0$};
			\node[font=\tiny] at ( 1.85,-1.35) {$<0$};
		\end{tikzpicture}
\end{minipage}
}
	\hfill
	\subfloat[$\operatorname{Re}\Phi_{32}$]{%
	\begin{minipage}[t]{0.32\textwidth}
		\centering
		\begin{tikzpicture}[scale=0.62,line cap=round,line join=round]
			\clip (-\SKV,-\SKV) rectangle (\SKV,\SKV);
			
			\begin{scope}[rotate=240]
				\FillPhiTwentyOnePositive{black!19}
				\DrawPhiTwentyOneZeroCurves
			\end{scope}
			
			\DrawSKContour
			
			\draw[black!35,line width=0.30pt]
			(-\SKV,-\SKV) rectangle (\SKV,\SKV);
			
			\fill[red!80!black] (60:\SKK) circle (1.35pt);
			\fill[red!80!black] (240:\SKK) circle (1.35pt);
			
			\node[font=\tiny,above right]
			at (60:\SKK) {$-\omega^2 k_0$};
			\node[font=\tiny,below left]
			at (240:\SKK) {$\omega^2 k_0$};
			
			\node[contourlabel,anchor=north west]
			at (-2.55,-2.45) {$\omega^2\mathbb R$};
			
			\node[font=\tiny] at (-1.85, 1.35) {$>0$};
			\node[font=\tiny] at ( 1.85, 1.35) {$<0$};
			\node[font=\tiny] at (-1.85,-1.35) {$>0$};
			\node[font=\tiny] at ( 1.85,-1.35) {$<0$};
		\end{tikzpicture}
\end{minipage}
}
	
	\caption{Signature tables of the SK phase functions in the positive-velocity region. Gray and white regions correspond to $\operatorname{Re}\Phi_{ij}>0$ and
		$\operatorname{Re}\Phi_{ij}<0$, respectively. The marked points are the saddles $\pm k_0$, $\pm\omega k_0$, and $\pm\omega^2k_0$.}
	\label{fig:SK-phase-signatures}
\end{figure}

\subsection{Deformation of the basic RH problem}
Absorb the diagonal factors $1-|r_j(k)|^2$ by defining \(\delta_1,\delta_4\) through
\[
\begin{aligned}
	\delta_{1+}(k) &= \delta_{1-}(k)\bigl(1-|r_1(k)|^2\bigr), \quad k\in(k_0,\infty),  \\
	\delta_{4+}(k) &= \delta_{4-}(k)\bigl(1-|r_2(k)|^2\bigr), \quad k\in(-\infty,-k_0),
\end{aligned}
\]
with $\delta_j(k)\to1$ as $k\to\infty$. By the Plemelj formula,
\[
\begin{aligned}
\delta_1(k) &= \exp\Bigl\{\frac{1}{2\pi i}\int_{k_0}^\infty \frac{\ln(1-|r_1(s)|^2)}{s-k}\,ds\Bigr\}, \\
\delta _4(k)&=\exp \Bigl\{ -\frac{1}{2\pi i}\int_{-\infty}^{-k_0}{\frac{\ln(1-|r_2(s)|^2)}{s-k}\,ds} \Bigr\}.
\end{aligned}
\]
The second sign comes from the outward orientation. Let $\log_\theta k$ have branch cut \(\arg k=\theta\). Near \(\pm k_0\),
\[
\delta_1(k)=e^{-i\nu_1\log_0(k-k_0)}e^{-\chi_1(k)},\qquad \delta_4(k)=e^{-i\nu_4\log_\pi(k+k_0)}e^{-\chi_4(k)},
\]
where
\[
\begin{aligned}
	\nu_1&=-\frac1{2\pi}\ln(1-|r_1(k_0)|^2),&
	\chi_1(\xi;k)
	&=\frac{1}{2\pi i}\int_{k_0}^{\infty}
	\log_0(k-s)\,d\ln(1-|r_1(s)|^2),\\
	\nu_4&=-\frac1{2\pi}\ln(1-|r_2(-k_0)|^2),&
	\chi_4(\xi;k)
	&=-\frac{1}{2\pi i}\int_{-\infty}^{-k_0}
	\log_\pi(k-s)\,d\ln(1-|r_2(s)|^2).
\end{aligned}
\]

Rotation gives
\[
\delta_3(k)=\delta_1(\omega^2 k),\quad \delta_5(k)=\delta_1(\omega k),\quad \delta_2(k)=\delta_4(\omega k),\quad \delta_6(k)=\delta_4(\omega^2 k).
\]
The jumps and Schwarz symmetries imply
\[
\overline{\delta_1(\bar k)}=\delta_1(k)^{-1},\qquad \overline{\delta_4(\bar k)}=\delta_4(k)^{-1},
\]
and their rotations.

Define the diagonal conjugating factor
\[
\Delta (k)=\mathrm{diag}\Bigl( \frac{\delta _1\delta _6}{\delta _3\delta _4},\;\frac{\delta _5\delta _4}{\delta _1\delta _2},\;\frac{\delta _3\delta _2}{\delta _5\delta _6} \Bigr) \doteq \mathrm{diag}\Bigl( \Delta _1,\;\Delta _2,\;\Delta _3 \Bigr) .
\]
Then
\[
\det \Delta =1,\qquad \Delta (k)=\mathcal{A} \Delta (\omega k)\mathcal{A} ^{-1}=\mathcal{B} \overline{\Delta (\bar{k})}\mathcal{B} .
\]
We set
\[
M^{(1)}(x,t;k)=M(x,t;k)\Delta(k),
\]
so \(v^{(1)}=\Delta_-^{-1}v\Delta_+\) and the normalization is unchanged. Let \(\Sigma^{(1)}=\Sigma\), oriented outward and divided at \(\pm\omega^m k_0\). The scalar jumps give the RH problem below. Figure~\ref{fig:SK-conjugated-jumps} shows the representative jumps, and symmetry gives the rest.

\begin{figure}[htbp]
	\centering
	\begin{tikzpicture}[
		scale=1.20,
		>=stealth,
		line cap=round,
		line join=round,
		every node/.style={font=\scriptsize}
		]
		\foreach \ang in {0,60,120,180,240,300}{
			\draw[thick,->] (0,0)--(\ang:1.80);
			\fill[red!80!black] (\ang:0.80) circle (1.20pt);
		}
		
		\node[right]       at (0:1.80)   {$\mathbb R$};
		\node[above right] at (60:1.80)  {$\omega^2\mathbb R$};
		\node[above left]  at (120:1.80) {$\omega\mathbb R$};
		
		\fill (0,0) circle (1.10pt);
		\node[below left] at (0,0) {$0$};
		\node[below] at (0.80,0) {$k_0$};
		\node[below] at (-0.80,0) {$-k_0$};
		
		\node at (1.30,0.20) {$v_1^{(1)}$};
		\node at (0.40,0.20) {$v_7^{(1)}$};
		\node at (-0.40,0.20) {$v_{10}^{(1)}$};
		\node at (-1.30,0.20) {$v_4^{(1)}$};
	\end{tikzpicture}
	\caption{The jump contour \(\Sigma^{(1)}\) and the saddle points \(\pm\omega^m k_0\), \(m=0,1,2\).}
	\label{fig:SK-conjugated-jumps}
\end{figure}
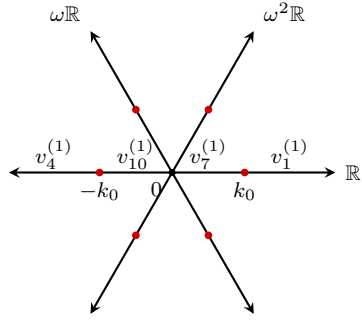

\begin{rhp}\label{rhp3}
Find a matrix $M^{(1)}$ meromorphic in $\mathbb C\setminus(\Sigma\cup\{0\})$, with simple poles at $\mathcal Z$, the two reductions of RH Problem~\ref{rhp1}, and jumps
\[
M_+^{(1)}=M_-^{(1)}v^{(1)},\qquad v^{(1)}=\Delta_-^{-1}v\Delta_+.
\]
At infinity and at the origin, respectively,
\[
M^{(1)}=I_3+\frac{M_1+\Delta^{[1]}}{k}+O(k^{-2}),\qquad
M_{0,n}^{(1,-1)}=M_{0,n}^{(-1)}\Delta(0),
\]
where $\Delta=I_3+\Delta^{[1]}/k+O(k^{-2})$. The residue conditions are
\[
\operatorname*{Res}_{k=\lambda}M^{(1)}(k)
=\lim_{k\to\lambda}M^{(1)}(k)C_\lambda^{(1)},\qquad
C_\lambda^{(1)}=\Delta(\lambda)^{-1}C_\lambda\Delta(\lambda),
\]
with $C_\lambda$ given by \eqref{2.6} and the representative residues above.
\end{rhp}
On the four real intervals $(k_0,\infty)$, $(-\infty,-k_0)$, $(0,k_0)$, and $(-k_0,0)$, denote the jumps by $v_1^{(1)}$, $v_4^{(1)}$, $v_7^{(1)}$, and $v_{10}^{(1)}$. Their factorizations below use
\[
\rho_j=\frac{r_j}{1-|r_j|^2},\qquad
\widetilde\delta_{v1}=\frac{\delta_3\delta_4^2\delta_5}{\delta_6\delta_2},\qquad
\widetilde\delta_{v4}=\frac{\delta_1^2\delta_2\delta_6}{\delta_5\delta_3}.
\]

Since $M^{(1)}=M\Delta$, $M_1^{(1)}=M_1+\Delta^{[1]}$, the row reconstruction \eqref{2.10} yields
\[
u(x,t)=-\frac{1}{2}\partial _x\Bigl[ \ell ^{\mathrm{T}}M_{1}^{(1)}e_3-(\Delta ^{[1]})_{33} \Bigr] ,
\]
where we used that $\Delta^{[1]}$ is diagonal.

\subsection{Transformation to a mixed \texorpdfstring{$\bar\partial$}{d-bar}-RH problem}
Set \(M^{(2)}=M^{(1)}R^{(2)}\) and \(N^{(2)}=\ell^{\mathrm T}M^{(2)}\). Then
\[
\begin{array}{c}
	M_{+}^{(2)}=M_{-}^{(2)}v^{(2)},\quad \bar{\partial}M^{(2)}=M^{(2)}W^{(2)},\quad W^{(2)}=(R^{(2)})^{-1}\bar{\partial}R^{(2)},\\
	N_{+}^{(2)}=N_{-}^{(2)}v^{(2)},\quad \bar{\partial}N^{(2)}=N^{(2)}W^{(2)},\quad N^{(2)}(k)\rightarrow \ell ^{\mathrm{T}}.\\
\end{array}
\]
Use the matrix form for right parametrices and the bounded row for reconstruction. The carriers are \(e^{\pm t\Phi_{21}}\).

By the $\mathbb Z_3$ and Schwarz symmetries, it suffices to work on the real axis. The four jumps admit the factorizations
\[
v_{1}^{(1)}=U_1L_1,\qquad v_{7}^{(1)}=L_7U_7,\qquad v_{4}^{(1)}=L_4U_4,\qquad v_{10}^{(1)}=U_{10}L_{10},
\]
on $(k_0,\infty)$, $(0,k_0)$, $(-\infty,-k_0)$, and $(-k_0,0)$, respectively, where
\[
\begin{aligned}
	U_1&=I-\frac{\widetilde\delta_{v1}}{\delta_{1-}^2}
	\rho_1e^{-t\Phi_{21}}E_{12},&
	L_1&=I+\frac{\delta_{1+}^2}{\widetilde\delta_{v1}}
	\overline{\rho_1}e^{t\Phi_{21}}E_{21},\\
	L_7&=I+\frac{\delta_1^2}{\widetilde\delta_{v1}}
	\overline{r_1}e^{t\Phi_{21}}E_{21},&
	U_7&=I-\frac{\widetilde\delta_{v1}}{\delta_1^2}
	r_1e^{-t\Phi_{21}}E_{12},\\
	L_4&=I+\frac{\widetilde\delta_{v4}}{\delta_{4-}^2}
	\rho_2e^{t\Phi_{21}}E_{21},&
	U_4&=I-\frac{\delta_{4+}^2}{\widetilde\delta_{v4}}
	\overline{\rho_2}e^{-t\Phi_{21}}E_{12},\\
	U_{10}&=I-\frac{\delta_4^2}{\widetilde\delta_{v4}}
	\overline{r_2}e^{-t\Phi_{21}}E_{12},&
	L_{10}&=I+\frac{\widetilde\delta_{v4}}{\delta_4^2}
	r_2e^{t\Phi_{21}}E_{21}.
\end{aligned}
\]
Symmetry gives the other factors. Figures~\ref{fig:SK-mixed-contour} and \ref{fig:SK-phase-signatures} show the mixed lens and its phase skeleton.

\begin{figure}[htbp]
	\centering
	
	\definecolor{SKMixBlue}{HTML}{4477AA}
	\definecolor{SKMixCoral}{HTML}{EE6677}
	\definecolor{SKMixCyan}{HTML}{66CCEE}
	\definecolor{SKMixGray}{HTML}{999999}
	
	\begin{tikzpicture}[
		scale=0.85,
		>=stealth,
		line cap=round,
		line join=round
		]
		\def\K{1.42}
		\def\R{4.15}
		\def\h{0.48}
		\def\E{0.36}
		
		\tikzset{
			oldcontour/.style={
				SKMixGray!80,
				dashed,
				line width=0.52pt
			},
			newcontour/.style={
				SKMixBlue,
				line width=0.92pt
			},
			dbarfill/.style={
				fill=SKMixCyan!18,
				draw=none
			},
			saddle/.style={
				circle,
				fill=SKMixCoral,
				draw=SKMixCoral!75!black,
				inner sep=1.35pt
			},
			pole/.style={
				circle,
				fill=black,
				draw=black,
				inner sep=0.95pt
			},
			cutoff/.style={
				circle,
				fill=white,
				draw=SKMixGray!85!black,
				line width=0.45pt,
				minimum size=7.5pt,
				inner sep=0pt
			},
			labelstyle/.style={
				font=\scriptsize
			}
		}
		
		\newcommand{\FillDbarSector}{%
			\fill[dbarfill]
			(\K,0)--(\R,0)
			--({\K+(\R-\K)*cos(12)},{(\R-\K)*sin(12)})
			--cycle;
			
			\fill[dbarfill]
			(\K,0)--(\R,0)
			--({\K+(\R-\K)*cos(12)},{-(\R-\K)*sin(12)})
			--cycle;
			
			\fill[dbarfill]
			(\E,0)--(\K,0)
			.. controls (0.70*\K,\h) and (0.35*\K,\h)
			.. (\E,0)--cycle;
			
			\fill[dbarfill]
			(\E,0)--(\K,0)
			.. controls (0.70*\K,-\h) and (0.35*\K,-\h)
			.. (\E,0)--cycle;
			
			\fill[dbarfill]
			(-\K,0)--(-\R,0)
			--({-\K-(\R-\K)*cos(12)},{(\R-\K)*sin(12)})
			--cycle;
			
			\fill[dbarfill]
			(-\K,0)--(-\R,0)
			--({-\K-(\R-\K)*cos(12)},{-(\R-\K)*sin(12)})
			--cycle;
			
			\fill[dbarfill]
			(-\E,0)--(-\K,0)
			.. controls (-0.70*\K,\h) and (-0.35*\K,\h)
			.. (-\E,0)--cycle;
			
			\fill[dbarfill]
			(-\E,0)--(-\K,0)
			.. controls (-0.70*\K,-\h) and (-0.35*\K,-\h)
			.. (-\E,0)--cycle;
		}
		
		\newcommand{\DrawDbarSector}{%
			\draw[oldcontour] (-\R,0)--(\R,0);
			
			\draw[
			->,
			SKMixGray!85!black,
			line width=0.42pt
			]
			(0,0)--(0.83,0);
			
			\draw[
			->,
			SKMixGray!85!black,
			line width=0.42pt
			]
			(0,0)--(-0.83,0);
			
			\node[saddle] at (-\K,0) {};
			\node[saddle] at (\K,0) {};
			
			\draw[newcontour,->]
			(\K,0)
			--({\K+(\R-\K)*cos(12)},{(\R-\K)*sin(12)});
			
			\draw[newcontour,->]
			(\K,0)
			--({\K+(\R-\K)*cos(12)},{-(\R-\K)*sin(12)});
			
			\draw[newcontour,->]
			(\E,0)
			.. controls (0.35*\K,\h) and (0.70*\K,\h)
			.. (\K,0);
			
			\draw[newcontour,->]
			(\E,0)
			.. controls (0.35*\K,-\h) and (0.70*\K,-\h)
			.. (\K,0);
			
			\draw[newcontour,->]
			(-\K,0)
			--({-\K-(\R-\K)*cos(12)},{(\R-\K)*sin(12)});
			
			\draw[newcontour,->]
			(-\K,0)
			--({-\K-(\R-\K)*cos(12)},{-(\R-\K)*sin(12)});
			
			\draw[newcontour,->]
			(-\E,0)
			.. controls (-0.35*\K,\h) and (-0.70*\K,\h)
			.. (-\K,0);
			
			\draw[newcontour,->]
			(-\E,0)
			.. controls (-0.35*\K,-\h) and (-0.70*\K,-\h)
			.. (-\K,0);
		}
		
		\foreach \rot in {0,120,240}{
			\begin{scope}[rotate=\rot]
				\FillDbarSector
			\end{scope}
		}
		
		\fill[white] (0,0) circle (\E);
		\draw[
		SKMixGray!85!black,
		densely dotted,
		line width=0.70pt
		]
		(0,0) circle (\E);
		
		\foreach \rot in {0,120,240}{
			\begin{scope}[rotate=\rot]
				\DrawDbarSector
			\end{scope}
		}
		
		\foreach \ang in {7,113,127,233,247,353}{
			\node[cutoff] at (\ang:2.55) {};
			\node[pole] at (\ang:2.55) {};
		}
		
		\fill (0,0) circle (1.15pt);
		\node[labelstyle,below right] at (0,0) {$0$};
		\node[labelstyle,right]
		at (18:{\E+0.07}) {$\mathcal U_0$};
		
		\node[labelstyle,below]
		at (\K,0) {$k_0$};
		
		\node[labelstyle,below]
		at (-\K,0) {$-k_0$};
		
		\node[labelstyle,above left]
		at (120:{\K+0.01}) {$\omega k_0$};
		
		\node[labelstyle,below right]
		at (300:{\K+0.01}) {$-\omega k_0$};
		
		\node[labelstyle,below left]
		at (240:{\K+0.01}) {$\omega^2 k_0$};
		
		\node[labelstyle,above right]
		at (60:{\K+0.01}) {$-\omega^2 k_0$};
		
		\fill[SKMixCyan!18]
		(-4.65,-4.28) rectangle (-4.00,-4.10);
		
		\draw[SKMixCyan!60!black,line width=0.45pt]
		(-4.65,-4.28) rectangle (-4.00,-4.10);
		
		\node[labelstyle,right]
		at (-3.92,-4.19)
		{$\operatorname{supp}\bar\partial R^{(2)}$ (schematic)};
		
		\draw[newcontour]
		(-4.65,-4.56)--(-4.00,-4.56);
		
		\node[labelstyle,right]
		at (-3.92,-4.56)
		{lens boundary $\Sigma^{(2)}$};
		
		\draw[oldcontour]
		(-4.65,-4.84)--(-4.00,-4.84);
		
		\node[labelstyle,right]
		at (-3.92,-4.84)
		{original six-ray contour $\Sigma^{(1)}$};
		
		\node[cutoff] at (-4.34,-5.14) {};
		\node[pole] at (-4.34,-5.14) {};
		
		\node[labelstyle,right]
		at (-3.92,-5.14)
		{pole cutoff: $\bar\partial R^{(2)}=0$};
		
	\end{tikzpicture}
	
	\caption{Mixed $\bar\partial$-RH deformation in the positive-velocity region.
		The shaded regions represent $\operatorname{supp}\bar\partial R^{(2)}$,
		excluding $\mathcal U_0$ and the discrete-spectrum cutoff neighborhoods.}
	\label{fig:SK-mixed-contour}
\end{figure}
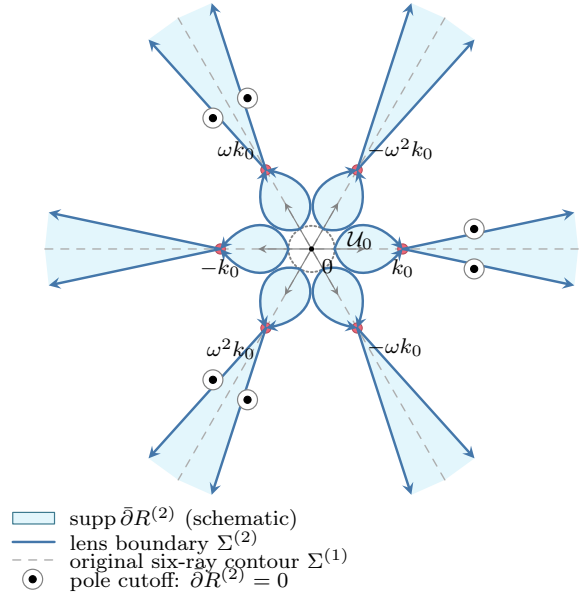

Separate the fractional stationary factors from the regular scalar parts. Near \(k=\pm k_0\), write
\[
\frac{\widetilde{\delta }_{v1}}{\delta _{1}^{2}}=\left( \frac{k-k_0}{k+k_0} \right) ^{2i\nu _1}e^{X_1(k)},\qquad \frac{\widetilde{\delta }_{v4}}{\delta _{4}^{2}}=\left( \frac{k+k_0}{k-k_0} \right) ^{2i\nu _4}e^{X_4(k)},
\]
where
\[
X_1(k)=\log \frac{\widetilde{\delta }_{v1}(k)}{\delta _{1}^{2}(k)}-2i\nu _1\log \frac{k-k_0}{k+k_0},\quad X_4(k)=\log \frac{\widetilde{\delta }_{v4}(k)}{\delta _{4}^{2}(k)}-2i\nu _4\log \frac{k+k_0}{k-k_0}.
\]
Then \(X_1,X_4\) are bounded and continuous at their saddles. Set $I_1=[k_0,\infty)$, $I_7=(0,k_0)$, $I_4=(-\infty,-k_0]$, $I_{10}=(-k_0,0)$, and define
\[
\begin{aligned}
	a_1^-&=e^{X_1}\rho_1,&
	a_1^+&=e^{-X_1}\overline{\rho_1},&
	a_7^+&=e^{X_1}r_1,&
	a_7^-&=e^{-X_1}\overline{r_1},\\
	a_4^-&=e^{X_4}\rho_2,&
	a_4^+&=e^{-X_4}\overline{\rho_2},&
	a_{10}^+&=e^{X_4}r_2,&
	a_{10}^-&=e^{-X_4}\overline{r_2}.
\end{aligned}
\]
Extend the fractional factors analytically, so they do not enter \(\bar\partial\)-derivatives.

Uniformly for $\xi\in K\Subset(0,\infty)$, the regular amplitudes satisfy
\begin{equation}\label{3.2}
	(1+|s|)(a_j^\pm)'(s)\in L^2(I_j),\qquad j\in\{1,4,7,10\},
\end{equation}
with uniformly bounded norms. Moreover, if $\lambda_j=k_0$ for $j=1,7$ and $\lambda_j=-k_0$ for $j=4,10$, then, for $q=0,1$,
\begin{equation}\label{3.3}
	\left| \partial _{k_0}^{q}\bigl[ a_{j}^{\pm}(\lambda _j+h;k_0)-a_{j}^{\pm}(\lambda _j;k_0) \bigr] \right|\le C|h|\bigl( 1+|\log |h|| \bigr) ,
\end{equation}
for small one-sided $h$. To obtain these bounds, subtract the endpoint value in the Cauchy integrals and split at distance $2|h|$. The resulting logarithmic term is $O(|h|(1+|\log|h||))$, also after one $k_0$-derivative. Together with the differentiated Volterra bounds, this proves \eqref{3.2}-\eqref{3.3}. A further mixed derivative satisfies
\[
|\partial _{k_0}\partial _sa_{j}^{\pm}(s;k_0)|\le C\left( |s-\lambda _j|^{-1}+(1+|\log |s-\lambda _j||)^2 \right) ,
\]
for \(0<|s-\lambda_j|<r_*\), with uniform weighted \(L^2\) and rapid decay away from the endpoints. For large \(t\), set \(\varepsilon_t=t^{-1/2}\), \(\mathcal U_0=\{|k|<\varepsilon_t\}\), so
\[
\varepsilon_t<\frac14\inf_{\xi\in K}k_0(\xi),\qquad \overline{\mathcal U_0}\cap\mathcal Z=\varnothing.
\]
Truncate the inner lips at \(\partial\mathcal U_0\) inside fixed decay subsectors. Then
\[
|\Im k|\ge c_0|k|,\qquad \Re \bigl( \pm \Phi _{21}(\xi ;k) \bigr) \le -c_1|k|,\qquad |k|\le \frac{1}{2}k_0,
\]
with the appropriate lip sign. Include the induced arcs clockwise in \(\Sigma^{(2)}\), and obtain the other sectors by rotation. Because \(\mathcal U_0\) depends only on \(t\), \(x\)-differentiation fixes the contour.

If $\mathcal Z\ne\varnothing$, choose
\[
0<\rho<\frac12\min\left\{\min_{\substack{\lambda,\mu\in\mathcal Z\\\lambda\ne\mu}}|\lambda-\mu|,\operatorname{dist}(\mathcal Z,\Sigma\cup\{0\})\right\},
\]
and a smooth reduction-invariant cutoff $\chi_{\mathcal Z}$ satisfying
\[
\chi_{\mathcal Z}=1\quad\text{if }\operatorname{dist}(k,\mathcal Z)\le\rho/3,\qquad\chi_{\mathcal Z}=0\quad\text{if }\operatorname{dist}(k,\mathcal Z)\ge2\rho/3.
\]
If $\mathcal Z=\varnothing$, set $\chi_{\mathcal Z}\equiv0$.

Introduce piecewise \(C^1\) extensions \(\mathfrak R_1^\pm,\mathfrak R_7^\pm,\mathfrak R_4^\pm,\mathfrak R_{10}^\pm\), incorporating the reflection and \(\delta\)-factors.

Define \(R^{(2)}(k)\) in the real-axis sector by
\[
R^{(2)}(k)=\begin{cases}
	I-\mathfrak{R} _{1}^{+}(k)e^{t\Phi _{21}(\xi ;k)}E_{21},&		k\in \mathcal{L} _{1}^{+},\\
	I-\mathfrak{R} _{1}^{-}(k)e^{-t\Phi _{21}(\xi ;k)}E_{12},&		k\in \mathcal{L} _{1}^{-},\\
	I+\mathfrak{R} _{7}^{+}(k)e^{-t\Phi _{21}(\xi ;k)}E_{12},&		k\in \mathcal{L} _{7}^{+}\setminus \overline{\mathcal{U} _0},\\
	I+\mathfrak{R} _{7}^{-}(k)e^{t\Phi _{21}(\xi ;k)}E_{21},&		k\in \mathcal{L} _{7}^{-}\setminus \overline{\mathcal{U} _0},\\
	I+\mathfrak{R} _{4}^{+}(k)e^{-t\Phi _{21}(\xi ;k)}E_{12},&		k\in \mathcal{L} _{4}^{+},\\
	I+\mathfrak{R} _{4}^{-}(k)e^{t\Phi _{21}(\xi ;k)}E_{21},&		k\in \mathcal{L} _{4}^{-},\\
	I-\mathfrak{R} _{10}^{+}(k)e^{t\Phi _{21}(\xi ;k)}E_{21},&		k\in \mathcal{L} _{10}^{+}\setminus \overline{\mathcal{U} _0},\\
	I-\mathfrak{R} _{10}^{-}(k)e^{-t\Phi _{21}(\xi ;k)}E_{12},&		k\in \mathcal{L} _{10}^{-}\setminus \overline{\mathcal{U} _0},\\
	I,&		\text{elsewhere in the real-axis sector},\\
\end{cases}
\]
where the scalar extensions \(\mathfrak R_j^\pm\) are specified in the following proposition.
\begin{prop}\label{prop:3.1}
The conjugated reflection data admit piecewise \(C^1\) extensions \(\mathfrak R_j^\pm\), \(j\in\{1,4,7,10\}\), satisfying the \(\bar\partial\)-estimates below.
Write the analytic fractional factors as
\[
\begin{aligned}
F_1^-(k)=F_7^+(k)&=\left(\frac{k-k_0}{k+k_0}\right)^{2i\nu_1},&
F_1^+(k)=F_7^-(k)&=(F_1^-(k))^{-1},\\
F_4^-(k)=F_{10}^+(k)&=\left(\frac{k+k_0}{k-k_0}\right)^{2i\nu_4},&
F_4^+(k)=F_{10}^-(k)&=(F_4^-(k))^{-1}.
\end{aligned}
\]
With $a_j^\pm$ and $\lambda_j$ defined above, the boundary values are
\[
\mathfrak R_j^\pm(k)=
\begin{cases}
F_j^\pm(k)a_j^\pm(k),&k\in I_j\setminus\overline{\mathcal U_0},\\
F_j^\pm(k)a_j^\pm(\lambda_j)(1-\chi_{\mathcal Z}(k)),&k\in\Gamma_j^\pm.
\end{cases}
\]
The extensions satisfy
	\[
	|\mathfrak{R} _{j}^{\pm}(k)|\lesssim 1+\langle \mathrm{Re}k\rangle ^{-1/2},\qquad j\in\{1,4,7,10\},
	\]
	\[
	|\bar{\partial}\mathfrak{R} _{j}^{\pm}(k)|\le C\left( |\bar{\partial}\chi _{\mathcal{Z}}(k)|+|(a_{j}^{\pm})\prime(\mathrm{Re}k)|+d_j(k)^{-1/2} \right) ,\qquad j\in \{1,4,7,10\}.
	\]
	\[
	\bar\partial\mathfrak R_j^\pm(k)=0 \quad\text{if }k\notin \bigcup_{m\in\{1,4,7,10\}} \bigl(\mathcal L_m^+\cup\mathcal L_m^-\bigr), \quad\text{or if }
	\operatorname{dist}(k,\mathcal Z)<\frac{\rho}{3},
	\]
	where
	\[
	d_j(k)=\begin{cases}
		|k-k_0|,&		j=1,7,\\
		|k+k_0|,&		j=4,10.\\
	\end{cases}
	\]
	The corresponding estimates in the remaining sectors follow by symmetry $k\mapsto \omega k$. For $j=7,10$, the preceding estimates are understood only in $\left( \mathcal{L} _{j}^{+}\cup \mathcal{L} _{j}^{-} \right) \setminus \overline{\mathcal{U} _0}$. Moreover, $\bar{\partial}R^{(2)}(k)=0$, $k\in \mathcal{U} _0$. For fixed $t$ and $\xi$ in a compact positive interval, the extensions may be chosen so that
	\begin{equation}\label{3.4}
		x\longmapsto (R^{(2)})^{-1}\bar\partial R^{(2)}\quad\hbox{is }C^1\bigl(L^1(\mathbb C)\cap L^{6/5}(\mathbb C)\bigr).
	\end{equation}
\end{prop}
The definition in the remaining sectors is generated by the reductions
\[
R^{(2)}(\omega^m z)=\mathcal A^{-m}R^{(2)}(z)\mathcal A^m,\qquad m=0,1,2,
\]
together with $R^{(2)}(k)=\mathcal B\overline{R^{(2)}(\bar k)}\mathcal B$.

It remains to regularize \(k=0\). Freezing \(r_1(0)=\omega^2\), \(r_2(0)=1\) gives
\[
\begin{array}{c}
	v_{1}^{0}=\left( \begin{matrix}
		1&		-\omega ^2&		0\\
		\omega&		0&		0\\
		0&		0&		1\\
	\end{matrix} \right) ,\qquad v_{2}^{0}=\left( \begin{matrix}
		1&		0&		0\\
		0&		0&		-1\\
		0&		1&		1\\
	\end{matrix} \right) ,\\
	v_{2m+1}^{0}=\mathcal{A} ^{-m}v_{1}^{0}\mathcal{A} ^m,\qquad v_{2m+2}^{0}=\mathcal{A} ^{-m}v_{2}^{0}\mathcal{A} ^m,\qquad m=0,1,2.\\
\end{array}
\]
Define
\[
S_{1}^{(0)}=\left( \begin{matrix}
	1&		-\omega ^2&		1\\
	0&		1&		1\\
	0&		0&		1\\
\end{matrix} \right) ,\qquad S_{n}^{(0)}=S_{n-1}^{(0)}v_{n}^{0},\quad 2\le n\le 6.
\]
Since \(v_1^0\cdots v_6^0=I_3\), let \(\Delta_0=\Delta(0)\) and, in \(D_n\cap\mathcal U_0\), set
\[
F_n(k)=e^{\widehat{k\Lambda x+9k^5\Lambda ^2t}}\bigl( \Delta _{0}^{-1}S_{n}^{(0)}\Delta _0 \bigr) ,\qquad R^{(2)}(k)=F_n(k)^{-1}.
\]
The closed cycle makes the frames compatible, their root supports are non-growing, and \(t|k|^5=O(t^{-3/2})\) makes \(F_n^{\pm1}\) uniformly bounded. The endpoint expansions
\[
\Delta(k)\Delta(0)^{-1}=I_3+O(k),\qquad\partial_x[\Delta(k)\Delta(0)^{-1}]=O(k/t),
\]
give
\[
F_{n-1}v_n^{(1)}F_n^{-1}=I_3+O(k),\qquad\partial_x\bigl(F_{n-1}v_n^{(1)}F_n^{-1}-I_3\bigr)=O(k^2+k/t),
\]
where $F_0=F_6$.

On clockwise \(\partial\mathcal U_0\), \(v_{\rm cir}=F_nR_{\rm out}\), with the other arcs given by the reductions. Thus, for \(\Sigma_0(t)=(\Sigma\cap\mathcal U_0)\cup\partial\mathcal U_0\),
\begin{equation}\label{3.5}
	\begin{aligned}
		\|v^{(2)}-I_3\|_{L^\infty(\Sigma_0(t))}
		+\|\partial_xv^{(2)}\|_{L^\infty(\Sigma_0(t))}
		&\le Ct^{-1/2},\\
		\|v^{(2)}-I_3\|_{L^2(\Sigma_0(t))}
		+\|\partial_xv^{(2)}\|_{L^2(\Sigma_0(t))}
		&\le Ct^{-3/4},\\
		\|v^{(2)}-I_3\|_{L^1(\Sigma_0(t))}
		+\|\partial_xv^{(2)}\|_{L^1(\Sigma_0(t))}
		&\le Ct^{-1}.
	\end{aligned}
\end{equation}
These jumps enter the small-norm error, while \(\bar\partial R^{(2)}=0\) in \(\mathcal U_0\). The pole cutoff gives
\[
R^{(2)}(k)=I_3,\qquad\operatorname{dist}(k,\mathcal Z)<\frac{\rho}{3},
\]
so the discrete residue conditions are unchanged.

Then $M^{(2)}$ satisfies the following mixed $\bar\partial$-RH problem.
\begin{rhp}\label{rhp4}
Find $M^{(2)}$ continuous in $\mathbb C\setminus(\Sigma^{(2)}\cup\mathcal Z\cup\{0\})$ such that
\[
\begin{aligned}
M_+^{(2)}&=M_-^{(2)}v^{(2)},&
v^{(2)}&=(R_-^{(2)})^{-1}v^{(1)}R_+^{(2)},\\
\bar\partial M^{(2)}&=M^{(2)}W^{(2)},&
W^{(2)}&=(R^{(2)})^{-1}\bar\partial R^{(2)}.
\end{aligned}
\]
The normalization is $M^{(2)}\to I_3$ at infinity, and $M^{(2)}=O(k^{-1})$ at zero. Its reductions and residue conditions are those of RH Problem~\ref{rhp3}, since $R^{(2)}=I_3$ near every nonzero pole.
\end{rhp}
Schwartz decay and the phase signature give \(R^{(2)}=I_3+o(k^{-1})\) and
\(\partial_xR^{(2)}=o(k^{-1})\) on closed lens subsectors, including their
boundaries. Retaining the pole divisor, reconstruction on a recovery ray reads
\[
u(x,t)=-\frac12\partial_x\lim_{k\to\infty}k\left\{\bigl[N^{(2)}(k)D_I(k)^{-1}(R^{(2)}(k))^{-1}\Delta(k)^{-1}\bigr]_3-1\right\}.
\]
The braces contain the undeformed row. Set
\(N^{(3)}=N^{(2)}(M_{\mathrm{RHP}}^{(2)})^{-1}\) after constructing the pure RH
solution; only this row quotient is asserted removable at the singular SK endpoint.

\subsection{Analysis of the pure RH problem}
The pure part of RH Problem~\ref{rhp4} keeps its jumps and residues and sets \(\bar\partial R^{(2)}=0\):
\begin{rhp}
Find $M_{\rm RHP}^{(2)}$ meromorphic off $\Sigma^{(2)}$, with the jumps, simple poles, residue conditions, and reductions of RH Problem~\ref{rhp4}, such that
\[
M_{{\rm RHP},+}^{(2)}=M_{{\rm RHP},-}^{(2)}v^{(2)},\qquad
M_{\rm RHP}^{(2)}(k)\to I_3\quad(k\to\infty).
\]
The matrix is bounded at $k=0$ and satisfies $\bar\partial M_{\rm RHP}^{(2)}=0$ off its contour and poles.
\end{rhp}

The pure RH model is constructed as
\begin{equation}\label{3.6}
	M_{\mathrm{RHP}}^{(2)}(k)=\begin{cases}
		E_r(k)M^{(\mathrm{out)}}(k),&		k\in \mathbb{C} \setminus \overline{\mathcal{U} },\\
		E_r(k)M^{(\mathrm{LC)}}(k)=E_r(k)M^{(\mathrm{out)}}(k)M^{(\mathrm{mod)}}(k),&		k\in \mathcal{U} ,\\
	\end{cases}
\end{equation}
Here \(M^{(\mathrm{out})}\) carries the retained poles from \eqref{3.19}, the local factors come from the normalized parabolic-cylinder model, and \(E_r\) solves the error problem.

Choose $\rho_0>0$ sufficiently small and set
\[
\mathcal U=\bigcup_{\lambda\in\mathcal K_0}\mathcal U_\lambda,\qquad\mathcal U_\lambda=\{k:|k-\lambda|<\rho_0\},
\]
with mutually disjoint disks satisfying $\overline{\mathcal U}\cap\bigl(\overline{\mathcal U_0}\cup\mathcal Z\bigr)=\varnothing$.

For $\sigma_d=\{(\kappa_j,c_j)\}_{j=1}^{N_1}\cup\{(\eta_\ell,d_\ell)\}_{\ell=1}^{N_4}$, define
$\sigma_d^{(\mathrm{tr})}=\{(\kappa_j,\widetilde c_j)\}_{j=1}^{N_1}\cup\{(\eta_\ell,\widetilde d_\ell)\}_{\ell=1}^{N_4}$, where $\widetilde c_j=c_j\frac{\widetilde\delta_{v1}(\kappa_j)}{\delta_1^2(\kappa_j)}$, $\widetilde d_\ell=d_\ell\frac{\widetilde\delta_{v4}(\eta_\ell)}{\delta_4^2(\eta_\ell)}$. Each transmission factor is absorbed once, and its dependence on \(\xi\) through \(k_0\) is suppressed. In particular,
\[
\partial_x\log\delta_1(\kappa)=-\frac{\partial_xk_0}{2\pi i}\frac{\log(1-|r_1(k_0)|^2)}{k_0-\kappa},
\]
so outer coefficients use fixed-\(t\) total derivatives. Symmetry gives the remaining data.

Let \(\widehat\sigma_d=\{(\kappa_j,\widehat c_j)\}_j\cup\{(\eta_\ell,\widehat d_\ell)\}_\ell\) generate finitely many disjoint admissible six-point orbits under \eqref{2.6}, with pole set \(\mathcal Z(\widehat\sigma_d)\).

\begin{rhp}\label{rhp6}
	Find a \(3\times3\) matrix-valued function \(M(k;x,t\mid\widehat\sigma_d)\) such that:
	\begin{enumerate}
		\item[(i)]
		\(M\) is analytic in \(\mathbb C\setminus\mathcal Z(\widehat\sigma_d)\), has simple poles precisely at \(\mathcal Z(\widehat\sigma_d)\), and is analytic at \(k=0\).
		
		\item[(ii)]
		As \(k\to\infty\),
		\[
		M(k;x,t\mid \widehat{\sigma }_d)\rightarrow I_3.
		\]
		
		\item[(iii)]
		The function \(M\) satisfies
		\[
		M(k;x,t\mid\widehat\sigma_d)=\mathcal A M(\omega k;x,t\mid\widehat\sigma_d)\mathcal A^{-1}=\mathcal B\overline{M(\bar k;x,t\mid\widehat\sigma_d)}\mathcal B.
		\]
		
		\item[(iv)]
		Residue conditions:
		\[
		\underset{k=\kappa _j}{\mathrm{Res}}M(k;x,t\mid \widehat{\sigma }_d)=\lim_{k\rightarrow \kappa _j} M(k;x,t\mid \widehat{\sigma }_d)\left( \begin{matrix}
			0&		-\widehat{c}_je^{-\theta _{21}(x,t;\kappa _j)}&		0\\
			0&		0&		0\\
			0&		0&		0\\
		\end{matrix} \right) ,
		\]
		\[
		\underset{k=\eta _{\ell}}{\mathrm{Res}}M(k;x,t\mid \widehat{\sigma }_d)=\lim_{k\rightarrow \eta _{\ell}} M(k;x,t\mid \widehat{\sigma }_d)\left( \begin{matrix}
			0&		0&		0\\
			\widehat{d}_{\ell}e^{\theta _{21}(x,t;\eta _{\ell})}&		0&		0\\
			0&		0&		0\\
		\end{matrix} \right) .
		\]
		The residue conditions at all remaining orbit points are generated by \eqref{2.6}.
	\end{enumerate}
\end{rhp}
The reflectionless problem is finite dimensional. Write
\[
M(k;x,t\mid\widehat\sigma_d)=I_3+\sum_{\lambda\in\mathcal Z(\widehat\sigma_d)}\frac{R_\lambda}{k-\lambda},
\]
the residue conditions are equivalent to
\[
R_\lambda=\left(I_3+\sum_{\substack{\mu\in\mathcal Z(\widehat\sigma_d)\\ \mu\neq\lambda}}\frac{R_\mu}{\lambda-\mu}\right)C_\lambda,\qquad\lambda\in\mathcal Z(\widehat\sigma_d).
\]
Whenever the finite residue system is invertible, it determines the unique reflectionless RH model. The nilpotent residue conditions give \(\det M=1\), while bounds on the residues and their first \(x\)-derivatives control \(M^{\pm1}\) away from the poles. The identification with a scalar SK soliton uses the compatibility conditions below.

For the scalar reflectionless solutions, let
\[
0<\eta _1<\cdots <\eta _N,\quad p_j=\sqrt{3}\,\eta _j,\quad v_j=p_{j}^{4}=9\eta _{j}^{4},\quad y_j=\beta _je^{p_j(x-v_jt)},\quad \beta _j>0.
\]
For \(i<j\), set
\[
A_{ij}=\frac{(p_i-p_j)^2(p_{i}^{2}-p_ip_j+p_{j}^{2})}{(p_i+p_j)^2(p_{i}^{2}+p_ip_j+p_{j}^{2})}\in (0,1).
\]
For \(S\subset\{1,\ldots,N\}\), put \(A(S)=\prod_{i<j,\ i,j\in S}A_{ij}\), with empty products equal to one, and define
\begin{equation}\label{3.7}
	\Delta_N(x,t)
	=\sum_{S\subset\{1,\ldots,N\}}A(S)\prod_{j\in S}y_j,\qquad u_{\mathrm{sol}}(x,t;\sigma^{\mathrm{phys}})=\partial_x^2\log\Delta_N(x,t).
\end{equation}
For \(k_j=i\omega\eta_j\), put \(a_{2j-1}=k_j\), \(a_{2j}=\bar k_j\), and
\[
\nu=\prod_{p<q}(a_p-a_q),\qquad
\nu_{j0}=\nu\big|_{k_j\mapsto\omega k_j},\qquad
\Sigma_N=\sum_{p=1}^{2N}a_p,\qquad
\Pi_N=\sum_{1\le p\le q\le2N}a_pa_q.
\]
Given \(\beta_j>0\), define
\begin{equation}\label{3.8}
	\gamma_j=-2\omega\beta_j\frac{\nu}{\nu_{j0}},\qquad\beta_j=\frac{\nu_{j0}\gamma_j+\nu_{0j}\bar\gamma_j}{2\nu}>0,
\end{equation}
where \(\nu_{0j}=\nu|_{\bar k_j\mapsto\omega^2\bar k_j}\), and write
\[
\sigma ^{\mathrm{phys}}=\{(k_j,\gamma _j)\}_{j=1}^{N}\longleftrightarrow \{(\eta _j,\beta _j)\}_{j=1}^{N},
\]
for scalar simple-pole \(P=0\) data. The conditions below specify their scalar closure and matrix realization. Also, $\bar\gamma_j=-2\omega^2\beta_j\nu/\nu_{0j}$.

\medskip
\noindent\emph{Scalar closure and matrix realization.}
	For the scalar SK data above, set
	\[
	\chi_j(x,t)=e^{-(k_j+\bar k_j)x-9(k_j^5+\bar k_j^5)t},\qquad j=1,\ldots,N.
	\]
	For \(0\le\ell\le2N\), define
	\[
	\mathfrak m_{\ell,2j-1}=k_j^\ell+(\omega k_j)^\ell\gamma_j\chi_j,\qquad\mathfrak m_{\ell,2j}=\bar k_j^\ell+(\omega^2\bar k_j)^\ell\bar\gamma_j\chi_j.
	\]
	Let \(\mathbf M_{\rm alg}\) and \(\mathbf B_{\rm alg}\) be given by
	\[
	(\mathbf M_{\rm alg})_{pq}=\mathfrak m_{2N-q,p},\qquad(\mathbf B_{\rm alg})_p=\mathfrak m_{2N,p},\qquad 1\le p,q\le2N.
	\]
	We assume
	\begin{equation}\label{3.9}
		\det\mathbf M_{\rm alg}(x,t)=\nu\,\Delta_N(x,t)^2\ne0,
	\end{equation}
	and define
	\[
	\bigl(A_{2N-1}^{\rm alg},\ldots,A_0^{\rm alg}\bigr)^{\rm T}=-\mathbf M_{\rm alg}^{-1}\mathbf B_{\rm alg}.
	\]
	We further assume
	\begin{equation}\label{3.10}
		\frac{1}{2}\bigl( A_{2N-1}^{\mathrm{a}\lg} \bigr) ^2-\partial _xA_{2N-1}^{\mathrm{a}\lg}-A_{2N-2}^{\mathrm{a}\lg}=\Pi _N-\frac{1}{2}\Sigma _{N}^{2},
	\end{equation}
	for all real \((x,t)\). Equivalently, the scalar reconstruction lies on the \(P=0\) branch and
	\[
	u_{\rm sol}=-\frac12\partial_xA_{2N-1}^{\rm alg}=\partial_x^2\log\Delta_N.
	\]
	Form the sectorial matrices \(\mathcal M_{{\rm sc},n}\) using the normalizations in \eqref{2.2} and \eqref{2.3}. Assume that their nonzero poles are simple, with invertible constant Laurent coefficients and the residue conditions of RH Problem~\ref{rhp6} for a datum \(\widehat\sigma_d\). At \(k=0\), require the exact sectorial factorization
\[
\mathcal M_{{\rm sc},n}(k)=(I_3+A_n/k)R_n(k),
\qquad
A_n^2=0,\quad \operatorname{rank}A_n\le1,\quad
\ell^{\rm T}A_n=0.
\]
Here \(R_n\) are analytic and invertible at \(0\) and glue to a rational solution \(R\) of RH Problem~\ref{rhp6} with the same datum. Denote the resulting correspondence by
\(\widehat\sigma_d=\mathfrak G_N(\sigma^{\rm phys})\).

A rational datum is \emph{SK-realizable} if it equals \(\mathfrak G_N(\sigma^{\rm phys})\) for scalar data satisfying the preceding closure and matrix realization conditions. A transmission-freezing construction has the \emph{scalar-to-matrix compatibility property} if the fixed datum obtained by evaluating each retained transmission factor at its orbit velocity is SK-realizable.

\begin{prop}\label{prop:3.2}
	Let $\widehat\sigma_d=\mathfrak G_N(\sigma^{\rm phys})$ be SK-realizable. Then RH Problem~\ref{rhp6} is uniquely solvable. If
	\[
	R(k)=I_3+\frac{R_1}{k}+O(k^{-2}),
	\]
	denotes its solution, there exist reduction-related matrices $A_n$ with $\operatorname{rank}A_n\le1$ such that
	\[
	\mathcal M_{{\rm sol},n}(k)=(I_3+A_n/k)R(k),\qquad\ell^{\rm T}A_n=0,\qquad\det\mathcal M_{\rm sol}=1.
	\]
	The completed matrix satisfies the sectorial origin condition and the two reductions of RH Problem~\ref{rhp1}, together with
	\begin{equation}\label{3.11}
			(\mathcal M_{\rm sol})_x-[k\Lambda,\mathcal M_{\rm sol}]
			=\frac{Q_1(u_{\rm sol})}{k}\mathcal M_{\rm sol}, \quad
			(\mathcal M_{\rm sol})_t-[9k^5\Lambda^2,\mathcal M_{\rm sol}]
			=P(u_{\rm sol};k)\mathcal M_{\rm sol}.
	\end{equation}
	Moreover, the reflectionless SK solution associated with
	\(\widehat\sigma_d\) is defined by
	\begin{equation}\label{3.12}
		u_{\rm sol}(x,t;\widehat\sigma_d)
		\doteq-\frac12\partial_x(\ell^{\rm T}R_1e_3)
		=u_{\rm sol}(x,t;\sigma^{\rm phys})
		=\partial_x^2\log\Delta_N(x,t).
	\end{equation}
\end{prop}

\begin{proof}
	Equations~\eqref{3.9}-\eqref{3.10} give the invertible \(P=0\) scalar system and \(u_{\rm sol}=\partial_x^2\log\Delta_N\). The matrix realization identifies its residues and origin factorization with RH Problem~\ref{rhp6}, whose quotient argument gives uniqueness. The scalar columns and Abel's identity give the reductions and determinant, differentiation gives \eqref{3.11}, and \(\ell^{\rm T}A_n=0\) together with the Marchenko formula gives \eqref{3.12}.
\end{proof}
For $\lambda=a+ib\in\mathcal Z_1\cup\mathcal Z_4$, define
\begin{equation}\label{3.13}
	v_\lambda=9\frac{\operatorname{Im}(\lambda^5)}{\operatorname{Im}\lambda}=45a^4-90a^2b^2+9b^4,
\end{equation}
so that
\[
\operatorname{Re}\Phi_{21}(\xi;\lambda)=\sqrt3\,b(\xi-v_\lambda).
\]
Equation~\eqref{2.6} gives the same rate throughout each orbit.
Let
\[
	V_*=\max\{v_\lambda:\lambda\in\mathcal Z_1\cup\mathcal Z_4\},
\]
with \(V_*=-\infty\) when \(\mathcal Z_1\cup\mathcal Z_4=\varnothing\).

For the positive half-cone \(\xi\ge M\), let
\[
	r=r(M):=\#\{j\in\{1,\ldots,N\}:v_j<M\}.
\]
The retained scalar soliton data are defined by
\begin{equation}\label{3.14}
	\begin{aligned}
		\widehat\beta_{j,M}&=\beta_j\tau_j(v_j)\prod_{i=1}^{r}A_{ij},
		\qquad r<j\le N,\\
		\widehat\sigma_M^{\rm phys}
		&=\{(\eta_j,\widehat\beta_{j,M}):r<j\le N\}, \qquad
		\widehat\sigma_M=\mathfrak G_{N-r}(\widehat\sigma_M^{\rm phys}).
	\end{aligned}
\end{equation}

For the first and second discrete families, respectively, define
\begin{equation}\label{3.15}
	D_{\kappa}(k)=\operatorname{diag}\left(
	\frac{(k-\omega\kappa)(k-\bar\kappa)}
	{(k-\kappa)(k-\omega\bar\kappa)},
	\frac{(k-\kappa)(k-\omega^2\bar\kappa)}
	{(k-\omega^2\kappa)(k-\bar\kappa)},
	\frac{(k-\omega^2\kappa)(k-\omega\bar\kappa)}
	{(k-\omega\kappa)(k-\omega^2\bar\kappa)}
	\right),
\end{equation}
\begin{equation}\label{3.16}
	D_{\eta}(k)=\operatorname{diag}\left(
	\frac{(k-\eta)(k-\omega\bar\eta)}
	{(k-\omega\eta)(k-\bar\eta)},
	\frac{(k-\omega^2\eta)(k-\bar\eta)}
	{(k-\eta)(k-\omega^2\bar\eta)},
	\frac{(k-\omega\eta)(k-\omega^2\bar\eta)}
	{(k-\omega^2\eta)(k-\omega\bar\eta)}
	\right).
\end{equation}

Let $I=[\xi_-,\xi_+]\Subset(0,\infty)$ and set
\[
D_I(k)=
\prod_{v_{\kappa_j}>\xi_+}D_{\kappa_j}(k)
\prod_{v_{\eta_\ell}>\xi_+}D_{\eta_\ell}(k)
=\operatorname{diag}\bigl((D_I)_1,(D_I)_2,(D_I)_3\bigr),
\]
with the empty product equal to $I_3$. Then
\[
\det D_I=1,\qquad D_I(0)=I_3,\qquad D_I(k)=\mathcal A D_I(\omega k)\mathcal A^{-1}=\mathcal B\overline{D_I(\bar k)}\mathcal B,
\]
and
\[
D_I(k)=I_3+\frac{D_I^{[1]}}{k}+O(k^{-2}),\qquad \partial_xD_I^{[1]}=0.
\]

For the retained representatives, define
\begin{align}
	c_j(I;\xi)
	&=\widetilde c_j(\xi)
	\frac{(D_I)_2(\kappa_j)}{(D_I)_1(\kappa_j)}, \qquad
	v_{\kappa_j}\in I, \label{3.17}\\
	d_\ell(I;\xi)
	&=\widetilde d_\ell(\xi)
	\frac{(D_I)_1(\eta_\ell)}{(D_I)_2(\eta_\ell)}, \qquad
	v_{\eta_\ell}\in I, \label{3.18}
\end{align}
and
\[
\sigma_d(I;\xi)=\{(\kappa_j,c_j(I;\xi)):v_{\kappa_j}\in I\}\cup\{(\eta_\ell,d_\ell(I;\xi)):v_{\eta_\ell}\in I\}.
\]
Complete the orbits by \eqref{2.6} and set
\begin{equation}\label{3.19}
	M^{(\mathrm{out})}(k)=M(k;x,t\mid\sigma_d(I;\xi)).
\end{equation}

For a switched pole with residue $C_\zeta=\gamma E_{pq}$, write locally $D=\widehat D_\zeta D_{s,\zeta}$, where $\widehat D_\zeta$ is analytic and invertible at $\zeta$, and the $p$- and $q$-entries of $D_{s,\zeta}$ are $(k-\zeta)^{-1}$ and $k-\zeta$. If
\[
\widehat D_\zeta(\zeta)^{-1}C_\zeta\widehat D_\zeta(\zeta)=\widehat\gamma E_{pq},
\]
then the switched residue is $C_\zeta^{(0)}=\widehat\gamma^{-1}E_{qp}$, as follows from
\begin{equation}\label{3.20}
	\left( \begin{matrix}
		(k-\zeta )^{-1}&		\widehat{\gamma }\\
		0&		k-\zeta\\
	\end{matrix} \right) =\left( \begin{matrix}
		0&		\widehat{\gamma }\\
		-\widehat{\gamma }^{-1}&		k-\zeta\\
	\end{matrix} \right) \left( \begin{matrix}
		1&		0\\
		\frac{\widehat{\gamma }^{-1}}{k-\zeta}&		1\\
	\end{matrix} \right) .
\end{equation}
Thus the conjugation never uses the singular value \(D(\zeta)\).

Choose the switched channel when its scalar residue factor exceeds one in modulus. The balanced equations are
\begin{equation}\label{3.21}
	\mathbb A_I^{\rm bal}(x,t)\mathbf r_I=\mathbf b_I.
\end{equation}
Both representations agree on switching curves.

\begin{lemma}\label{lem:3.1}
	For a fixed finite orbit set separated from $\Sigma\cup\{0\}$, the entries of $\mathbb A_I^{\rm bal}$, $\mathbf b_I$, and their fixed-$t$ total $x$-derivatives are uniformly bounded for $x/t\in I$. If
	\begin{equation}\label{3.22}
		\inf_{t\ge t_0,\ x/t\in I}|\det\mathbb A_I^{\rm bal}(x,t)|\ge d_I>0,
	\end{equation}
	then the outer residues and their $x$-derivatives, as well as $M^{(\mathrm{out})}$ and $(M^{(\mathrm{out})})^{-1}$ away from the fixed pole disks, are uniformly bounded.
\end{lemma}

\begin{proof}
	Balancing selects the member of \(\{\gamma,\gamma^{-1}\}\) of modulus at most one, whose \(x\)-derivative is a fixed spectral multiple. Pole separation bounds the system and its derivative. The determinant gap, Cramer's rule, and differentiated \eqref{3.21} bound \(\mathbf r_I,\partial_x\mathbf r_I\), and the rational representation gives the matrix bounds.
\end{proof}

For a fixed interval $I=[\xi_-,\xi_+]\Subset(0,\infty)$, the selected moving datum is called \emph{regular on $I$} if the determinant gap \eqref{3.22} holds.

Choose uniformly disjoint stationary and pole disks, retaining endpoint-velocity orbits. Then
\[
\mu(I)=\frac{\sqrt3}{2}\min_{\substack{\lambda\in\mathcal Z_1\cup\mathcal Z_4\\
		v_\lambda\notin I}}|\operatorname{Im}\lambda|\,\operatorname{dist}(v_\lambda,I)>0,
\]
with empty minimum \(+\infty\). Lemma~\ref{lem:3.1} bounds \(M^{(\mathrm{out})\pm1}\) off the pole disks.

\begin{prop}
	Apply $D_I$ to the exact mixed $\bar\partial$-RH problem before constructing the outer and local models. There exists a piecewise right multiplier $\mathcal T_I$, equal to $D_I$ outside small neighborhoods of the omitted orbits, such that $\widetilde M^{(2)}=M^{(2)}\mathcal T_I$ has precisely the retained poles $\mathcal Z(\sigma_d(I;\xi))$, with
	\begin{equation}\label{3.23}
		C_\zeta^I=D_I(\zeta)^{-1}C_\zeta^{(2)}D_I(\zeta).
	\end{equation}
	The omitted poles are replaced by clockwise circle jumps satisfying
	\begin{equation}\label{3.24}
		\|v_{\mathrm{pole}}-I_3\|_{L^\infty\cap L^2\cap L^1}+\|\partial_xv_{\mathrm{pole}}\|_{L^\infty\cap L^2\cap L^1}=O(e^{-\mu(I)t}).
	\end{equation}
	Outside the pole disks,
	\begin{equation}\label{3.25}
		v^{(2,I)}=D_I^{-1}v^{(2)}D_I, \qquad W^{(2,I)}=D_I^{-1}W^{(2)}D_I.
	\end{equation}
	The origin frame and the two reductions are preserved.
\end{prop}

\begin{proof}
	Analytic conjugation by \(D_I\) gives \eqref{3.23} and \eqref{3.17}-\eqref{3.18}. For omitted poles, \eqref{3.20}, the phase identity, and \(\mu(I)\) give \eqref{3.24}, including one \(x\)-derivative. Since \(W^{(2)}=0\) in the pole disks and \(\mathcal T_I=D_I\) elsewhere, \eqref{3.25} follows. Finally, \(D_I(0)=I_3\) and the local switches are away from zero, preserving the origin frame and reductions.
\end{proof}

Relabel \(\widetilde M^{(2)}\) as \(M^{(2)}\), add the omitted-pole circles to \(\Sigma^{(2)}\), suppress \(I\) on the jumps, and define
\[
W_I^{(2)}:=\mathcal T_I^{-1}\bigl((R^{(2)})^{-1}\bar\partial R^{(2)}\bigr)\mathcal T_I .
\]
Because \(\bar\partial R^{(2)}=0\) in the pole disks,
\begin{equation}\label{3.26}
	W_I^{(2)}=D_I^{-1}(R^{(2)})^{-1}\bar\partial R^{(2)}D_I,\qquad\operatorname{supp}W_I^{(2)}=\operatorname{supp}\bar\partial R^{(2)}.
\end{equation}
Write \(W^{(2)}=W_I^{(2)}\). The outer model has
\[
	M^{(\mathrm{out})}(k)=I_3+\frac{M_1^{(\mathrm{out})}}{k}+O(k^{-2}),\qquad k\to\infty.
\]
Figure~\ref{fig:SK-pure-RH-geometry} shows the resulting geometry.

\begin{figure}[htbp]
	\centering
	
	\definecolor{SKRHBlue}{HTML}{4477AA}
	\definecolor{SKRHCoral}{HTML}{EE6677}
	\definecolor{SKRHPurple}{HTML}{AA3377}
	\definecolor{SKRHGray}{HTML}{999999}
	
	\begin{tikzpicture}[
		scale=1.05,
		>=stealth,
		line cap=round,
		line join=round
		]
		\def\ks{1.82}
		\def\rd{0.23}
		\def\rzero{0.47}
		\def\Rout{3.42}
		
		\tikzset{
			base/.style={
				SKRHGray!80,
				dashed,
				line width=0.46pt
			},
			err/.style={
				SKRHBlue,
				line width=0.88pt
			},
			localcircle/.style={
				SKRHPurple,
				line width=0.78pt
			},
			retained/.style={
				circle,
				fill=black,
				draw=black,
				inner sep=1.2pt
			},
			omitted/.style={
				circle,
				fill=white,
				draw=SKRHGray!85!black,
				line width=0.65pt,
				inner sep=1.15pt
			},
			polecircle/.style={
				SKRHGray!85,
				dashed,
				line width=0.55pt
			},
			tinylabel/.style={
				font=\scriptsize
			}
		}
		
		\foreach \ang in {0,60,120,180,240,300}{
			\draw[base] (0,0)--(\ang:\Rout);
		}
		
		\draw[
		SKRHGray!85!black,
		densely dotted,
		line width=0.72pt
		]
		(0,0) circle (\rzero);
		
		\foreach \ang in {0,60,120,180,240,300}{
			
			\draw[err]
			({\ang+8}:\rzero)
			.. controls ({\ang+5}:0.94)
			and ({\ang+2}:1.38)
			.. (\ang:{\ks-\rd});
			
			\draw[err]
			({\ang-8}:\rzero)
			.. controls ({\ang-5}:0.94)
			and ({\ang-2}:1.38)
			.. (\ang:{\ks-\rd});
			
			\draw[err]
			(\ang:{\ks+\rd})
			--({\ang+8}:\Rout);
			
			\draw[err]
			(\ang:{\ks+\rd})
			--({\ang-8}:\Rout);
			
			\draw[localcircle]
			(\ang:\ks) circle (\rd);
			
			\fill[SKRHCoral]
			(\ang:\ks) circle (1.15pt);
		}
		
		\foreach \ang in {24,96,144,216,264,336}{
			\node[retained] at (\ang:2.62) {};
		}
		
		\foreach \ang in {40,80,160,200,280,320}{
			\draw[polecircle,<-]
			(\ang:3.02) circle (0.115);
			
			\node[omitted]
			at (\ang:3.02) {};
		}
		
		\fill (0,0) circle (1.05pt);
		
		\node[tinylabel,below]
		at (0:{\ks+0.05}) {$k_0$};
		
		\node[tinylabel,below]
		at (180:{\ks+0.05}) {$-k_0$};
		
		\node[tinylabel,above right]
		at (38:{\rzero+0.05}) {$\mathcal U_0$};
		
		\node[tinylabel,above]
		at (90:2.25) {$\mathcal K_0$ local disks};
		
		\draw[err]
		(-3.55,-3.34)--(-2.90,-3.34);
		
		\node[tinylabel,right]
		at (-2.82,-3.34)
		{$\Sigma^{(2)}$ part of $\Sigma^{(E)}$};
		
		\draw[localcircle]
		(-3.30,-3.67) circle (0.10);
		
		\node[tinylabel,right]
		at (-3.08,-3.67)
		{local matching circles $\partial\mathcal U_\lambda$};
		
		\node[retained]
		at (-3.30,-4.00) {};
		
		\node[tinylabel,right]
		at (-3.08,-4.00)
		{retained poles $\mathcal Z(\sigma_d(I;\xi))$};
		
		\draw[polecircle]
		(-3.30,-4.33) circle (0.10);
		
		\node[omitted]
		at (-3.30,-4.33) {};
		
		\node[tinylabel,right]
		at (-3.08,-4.33)
		{omitted poles $\to$ exponentially small circle jumps};
		
	\end{tikzpicture}
	
	\caption{Pure RH geometry after the cone-dependent pole reduction.
		Solid points are retained poles, and open points are omitted poles replaced by exponentially small jumps. The disks around \(\mathcal K_0\) and \(\mathcal U_0\) correspond to the stationary-point and endpoint local models, respectively.}
	\label{fig:SK-pure-RH-geometry}
\end{figure}
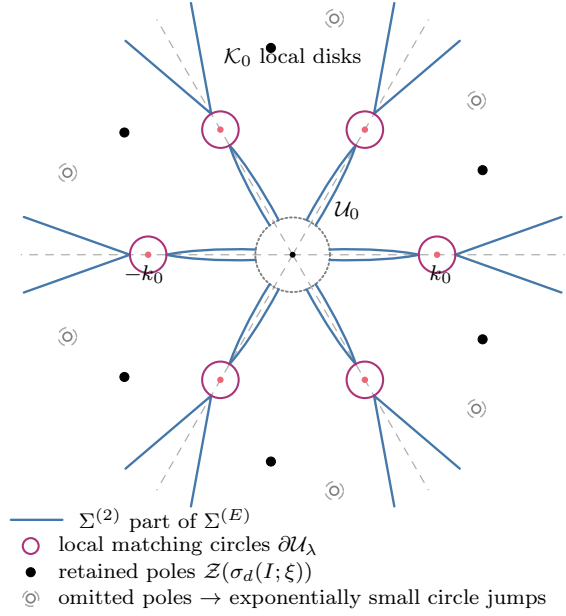
Define the cone-dependent pure-discrete outer RH profile by
\[
u_{\rm sol}(x,t\mid\sigma_d(I;\xi)):=-\frac12\partial_x\bigl(\ell^{\rm T}M_1^{(\mathrm{out})}e_3\bigr),
\]
Here \(\partial_x\) is the fixed-\(t\) total derivative. The vertical bar marks the \(\xi\)-dependent moving outer profile, which is an exact SK soliton only under the completion condition.

Differentiating the scalar conjugator with \(\partial_xk_0=(180tk_0^3)^{-1}\) gives
\[
\partial_x(\Delta^{[1]})_{33}
=\frac{\sqrt3}{360\pi tk_0^3}
\left[\log(1-|r_1(k_0)|^2)+\log(1-|r_2(-k_0)|^2)\right]
=O(t^{-1}),
\]
uniformly on compact positive-velocity intervals.

Uniformly for $\xi$ in compact positive intervals, the phase signature and Schwartz bounds give, for $p=1,2,\infty$ and $q=0,1$,
\[
\|\partial_x^q(v^{(2)}-I_3)\|_{L^p(\Sigma^{(2)}\setminus(\mathcal U\cup\Sigma_0(t)))}
\le Ce^{-c\sqrt t}.
\]
The endpoint jumps satisfy \eqref{3.5}, and the local jumps are uniformly bounded. Set $\Sigma_\lambda^{\rm loc}=\Sigma^{(2)}\cap\mathcal U_\lambda$.
Use the normalized complete-cross models of Appendix~\ref{app:PC-model} in each \(\mathcal U_\lambda\), retaining the exact-to-frozen difference in the global error. No nonidentity model jump ends at \(\partial\mathcal U_\lambda\).

For \(\lambda\in\mathcal K_0\), write \(\Phi_\lambda\) for the corresponding
phase and set
\[
\zeta_\lambda(k)=\alpha_\lambda\sqrt t\,(k-\lambda),\qquad
\alpha_{\omega^m(\pm k_0)}=\omega^{-m}\alpha_0,\qquad
\alpha_0=(180\sqrt3\,k_0^3)^{1/2}.
\]
Then
\[
t\{\Phi_\lambda(k)-\Phi_\lambda(\lambda)\}
=\frac{i\sigma_\lambda}{2}\zeta_\lambda^2+O(t^{-1/2}\zeta_\lambda^3),
\qquad \sigma_{\omega^mk_0}=1,\quad \sigma_{-\omega^mk_0}=-1.
\]
Set \(q_+=r_1(k_0)\), \(q_-=r_2(-k_0)\), and
\[
\begin{aligned}
g_+&=\widetilde\delta_{v1}(k_0)e^{2\chi_1(k_0)}
(\alpha_0\sqrt t)^{-2i\nu_1}e^{-t\Phi_{21}(\xi;k_0)},\\
g_-&=\widetilde\delta_{v4}(-k_0)e^{2\chi_4(-k_0)}
(\alpha_0\sqrt t)^{-2i\nu_4}e^{t\Phi_{21}(\xi;-k_0)},
\end{aligned}
\]
with \(G_+^{\rm PC}=\operatorname{diag}(g_+,1,1)\) and
\(G_-^{\rm PC}=\operatorname{diag}(1,g_-,1)\). Using the models of
Appendix~\ref{app:PC-model}, define \(M_{\pm k_0}^{\rm PC}\) by conjugation with
\(G_\pm^{\rm PC}\), and obtain the remaining four models by \(\mathbb Z_3\)
rotation. If \(M_1^{(\lambda)}\) denotes their first coefficient, put
\(\mathcal M_{\lambda,\rm raw}=\alpha_\lambda^{-1}M_1^{(\lambda)}\) and
\[
\mathcal P_\lambda(k)=D_I(\lambda)^{-1}
M_\lambda^{\rm PC}(\zeta_\lambda(k))D_I(\lambda).
\]
Set
\begin{equation}\label{3.27}
\mathcal M_\lambda=D_I(\lambda)^{-1}\mathcal M_{\lambda,\rm raw}D_I(\lambda),
\qquad
v_\lambda^{(0)}=D_I(\lambda)^{-1}v_{\lambda,\rm raw}^{(0)}D_I(\lambda).
\end{equation}
Then \(\mathcal P_{\lambda,+}=\mathcal P_{\lambda,-}v_\lambda^{(0)}\), and the
exact-to-model residual is
\begin{equation}\label{3.28}
J_\lambda^{\rm res}=\mathcal P_{\lambda,-}v^{(2)}\mathcal P_{\lambda,+}^{-1}
\quad\hbox{on }\Sigma_\lambda^{\rm loc}.
\end{equation}
The standard differentiated PC matching estimates give, for \(q=0,1\) and
\(p=1,2,\infty\),
\[
\|D_x^q(J_\lambda^{\rm res}-I_3)\|_{L^p}
\le C\log(2+t)\,t^{-1/2-1/(2p)},
\]
and on \(\partial\mathcal U_\lambda\),
\[
\sum_{q=0}^1\left\|D_x^q\left[\mathcal P_\lambda-I_3-
\frac{\mathcal M_\lambda}{\sqrt t\,(k-\lambda)}\right]\right\|_\infty
\le C\frac{\log(2+t)}t.
\]
Moreover, for \(0<|k-\lambda|<\rho_0\),
\begin{equation}\label{3.29}
\|\partial_x\mathcal P_\lambda(k)\|+\|\partial_x\mathcal P_\lambda(k)^{-1}\|
\le C\left[1+\frac{1+\log(2+t)+|\log|k-\lambda||}{t|k-\lambda|}\right].
\end{equation}
These estimates follow from the pulled-back endpoint expansion, the Jacobian,
and the uniform bounds of Appendix~\ref{app:PC-model}.

Orient \(\partial\mathcal U\) clockwise and solve the error problem for \(E_r\).

\begin{rhp}
	Find \(E_r(k)=E_r(k;x,t)\) satisfying:
	\begin{enumerate}
		\item[(i)]
		\(E_r\) is analytic in \(\mathbb C\setminus\Sigma^{(E)}\), $\Sigma^{(E)}=\Sigma^{(2)}\cup\partial\mathcal U$.
		
		\item[(ii)] The boundary values satisfy
		\[
		E_{r,+}(k)=E_{r,-}(k)v^{(E)}(k),\qquad k\in\Sigma^{(E)},
		\]
		where
		\[
		v^{(E)}(k)=\begin{cases}
			M^{(\mathrm{out)}}(k)v^{(2)}(k)\left( M^{(\mathrm{out)}}(k) \right) ^{-1},&		k\in \Sigma ^{(2)}\setminus \overline{\mathcal{U}},\\
			M^{(\mathrm{out})}(k)J_\lambda^{\rm res}(k)
			\left(M^{(\mathrm{out})}(k)\right)^{-1},
			&k\in\Sigma_\lambda^{\rm loc},\\
			M^{(\mathrm{LC})}(k)
			\left(M^{(\mathrm{out})}(k)\right)^{-1},
			&k\in\partial\mathcal U_\lambda ,
		\end{cases}
		\]
		
		\item[(iii)] As $k\to\infty$,
		\[
		E_r(k)\to I_3.
		\]
	\end{enumerate}
\end{rhp}
The matching circles, \eqref{3.28}, and the endpoint and omitted-pole jumps
make \eqref{3.6} solve the pure RH problem. Uniqueness gives the reductions,
unit jump determinants give \(\det E_r=1\), and retained residues are preserved.
At \(0\), the endpoint bounds and closed incident cycle glue the bounded row
\(N^{(3)}=N^{(2)}(M^{(\mathrm{out})})^{-1}E_r^{-1}\); Morera's theorem removes
its singularity.

Exterior decay, local matching, and \eqref{3.5} give
\[
\|v^{(E)}-I_3\|_{L^\infty}+\|D_x(v^{(E)}-I_3)\|_{L^\infty}=O\!\left(t^{-1/2}\log(2+t)\right),
\]
\[
\|v^{(E)}-I_3\|_{L^1\cap L^2}+\|D_x(v^{(E)}-I_3)\|_{L^1\cap L^2}=O(t^{-1/2}).
\]
For large \(t\), the Beals-Coifman equation is invertible and
\(\|\mu_E-I_3\|_2=O(t^{-1/2})\). Its Cauchy representation, the matching
expansion, and the residue theorem give
\[
\begin{aligned}
	E_r(k)&=I_3+\frac{E_{r,1}}{k}+O(k^{-2}),\\
	E_{r,1}&=\frac{1}{\sqrt{t}}\sum_{\lambda \in \mathcal{K} _0}{M^{(\mathrm{out)}}(\lambda )\mathcal{M} _{\lambda}\bigl( M^{(\mathrm{out)}}(\lambda ) \bigr) ^{-1}}+O\!\left( \frac{\log\mathrm{(}2+t)}{t} \right) .\\
\end{aligned}
\]
while, uniformly for \(\xi\in K\),
\begin{equation}\label{3.30}
	\partial _x\left[ E_{r,1}-\frac{1}{\sqrt{t}}\sum_{\lambda \in \mathcal{K} _0}{M^{(\mathrm{out)}}(\lambda )\mathcal{M} _{\lambda}\bigl( M^{(\mathrm{out)}}(\lambda ) \bigr) ^{-1}} \right] =O\left( \frac{\log t}{t} \right) .
\end{equation}
On the fixed pulled-back contour, \(\|D_x\mu_E\|_2=O(t^{-1/2})\). Replacing \(\mu_E\) by \(I_3\) costs \(O(t^{-1})\), and the matching remainder gives \eqref{3.30}.

\subsection{Analysis of the pure \texorpdfstring{$\bar\partial$}{d-bar}-problem}
The regular row \(N^{(3)}\) satisfies the pure \(\bar\partial\)-problem below.

\begin{rhp}
Find a continuous $1\times3$ row $N^{(3)}$ on $\mathbb C$ such that
\[
\bar\partial N^{(3)}=N^{(3)}W^{(3)},\qquad
N^{(3)}(k)\to\ell^{\rm T}\quad(k\to\infty),
\]
where $W^{(3)}=M_{\rm RHP}^{(2)}W^{(2)}(M_{\rm RHP}^{(2)})^{-1}$ and $W^{(2)}$ is given by \eqref{3.26}.
\end{rhp}

The pure \(\bar\partial\)-problem is equivalent to
\begin{equation}\label{3.31}
	N^{(3)}(k)=\ell^{\mathrm T}-\frac{1}{\pi}\iint_{\mathbb C}\frac{N^{(3)}(s)W^{(3)}(s)}{s-k}\,dA(s),
\end{equation}
or \((I-\mathcal S)N^{(3)}=\ell^{\mathrm T}\), where \(\mathcal S\) denotes the negative Cauchy-Green integral in \eqref{3.31}.

Since \(\bar\partial R^{(2)}=0\) in \(\mathcal U_0\), the source stays away from the singular endpoint. Equation~\eqref{3.6} and the outer, local, and error bounds give
\[
\mathop {\mathrm{sup}} \limits_{s\in \mathrm{supp}\bar{\partial}R^{(2)}}\left( \parallel M_{\mathrm{RHP}}^{(2)}(s)\parallel +\parallel (M_{\mathrm{RHP}}^{(2)}(s))^{-1}\parallel \right) \le C.
\]
Their first \(x\)-derivatives are bounded outside the stationary disks. Inside, the local factorization gives the same bound as \eqref{3.29}. Define \(\mathcal S_x\) by replacing \(W^{(3)}\) with \(\partial_xW^{(3)}\) in \(\mathcal S\).

On a representative lens, Proposition~\ref{prop:3.1}, \eqref{3.29}, and the phase signature bound the source by \((g_j(u)+|s-k_0|^{-1/2}+|\bar\partial\chi_{\mathcal Z}|)e^{-ct|u||v|}\), with \(g_j\) uniformly in \(L^2\). Splitting at \(|u|=t^{-1/2}\), also for the differentiated source, gives
\[
	\parallel W^{(3)}\parallel _{L^1}\le Ct^{-3/4},\qquad \parallel W^{(3)}\parallel _{L^{3/2}}\le Ct^{-5/12},\qquad \parallel W^{(3)}\parallel _{L^2}\le Ct^{-1/4},
\]
and
\[
	\parallel \partial _xW^{(3)}\parallel _{L^1}\le Ct^{-3/4},\qquad \parallel \partial _xW^{(3)}\parallel _{L^{6/5}}\le Ct^{-7/12}.
\]
Symmetry, Cauchy splitting, and planar Hardy-Littlewood-Sobolev give
\begin{equation}\label{3.32}
	\parallel \mathcal{S} \parallel _{L^{\infty}\rightarrow L^{\infty}}+\parallel \mathcal{S} \parallel _{L^3\rightarrow L^3}\le Ct^{-1/4},\qquad \parallel \mathcal{S} _xf\parallel _{L^3}\le Ct^{-7/12}\parallel f\parallel _{L^{\infty}}.
\end{equation}
Thus \(I-\mathcal S\) is invertible on \(L^\infty\cap L^3\) for large \(t\). As \(k\to\infty\),
\[
N^{(3)}(k)=\ell ^{\mathrm{T}}+\frac{N_{1}^{(3)}(x,t)}{k}+O(k^{-2}),\qquad N_{1}^{(3)}=\frac{1}{\pi}\iint_{\mathbb{C}}{N^{(3)}(s)W^{(3)}(s)\,dA(s).}
\]
Equations~\eqref{3.31}-\eqref{3.32}, the differentiated resolvent equation,
and \(L^3\)-\(L^{3/2}\) H\"older give, uniformly for \(\xi\in K\),
\[\|N_1^{(3)}\|+\|\partial_xN_1^{(3)}\|\le Ct^{-3/4}.\]

\subsection{Asymptotic analysis in Sectors I and II}
We now study the asymptotic behavior in Sectors I and II.
\begin{prop}\label{prop:3.4}
	Let \(u\) be a global classical real solution of \eqref{1.1} with \(u_0\in\mathcal S(\mathbb R)\). Fix \(0<\xi_-<\xi_+\) and \(I=[\xi_-,\xi_+]\). Assume Assumption~\ref{ass:1.1}{\rm(a)} and that the selected moving datum is regular on \(I\) in the sense of \eqref{3.22}. Then, uniformly for \(\xi=x/t\in I\),
	\begin{equation}\label{3.33}
		u(x,t)=u_{\rm sol}\left(x,t\mid\sigma_d(I;\xi)\right)+t^{-1/2}h+O(t^{-3/4}),\qquad t\to\infty .
	\end{equation}
	Here $u_{\rm sol}(x,t\mid\sigma_d(I;\xi))$ is the cone-dependent pure-discrete outer RH profile defined above, whereas $h$ is given by
	\begin{equation}\label{3.34}
		h=-\frac{1}{2}\partial _x\left[ \sum_{\lambda \in \mathcal{K} _0}{\ell ^{\mathrm{T}}M^{(\mathrm{out)}}(\lambda )\mathcal{M} _{\lambda}(x,t)\left( M^{(\mathrm{out)}}(\lambda ) \right) ^{-1}e_3} \right] .
	\end{equation}
\end{prop}

\begin{proof}
On a fixed exterior recovery ray, \(\mathcal T_I=D_I\), \(R^{(2)}=I_3\), and
\(N=N^{(3)}E_rM^{(\mathrm{out})}D_I^{-1}\Delta^{-1}\). Taking its
\(k^{-1}\)-coefficient and using \eqref{3.30} and the moment bound gives
\eqref{3.33}-\eqref{3.34}. The diagonal factors are lower order because
\(\partial_xD_I^{[1]}=0\) and \(\partial_x(\Delta^{[1]})_{33}=O(t^{-1})\).
\end{proof}

For unbounded positive velocity assume
\[
\xi=\frac{x}{t}\ge M, \qquad 0<\zeta=\frac{t}{x}\le M^{-1}.
\]
and set
\[
t\Phi_{21}(\xi;k)=x\widetilde\Phi_{21}(\zeta;k),\qquad\widetilde\Phi_{21}(\zeta;k)=i\sqrt3(9\zeta k^5-k),
\]
and
\[
k_0=(45\zeta)^{-1/4},\qquad\mathfrak T=tk_0^3=\frac{x}{45k_0},\qquad tk_0^5=\frac{xk_0}{45}.
\]
For \(x/t\ge H>0\), \(k_0\ge k_H=(H/45)^{1/4}\) and
\(\mathfrak T\ge k_H^3t\). Near each saddle use the parabolic-cylinder
coordinate
\[
z=2\,3^{5/4}\sqrt{5t}\,k_0^{3/2}(k-k_0),
\]
with the other coordinates obtained by symmetry. Taylor expansion of the phase
has leading term \(iz^2/2\), while the remaining terms are uniformly smaller by
powers of \((xk_0)^{-1/2}\). Away from fixed relative neighborhoods of the six
saddles, the lens can be chosen so that the real part of the relevant phase is
bounded above by \(-c_\delta x\,|\operatorname{Im}(\omega^{-m}k)|\).

\begin{lemma}\label{lem:3.2}
Let \(H>0\). For \(k_0\ge k_H\), the extensions in Proposition~\ref{prop:3.1}
may be chosen \(C^1\) in \(x\) in the spaces of \eqref{3.4}. For every
\(N\ge1\), there is \(\epsilon_N(k_0)\le C_{N,H}\langle k_0\rangle^{-N}\)
such that on the relative stationary bands
\[
|\bar\partial\mathfrak R_j^\pm(k)|\le
C_{N,H}\epsilon_N(k_0)\bigl(g_{j,k_0}(s_m)+d_j(k)^{-1/2}\bigr),
\qquad s_m=\operatorname{Re}(\omega^{-m}k),
\]
where \(\sup_{k_0\ge k_H}\|\langle s\rangle g_{j,k_0}\|_2<\infty\).
On the complementary region the bounds of Proposition~\ref{prop:3.1} are
uniform. Moreover, with
\[
b_+(k_0)=\widetilde\delta_{v1}(k_0)e^{2\chi_1(k_0)},\qquad
b_-(k_0)=\widetilde\delta_{v4}(-k_0)e^{2\chi_4(-k_0)},
\]
one has \(b_\pm\in C^1([k_H,\infty))\) and
\[
\sup_{k_0\ge k_H}\bigl(|b_\pm|+|b_\pm|^{-1}+|\partial_{k_0}b_\pm|\bigr)
\le C_H(u_0).
\]
\end{lemma}

\begin{proof}
Weighted Volterra estimates give rapidly decaying \(H^2\) tails for the
reflection data on the relative stationary bands. Angular interpolation and
endpoint Cauchy estimates then give the stated \(\bar\partial\)-bounds.
Differentiating the endpoint factors and the scaled cutoffs proves the
\(C^1\)-bounds, while the remaining singular term is integrable in the spaces
of \eqref{3.4}.
\end{proof}

The Gaussian decay near the saddles, the nonstationary phase bound away from
them, and Lemma~\ref{lem:3.2} give, uniformly for \(x/t\ge H\),
\[
\|N_1^{(3)}\|+\|\partial_xN_1^{(3)}\|
\le C\epsilon_N(k_0)\mathfrak T^{-3/4},
\]
while the local parabolic-cylinder matching, including one \(x\)-derivative,
contributes
\[
O\!\left(\frac{\log(2+tk_0^5)}{tk_0^3}\right).
\]
Translation and dilation reduce the local crosses to a fixed
Lipschitz-Carleson family, so all constants are uniform in \(k_0\).
Let \(M_M^{\rm mov}(k;x,t)\) be the half-cone moving outer model from Lemma~\ref{lem:6.1}.  Define
\begin{equation}\label{3.35}
	h_M=-\frac{1}{2}\partial _x\left[ \sum_{\lambda \in \mathcal{K} _0}{\ell ^{\mathrm{T}}M_{M}^{\mathrm{mov}}(\lambda ;x,t)\mathcal{M} _{\lambda}(x,t)\bigl( M_{M}^{\mathrm{mov}}(\lambda ;x,t) \bigr) ^{-1}e_3} \right] .
\end{equation}

\begin{prop}\label{prop:3.5}
	Fix \(M>1\) and choose \(H>\max\{M,V_*\}\). Assume the hypotheses of Proposition~\ref{prop:3.4} on \(I_{M,H}=[M,2H]\) and those of Lemma~\ref{lem:6.1}.  Then, uniformly for \(x/t\ge M\),
	\[
	u(x,t)=-\frac12\partial_x\bigl(\ell^{\rm T}M_{M,1}^{\rm mov}e_3\bigr)+t^{-1/2}h_M+O(t^{-3/4}).
	\]
	Moreover, for some \(T_M,C_M>0\), we have
	\[
	\sup_{\substack{t\ge T_M\\x/t\ge M}}|h_M(x,t)|\le C_M.
	\]
	On the far tail \(x/t\ge H\), the remainder before simplification is bounded by
	\[
	C\left\{\epsilon_N(k_0)\mathfrak T^{-3/4}+\frac{\log(2+tk_0^5)}{tk_0^3}+t^{-1}+e^{-\mu_Ht}\right\},
	\]
	where, if the discrete spectrum is nonempty,
	\[
	\mu _H=\frac{\sqrt{3}}{2}\min_{\lambda \in \mathcal{Z} _1\cup \mathcal{Z} _4} |\mathrm{Im}\lambda |\,(H-v_{\lambda})>0,
	\]
	and \(\mu_H=+\infty\) otherwise.
\end{prop}

\begin{proof}
Proposition~\ref{prop:3.4} treats \(M\le x/t\le2H\). On \(x/t\ge H\), the
preceding uniform \(\bar\partial\)- and local-model estimates give the stated
remainder, and Schwartz decay makes \(h_M\) uniformly bounded. Since
\(\mathfrak T\ge k_H^3t\) and
\(\log(2+tk_0^5)/(tk_0^3)=O_H(t^{-1}\log(2+t))\), the remainder is
\(O(t^{-3/4})\). The overlap follows from the identical extensions and
Lemma~\ref{lem:6.1}.
\end{proof}
\section{The mSK equation and the Painlev\'e region}\label{sec4}
For the regular mSK equation, let $w$ be a global real classical solution of \eqref{1.2} with initial datum $w_0=w(\cdot,0)$ and $w(\cdot,t)\in\mathcal S(\mathbb R)$. A tilde denotes the mSK scattering quantities. We first specify the data and coefficients used in Theorem~\ref{th2}.

The representative discrete sets are \(\widetilde{\mathcal Z}_1=\{\widetilde\kappa_j\}\) and
\(\widetilde{\mathcal Z}_4=\{\widetilde\eta_\ell\}\), with norming constants \(\widetilde c_j,\widetilde d_\ell\), and \(\widetilde{\mathcal Z}\) denotes
the union of their complete reduction orbits. For \(\lambda\in\widetilde{\mathcal Z}_1\cup\widetilde{\mathcal Z}_4\), set
\[
v_{\lambda}=9\frac{\operatorname{Im}(\lambda^5)}{\operatorname{Im}\lambda}.
\]
Let \(\widetilde\delta_j\), \(j=1,\ldots,6\), be the scalar conjugating functions of Section~\ref{sec3} associated with \(\widetilde r_1,\widetilde r_2\), and define
\[
\widetilde{\mathfrak{t} }_1=\frac{\widetilde{\delta }_3\widetilde{\delta }_{4}^{2}\widetilde{\delta }_5}{\widetilde{\delta }_6\widetilde{\delta }_2\widetilde{\delta }_{1}^{2}}, \qquad
\widetilde{\mathfrak{t} }_4=\frac{\widetilde{\delta }_{1}^{2}\widetilde{\delta }_2\widetilde{\delta }_6}{\widetilde{\delta }_5\widetilde{\delta }_3\widetilde{\delta }_{4}^{2}}.
\]
For \(K\in\{K_M,I_M\}\), retain only representatives with \(v_\lambda\in K\) and their complete orbits. Let \(\widetilde D_{>K}=I_3\) if \(K=K_M\). If \(K=I_M\), let \(\widetilde D_{>K}\) be the product of the factors \eqref{3.15}-\eqref{3.16} over the representatives with velocity greater than \(M\). Define
\begin{equation}\label{4.1}
\begin{aligned}
\widetilde{c}_{j}^{K}(\xi )
&=\widetilde{c}_j\widetilde{\mathfrak{t} }_1(\widetilde{\kappa }_j;\xi )
\frac{(\widetilde{D}_{>K})_2(\widetilde{\kappa }_j)}{(\widetilde{D}_{>K})_1(\widetilde{\kappa }_j)},&
\widetilde{d}_{\ell}^{K}(\xi )
&=\widetilde{d}_{\ell}\widetilde{\mathfrak{t} }_4(\widetilde{\eta }_{\ell};\xi )
\frac{(\widetilde{D}_{>K})_1(\widetilde{\eta }_{\ell})}{(\widetilde{D}_{>K})_2(\widetilde{\eta }_{\ell})},\\
\widehat{\widetilde c}_j(K)
&=\widetilde c_j^K(v_{\widetilde\kappa_j}),&
\widehat{\widetilde d}_\ell(K)
&=\widetilde d_\ell^K(v_{\widetilde\eta_\ell}).
\end{aligned}
\end{equation}
The first line of \eqref{4.1}, together with the complete reduction orbits, defines \(\widetilde\sigma_d(K;\xi)\). The frozen constants in the second line and their complete orbits define \(\widehat{\widetilde\sigma}(K)\).

For \(J_M\), retain the representatives with \(v_\lambda\in J_M\) and their complete orbits. Let \(\widetilde D_{>J_M}\) be the mSK counterpart of \eqref{5.2}, and let \(\widetilde\Delta_{\rm IV}\) be the regular mSK Plemelj factor \(\Delta_{\rm P}\) introduced below, constructed from \(1-|\widetilde r_1|^2\) and \(1-|\widetilde r_2|^2\). Define
\begin{equation}\label{4.2}
\begin{aligned}
\widehat{\widetilde{c}}_j(J_M)
&=\widetilde{c}_j
\frac{(\widetilde{\Delta }_{\mathrm{IV}})_1(\widetilde{\kappa }_j)}{(\widetilde{\Delta }_{\mathrm{IV}})_2(\widetilde{\kappa }_j)}
\frac{(\widetilde{D}_{>J_M})_2(\widetilde{\kappa }_j)}{(\widetilde{D}_{>J_M})_1(\widetilde{\kappa }_j)},\\
\widehat{\widetilde{d}}_{\ell}(J_M)
&=\widetilde{d}_{\ell}
\frac{(\widetilde{\Delta }_{\mathrm{IV}})_2(\widetilde{\eta }_{\ell})}{(\widetilde{\Delta }_{\mathrm{IV}})_1(\widetilde{\eta }_{\ell})}
\frac{(\widetilde{D}_{>J_M})_1(\widetilde{\eta }_{\ell})}{(\widetilde{D}_{>J_M})_2(\widetilde{\eta }_{\ell})}.
\end{aligned}
\end{equation}
These constants and their complete orbits define \(\widehat{\widetilde\sigma}(J_M)\).

For \(K\in\{K_M,I_M,J_M\}\), let \(m_{\rm sol}\) denote the corresponding reflectionless model. For \(K\in\{K_M,I_M\}\), let \(m_K^{\rm mov}\) denote the moving rational outer model. If \(m_{{\rm sol},1}\) is the \(k^{-1}\)-coefficient at infinity, define
\begin{equation}\label{4.3}
w_{\rm sol}(x,t;\widehat{\widetilde\sigma}(K))
=3\bigl(m_{{\rm sol},1}(x,t;\widehat{\widetilde\sigma}(K))\bigr)_{13}.
\end{equation}
For \(\xi>0\), set
\[
k_0=\left(\frac{\xi}{45}\right)^{1/4},\qquad
\mathcal K_0=\{\pm\omega^m k_0:m=0,1,2\}.
\]
For \(K\in\{K_M,I_M\}\), let \(\widetilde{\mathcal M}_{\lambda,K}\) denote the tilded counterpart of the parabolic-cylinder coefficient in \eqref{3.27}, and define
\begin{equation}\label{4.4}
\widetilde h_K
=3\left[
\sum_{\lambda\in\mathcal K_0}
 m_K^{\rm mov}(\lambda;x,t)
 \widetilde{\mathcal M}_{\lambda,K}(x,t)
 \bigl(m_K^{\rm mov}(\lambda;x,t)\bigr)^{-1}
\right]_{13}.
\end{equation}

\begin{assumption}\label{ass:1.2}
	\begin{enumerate}
		\item[(a)] \emph{Generic scattering data.}
		The datum \(w_0\in\mathcal S(\mathbb R)\), RH Problem~\ref{rhp2} is
		uniquely solvable, and the nonzero zeros form finitely many
		normal-simple, exclusive-zero, pairwise disjoint complete orbits
		separated from \(\Sigma\cup\{0\}\). There are no spectral singularities on \(\Sigma\).
		
\item[(b)] \emph{Positive-sector outer models and endpoint frames.}
For $K\in\{K_M,I_M\}$, the moving and frozen balanced residue systems, including equal-velocity groups, are uniformly invertible in common single-valued $C^1$ charts. The zero-reflection RH problems are uniquely solvable. Their outer matrices, inverses, and first fixed-$t$ $x$-derivatives are uniformly bounded off the pole disks, and
\[
\sup_{\xi\in K}\sum_{q=0}^1
\left\|D_x^q\bigl[m_{K,1}^{\rm mov}(x,t)
-m_{{\rm sol},1}(x,t;\widehat{\widetilde\sigma}(K))\bigr]\right\|
\le C_K\frac{\log(2+t)}{t}.
\]
Here $D_x$ is the fixed-$t$ total derivative after $\xi=x/t$, and the bound is uniform on the whole half-cone when $K=K_M$.
With $a=\widetilde r_1(0+)$ and $b=\widetilde r_2(0-)$, assume $\bar a+b+a\bar b=0$. There are reduction-compatible, sectorwise analytic and invertible endpoint frames $\widetilde F_0=\widetilde F_6$ such that, for $|k|\le t^{-1/2}$,
\[
\begin{aligned}
\sup_n(\|\widetilde F_n\|+\|\widetilde F_n^{-1}\|)&\le C,\\
\widetilde F_{n-1}\widetilde v_n^{(1)}\widetilde F_n^{-1}&=I_3+O(k),\\
\partial_x(\widetilde F_{n-1}\widetilde v_n^{(1)}\widetilde F_n^{-1}-I_3)&=O(k^2+k/t).
\end{aligned}
\]
The same undifferentiated $O(k)$ matching holds on $|k|=t^{-1/2}$, and $\widetilde v_n^{(1)}$ is the exact post-conjugation jump in the $n$th sector.

		\item[(c)] \emph{Negative-sector outer model.}
		For some \(t_0>0\), the zero-reflection RH problem for the datum
		\(\widehat{\widetilde\sigma}(J_M)\) defined by \eqref{4.2}, together
		with its balanced residue systems, is uniformly solvable in
		\(x/t\in J_M\), with bounded outer matrices
		and inverses off the pole disks.
	\end{enumerate}
\end{assumption}

Applying the deformations and estimates of Sections~\ref{sec3}-\ref{sec5} to RH Problem~\ref{rhp2}, and using the reconstruction $w=3(m_1)_{13}$, gives the following theorem by analogous arguments. The Painlev\'e case is proved below. The regions are shown in Figure~\ref{fig:mSK-asymptotic-regions}.

\begin{theorem}[Long-time asymptotics of the mSK equation]\label{th2}
Let \(w\) be the solution above, and suppose that Assumption~\ref{ass:1.2} holds. Then, as \(t\to+\infty\):
	\begin{enumerate}
\item[(i)-(ii)] \emph{Sectors I and II.}
For $K=K_M$ or $K=I_M$, respectively, uniformly for $\xi=x/t\in K$,
\[
w(x,t)=w_{\rm sol}(x,t;\widehat{\widetilde\sigma}(K))
+t^{-1/2}\widetilde h_K+O(t^{-3/4}),
\]
with the frozen data, soliton profile, and radiation coefficient defined in \eqref{4.1}, \eqref{4.3}, and \eqref{4.4}.

		\item[(iii)] \emph{Sector III.}
		Assume \eqref{4.8} and \(\kappa_M<1\), where \(\kappa_M\) is defined by
		\eqref{4.14}. Then, uniformly for
		\(s=x/(5t)^{1/5}\in[-M,M]\),
		\[
		\begin{aligned}
			w(x,t)&=(5t)^{-1/5}p(s)+O((5t)^{-2/5}),\\
			w_x(x,t)&=(5t)^{-2/5}p'(s)+O((5t)^{-3/5}).
		\end{aligned}
		\]
where \(p\) is the real profile defined by \eqref{4.23} and satisfying \eqref{4.25}.
		
		\item[(iv)] \emph{Sector IV.}
		For every \(N\ge1\), there exists \(T_{N,M}\ge t_0\) such
		that uniformly for \(t\ge T_{N,M}\) and \(\xi\in J_M\),
		\[
		w(x,t)=w_{\rm sol}(x,t;\widehat{\widetilde\sigma}(J_M))+O\bigl((|x|+t)^{-N}\bigr),
		\]
		where \(w_{\rm sol}(x,t;\widehat{\widetilde\sigma}(J_M))\) and
		\(\widehat{\widetilde\sigma}(J_M)\) are defined in \eqref{4.3} and
		\eqref{4.2}, respectively.
	\end{enumerate}
\end{theorem}

\begin{figure}[htbp]
	\centering
	\begin{tikzpicture}[
		x=0.85cm,
		y=0.75cm,
		>=stealth,
		every node/.style={font=\scriptsize},
		line cap=round,
		line join=round
		]
		
		\def\a{1.38}
		\def\b{0.80}
		
		\coordinate (Lcurve) at ({\a*cos(142)},{\b*sin(142)});
		\coordinate (Rcurve) at ({\a*cos(38)},{\b*sin(38)});
		\coordinate (Pvtwo) at ({\a*cos(10)},{\b*sin(10)});
		
		\coordinate (Ltop) at (-2.55,6.25);
		\coordinate (Rtop) at ( 2.20,6.25);
		\coordinate (VtwoT) at (7.00,2.35);
		
		\draw[->,densely dashed,black!55,line width=0.65pt]
		(-6.45,0)--(7.35,0)
		node[below,text=black] {$x$};
		
		\draw[->,densely dashed,black!55,line width=0.65pt]
		(0,0)--(0,6.95)
		node[left,text=black] {$t$};
		
		\draw[black!72,line width=0.82pt]
		plot[
		domain=180:0,
		samples=140,
		smooth,
		variable=\ang
		]
		({\a*cos(\ang)},{\b*sin(\ang)});
		
		\draw[black!72,line width=0.82pt]
		(Lcurve)
		.. controls
		(
		{\a*cos(142)+0.58*(-2.55-\a*cos(142))},
		{\b*sin(142)+(1/3)*(6.25-\b*sin(142))}
		)
		and
		(
		{-2.55},
		{6.25-(1/3)*(6.25-\b*sin(142))}
		)
		.. (Ltop);
		
		\draw[black!72,line width=0.82pt]
		(Rcurve)
		.. controls
		(
		{\a*cos(38)+0.58*(2.20-\a*cos(38))},
		{\b*sin(38)+(1/3)*(6.25-\b*sin(38))}
		)
		and
		(
		{2.20},
		{6.25-(1/3)*(6.25-\b*sin(38))}
		)
		.. (Rtop);
		
		\draw[black!72,line width=0.82pt]
		(Pvtwo)--(VtwoT);
		
		\node[
		above right,
		font=\tiny,
		text=black!75,
		xshift=-15pt,
		yshift=0pt
		] at (Ltop)
		{$x\sim -(5t)^{1/5}$};
		
		\node[
		above right,
		font=\tiny,
		text=black!75,
		xshift=-15pt,
		yshift=0pt
		] at (Rtop)
		{$x\sim (5t)^{1/5}$};
		
		\node[
		align=center,
		text=violet!65!black
		] at (-4.55,3.55)
		{Sector IV};
		
		\node[
		align=center,
		text=green!45!black
		] at (-0.25,3.55)
		{Sector III};
		
		\node[
		align=center,
		text=brown!70!black
		] at (4.95,3.45)
		{Sector II};
		
		\node[
		align=center,
		text=blue!60!black
		] at (5.85,0.85)
		{Sector I};
		
		\node[below,text=black] at (0,0) {$0$};
		
	\end{tikzpicture}
	
	\caption{The asymptotic sectors I-IV for the mSK equation in the $(x,t)$-half plane.}
	\label{fig:mSK-asymptotic-regions}
\end{figure}
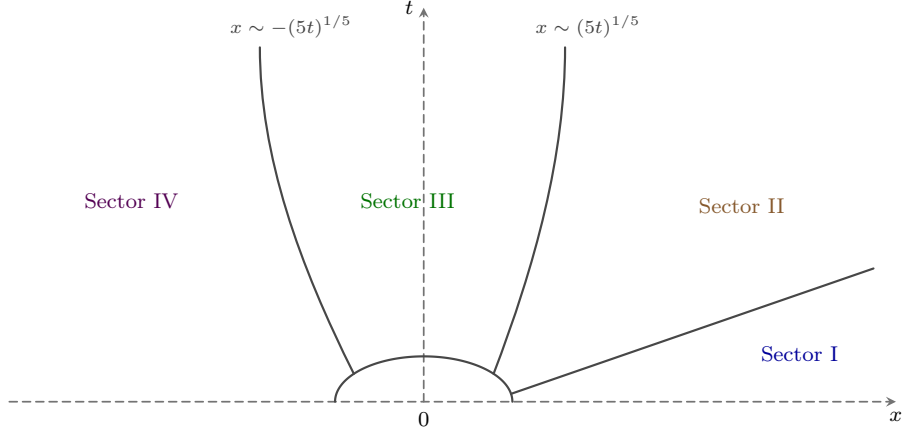

\subsection{Reduction to the Painlev\'e model}
We work in the Painlev\'e region
\[
\mathcal{P} =\left\{ (x,t):s=\frac{x}{(5t)^{1/5}}, |s|\le M, t\ge t_0 \right\} ,\qquad \lambda =(5t)^{1/5}k.
\]
Because the SK problem is singular at \(0\), analyze the regular mSK problem and recover SK by \eqref{4.5}. The saddles \(k^4=x/(45t)\) remain at finite \(\lambda\)-distance and enter one origin model. Under the scaling,
\[
\theta _{ij}\bigl( x,t;(5t)^{-1/5}\lambda \bigr) =\Theta _{ij}(\lambda ;s),\qquad \Theta _{ij}=(\omega ^i-\omega ^j)s\lambda +\frac{9}{5}(\omega ^{2i}-\omega ^{2j})\lambda ^5,
\]
and in particular
\[
\Theta_{21}(\lambda;s)=i\sqrt3\left(\frac95\lambda^5-s\lambda\right).
\]
The coalescence and rescaling of the stationary points are illustrated in Figure~\ref{fig:SK-Painleve-scaling}.

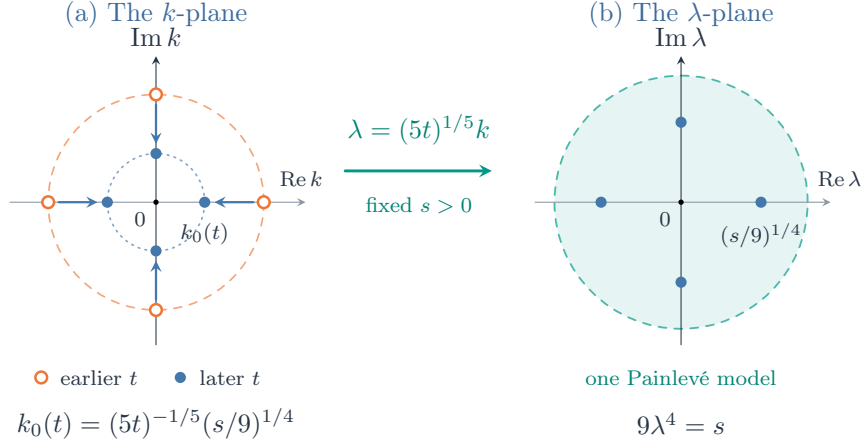
\begin{figure}[htbp]
\centering
\definecolor{SKPainBlue}{HTML}{4477AA}
\definecolor{SKPainOrange}{HTML}{EE7733}
\definecolor{SKPainTeal}{HTML}{009988}
\definecolor{SKPainAxis}{HTML}{8A96A4}
\definecolor{SKPainInk}{HTML}{263445}
\begin{tikzpicture}[x=0.92cm,y=0.92cm,>=stealth,line cap=round,line join=round,font=\small,text=SKPainInk]
  \begin{scope}
    \node[text=SKPainBlue!90!black] at (0,2.70) {(a) The $k$-plane};
    \draw[SKPainAxis,->] (-2.10,0)--(2.15,0);
    \node[above left,font=\scriptsize,text=SKPainInk] at (2.54,0.08) {$\operatorname{Re}k$};
    \draw[SKPainInk,->] (0,-2.05)--(0,2.10) node[above] {$\operatorname{Im}k$};
    \draw[SKPainOrange!65,dashed,line width=0.65pt] (0,0) circle (1.55);
    \draw[SKPainBlue!65,dotted,line width=0.65pt] (0,0) circle (0.70);
    \foreach \ang in {0,90,180,270}{
      \draw[SKPainOrange,fill=white,line width=0.95pt] (\ang:1.55) circle (2.2pt);
      \draw[SKPainBlue,->,line width=0.85pt] (\ang:1.40)--(\ang:0.84);
      \fill[SKPainBlue] (\ang:0.70) circle (2.1pt);
    }
    \fill (0,0) circle (1.0pt);
    \node[below left,font=\scriptsize] at (0,0) {$0$};
    \node[below=5pt,font=\scriptsize,inner sep=1pt] at (0.70,-0.06) {$k_0(t)$};
    \draw[SKPainOrange,fill=white,line width=0.95pt] (-1.65,-2.52) circle (2.2pt);
    \node[right,font=\scriptsize] at (-1.52,-2.52) {earlier $t$};
    \fill[SKPainBlue] (0.35,-2.52) circle (2.1pt);
    \node[right,font=\scriptsize] at (0.48,-2.52) {later $t$};
    \node at (0,-3.18) {$k_0(t)=(5t)^{-1/5}(s/9)^{1/4}$};
  \end{scope}
  \draw[SKPainTeal,->,line width=1.10pt] (2.72,0.45)--(4.83,0.45);
  \node[above,align=center,font=\small,text=SKPainTeal!85!black] at (3.78,0.68) {$\lambda=(5t)^{1/5}k$};
  \node[below,font=\scriptsize,text=SKPainTeal!85!black] at (3.78,0.22) {fixed $s>0$};
  \begin{scope}[shift={(7.55,0)}]
    \node[text=SKPainBlue!90!black] at (0,2.70) {(b) The $\lambda$-plane};
    \fill[SKPainTeal!10] (0,0) circle (1.82);
    \draw[SKPainTeal!70,dashed,line width=0.70pt] (0,0) circle (1.82);
    \draw[SKPainAxis,->] (-2.10,0)--(2.15,0);
    \node[above left,font=\scriptsize,text=SKPainInk] at (2.74,0.08) {$\operatorname{Re}\lambda$};
    \draw[SKPainInk,->] (0,-2.05)--(0,2.10) node[above] {$\operatorname{Im}\lambda$};
    \foreach \ang in {0,90,180,270}{
      \fill[SKPainBlue] (\ang:1.15) circle (2.1pt);
    }
    \fill (0,0) circle (1.0pt);
    \node[below left,font=\scriptsize] at (0,0) {$0$};
    \node[below=5pt,font=\scriptsize,inner sep=1pt] at (1.15,-0.06) {$(s/9)^{1/4}$};
    \node[font=\scriptsize,text=SKPainTeal!85!black] at (0,-2.52) {one Painlev\'e model};
    \node at (0,-3.18) {$9\lambda^4=s$};
  \end{scope}
\end{tikzpicture}
\caption{Stationary-point coalescence and Painlev\'e scaling.}
\label{fig:SK-Painleve-scaling}
\end{figure}

For $\lambda=u+iv$ with $u>0$ and $0<v<u\tan(\pi/12)$,
\[
	\mathrm{Re}\Theta _{21}(u+iv;s)=-\sqrt{3}\,v\left( 9u^4-18u^2v^2+\frac{9}{5}v^4-s \right) .
\]

Positivity of \(9-18q^2+\frac95q^4\) on \(0\le q\le\tan(\pi/12)\) gives, uniformly for \(|s|\le M\),
\[
\mathrm{Re}\Theta _{21}(u+iv;s)\le C_M-cu^4v,\qquad |e^{\Theta _{21}(u+iv;s)}|\le C_Me^{-cu^4v}.
\]
The remaining estimates follow from the $\mathbb Z_3$ and Schwarz symmetries.

At the spectral endpoint write
\[
G(k)=D(k)P,\qquad D(k)=\mathrm{diag}\mathrm{(}1,k,k^2),\qquad P=\left( \begin{matrix}
	\omega&		\omega ^2&		1\\
	\omega ^2&		\omega&		1\\
	1&		1&		1\\
\end{matrix} \right) .
\]
Since the mSK perturbation is diagonal and commutes with $D(k)$,
\[
V_1(x;k)=G(k)^{-1}\operatorname{diag}(0,-w_0,w_0)G(k)=w_0(x)A_0,
\]
where
\[
A_0=\frac{1}{3}\left( \begin{matrix}
	0&		1-\omega ^2&		1-\omega\\
	1-\omega&		0&		1-\omega ^2\\
	1-\omega ^2&		1-\omega&		0\\
\end{matrix} \right) .
\]
Thus \(V_1\) is \(k\)-independent and the apparent singularity at \(0\) is removable.

The equation $m_x-[k\Lambda,m]=wA_0m$ gives $[m_1,\Lambda]=wA_0$, hence $w=3(m_1)_{13}$, consistently with RH Problem~\ref{rhp2}.
Recover SK by the Miura map
\begin{equation}\label{4.5}
	u(x,t)=\frac16\bigl(w_x(x,t)-w(x,t)^2\bigr).
\end{equation}
Direct substitution verifies \eqref{4.5}, with initial datum
\[
u_0=\frac16(w_0'-w_0^2).
\]

\begin{prop}\label{prop:4.1}
	Let $w_0\in\mathcal S(\mathbb R)$. The relevant mSK Jost functions, $\widetilde s,\widetilde s^A$, and their first two $k$-derivatives extend continuously to $k=0$ along the corresponding closed spectral half-rays. If
	\begin{equation}\label{4.6}
		\widetilde s_{11}(k)\ne0\quad(k\ge0),\qquad \widetilde s_{11}^A(k)\ne0\quad(k\le0),
	\end{equation}
	then
	\[
	\widetilde r_1\in C^2([0,\infty)),\qquad \widetilde r_2\in C^2((-\infty,0]),
	\]
	with
	\[
	\parallel \widetilde{r}_j\parallel _{L^{\infty}}+\parallel \widetilde{r}_j\prime\parallel _{L^{\infty}}+\parallel \widetilde{r}_j''\parallel _{L^{\infty}}<\infty ,\qquad j=1,2,
	\]
	and their derivatives decay rapidly at infinity.
\end{prop}

\begin{proof}
	Since \(V_1=w_0A_0\), differentiated Volterra equations converge uniformly at \(0\), giving the one-sided \(C^2\) extensions. Moreover,
	\[
	\widetilde{s}(0)=P^{-1}\mathrm{diag}\mathrm{(}1,e^{W_0},e^{-W_0})P,\qquad \widetilde{s}_{11}(0)=\widetilde{s}_{11}^{A}(0)=\frac{1+e^{W_0}+e^{-W_0}}{3}\ge 1,
	\]
	where \(W_0=\int_{\mathbb R}w_0\). Condition~\eqref{4.6}, continuity, and the large-\(k\) limits bound the denominators away from zero. Quotient and Volterra estimates give the result.
\end{proof}

Set $A=e^{W_0}$, $d_0=1+A+A^{-1}$, then
\begin{equation}\label{4.7}
	\widetilde{r}_1(0)=\frac{\omega +A\omega ^2+A^{-1}}{d_0},\qquad \widetilde{r}_2(0)=\frac{\omega ^2+A^{-1}\omega +A}{d_0},
\end{equation}
where the second limit is taken from $\mathbb R_-$.

For $\Sigma_n=e^{(n-1)\pi i/3}\mathbb R_+$, the mSK jumps factor as
\[
\begin{aligned}
	\widetilde v_1(r)
	&=\bigl(I+\overline{\widetilde r_1(r)}
	e^{\theta_{21}}E_{21}\bigr)
	\bigl(I-\widetilde r_1(r)e^{-\theta_{21}}E_{12}\bigr),
	\\
	\widetilde v_2(e^{i\pi/3}r)
	&=\bigl(I-\overline{\widetilde r_2(-r)}
	e^{-\theta_{32}}E_{23}\bigr)
	\bigl(I+\widetilde r_2(-r)e^{\theta_{32}}E_{32}\bigr),\\
	\widetilde v_{2m+1}(\omega^m r)
	&=\mathcal A^{-m}\widetilde v_1(r)\mathcal A^m,\\
	\widetilde v_{2m+2}(e^{i(2m+1)\pi/3}r)
	&=\mathcal A^{-m}\widetilde v_2(e^{i\pi/3}r)\mathcal A^m,
	\quad m=0,1,2,\ r>0.
\end{aligned}
\]
Proposition~\ref{prop:4.1} gives bounded first derivatives on the closed rays and rapid decay.

To remove nonresonant orbits, set
\[
\rho _{\tilde{\kappa}_j}=\mathrm{Re(}-9i\sqrt{3}\,\tilde{\kappa}_{j}^{5}),\qquad \rho _{\tilde{\eta}_{\ell}}=\mathrm{Re(}9i\sqrt{3}\,\tilde{\eta}_{\ell}^{5}),
\]
and assume
\begin{equation}\label{4.8}
	\rho _*=\min \left( \{|\rho _{\tilde{\kappa}_j}|\}_j\cup \{|\rho _{\tilde{\eta}_{\ell}}|\}_{\ell} \right) >0.
\end{equation}
If the nonzero discrete spectrum is empty, set $\rho_*=+\infty$.

\begin{lemma}\label{lem:4.1}
	For some $\alpha$ sufficiently close to $\pi/12$, the reduction-compatible contour
	\[
		\Sigma ^{\mathrm{P}}(\alpha )=\bigcup_{m=0}^5{\left( e^{i(m\pi /3+\alpha )}\mathbb{R} _+\cup e^{i((m+1)\pi/3-\alpha )}\mathbb{R} _+ \right)},
	\]
	satisfies
	\begin{equation}\label{4.9}
		d_{\mathrm{P}}:=\mathrm{dist}\!\left( \widetilde{\mathcal{Z} },\Sigma \cup \Sigma ^{\mathrm{P}}(\alpha )\cup \{0\} \right) >0,
	\end{equation}
	and all model exponentials retain their strict decay signs.
\end{lemma}

\begin{proof}
	Strict decay persists near \(\alpha=\pi/12\). Excluding the finitely many angles whose rays meet \(\widetilde{\mathcal Z}\) gives \eqref{4.9}. For empty spectrum set \(d_{\rm P}=+\infty\).
\end{proof}

\begin{lemma}\label{lem:4.2}
	Let $\alpha_0,\alpha_1$ belong to the admissible neighborhood in Lemma~\ref{lem:4.1}. The corresponding normalized model RH problems are bijectively equivalent in the local Smirnov $E^2$ class at the common vertex and have the same coefficient $M_1^{\rm P}$ at infinity.
\end{lemma}

\begin{proof}
	Move each ray through a disjoint swept sector, multiplying by \(V^{-1}\), \(V=e^{\operatorname{ad}H}S\), on that sector. This transfers the jump. Since \(V^{\pm1}-I_3=O(e^{-c|\lambda|^5})\) at infinity and \(V^{\pm1}\) is bounded at \(0\), the normalization, \(M_1^{\rm P}\), and local Smirnov class are preserved. Reverse moves give the inverse.
\end{proof}

Fix such an \(\alpha\) and write \(\Sigma^{\rm P}=\Sigma^{\rm P}(\alpha)\). By Lemma~\ref{lem:4.2}, the canonical deformation preserves \(M_1^{\rm P}\).
For the orbits with positive growth rate, introduce the corresponding diagonal divisors and set
\[
D_{\rm P}(k)
=
\prod_{\rho_{\tilde\kappa_j}>0}D_{\tilde\kappa_j}(k)
\prod_{\rho_{\tilde\eta_\ell}>0}D_{\tilde\eta_\ell}(k)
=
\operatorname{diag}\bigl(\pi_1(k),\pi_2(k),\pi_3(k)\bigr).
\]
Then
\[
\det D_{\mathrm{P}}=1,\qquad D_{\mathrm{P}}(0)=I_3,\qquad D_{\mathrm{P}}(k)=I_3+\frac{D_{1}^{\mathrm{P}}}{k}+O(k^{-2}),
\]
where $D_1^{\rm P}$ is diagonal and independent of $x,t$.

Let \(m^{(0)}=mD_{\rm P}\). Divisor switching conjugates decaying residues and reverses growing channels. Nonresonance and \(x=O(t^{1/5})\) give
\[
	\max_{\zeta \in \widetilde{\mathcal{Z} }} \left( \parallel C_{\zeta}^{(0)}\parallel +\parallel \partial _xC_{\zeta}^{(0)}\parallel \right) \le Ce^{-ct}.
\]
Retain \(D_{\rm P}\) for its endpoint variation. In disjoint invariant pole disks set \(T_{\rm cir}=I_3-C_\zeta^{(0)}/(k-\zeta)\), and set it to \(I_3\) elsewhere. Then \(m^{(1)}=m^{(0)}T_{\rm cir}\) has circle jumps \(I_3+C_\zeta^{(0)}/(k-\zeta)\), which are \(O(e^{-ct})\) with one \(x\)-derivative. The reductions and \((m_1)_{13}\) are preserved.

The six rays and their angles are
\[
\alpha _n=\frac{(n-1)\pi}{3},\qquad \Sigma _{n}^{0}=e^{i\alpha _n}\mathbb{R} _+,\qquad n=1,...,6.
\]
Set \(d_1(r)=1-|\widetilde r_1(r)|^2\) and
\(d_2(r)=1-|\widetilde r_2(-r)|^2\). The cofactor identities make them
continuous and nonzero for \(r>0\), with limit one at infinity. By
\eqref{4.7}, \(d_1(0)=d_2(0)=3/(1+A+A^{-1})>0\), so
\(\inf_{r\ge0}\min\{d_1(r),d_2(r)\}>0\).
For \(i\ne j\), put
\[
D_{ij}(d)=I_3+(d-1)E_{ii}+(d^{-1}-1)E_{jj}.
\]
The reverse Gauss identity
\[
(I_3+aE_{ij})(I_3+bE_{ji})
=\left(I_3+\frac{b}{d}E_{ji}\right)D_{ij}(d)
 \left(I_3+\frac{a}{d}E_{ij}\right),\qquad d=1+ab,
\]
gives \(v_n=\widehat B_n\mathcal D_n\widehat A_n\), where
\(\mathcal D_1=\operatorname{diag}(d_1^{-1},d_1,1)\),
\(\mathcal D_2=\operatorname{diag}(1,d_2,d_2^{-1})\), and the remaining
\(\mathcal D_n\) follow by \(\mathbb Z_3\)-rotation. Since
\(\prod_{n=1}^6\mathcal D_n=I_3\), their diagonal jumps are removed by the
standard Cauchy exponential \(\Delta_{\rm P}\), normalized at infinity and
with determinant one. The endpoint expansion, together with the cofactor and
reduction identities, cancels the possible \(k\log k\) term and yields
\begin{equation}\label{4.10}
\Delta _{\mathrm{P}}(k)=C_{n}^{\Delta}(I_3+O(k)),\qquad
\partial _k\Delta _{\mathrm{P}}(k)=O(1),\qquad k\rightarrow 0,\quad k\in D_n,
\end{equation}
where \(C_n^\Delta\) is diagonal and invertible.
Set \(\mu=m^{(1)}\Delta_{\rm P}^{-1}\). On each original ray,
\begin{equation}\label{4.11}
	v_{n}^{\Delta}=L_nU_n,\qquad L_n=(\Delta _{\mathrm{P}})_-\widehat{B}_n(\Delta _{\mathrm{P}})_{-}^{-1},\qquad U_n=(\Delta _{\mathrm{P}})_+\widehat{A}_n(\Delta _{\mathrm{P}})_{+}^{-1}.
\end{equation}
Open the lenses by setting
\[
R^{(2)}=\begin{cases}
	L_n,&		\alpha _n-\alpha <\mathrm{arg}k<\alpha _n,\\
	U_{n}^{-1},&		\alpha _n<\mathrm{arg}k<\alpha _n+\alpha ,\\
	I_3,&		\mathrm{otherwise}.\\
\end{cases}
\]
This transfers the jumps to the twelve outward rays \(\Sigma_j^{\rm P}\) in counterclockwise order. Figure~\ref{fig:SK-Painleve-rays} uses the limiting angle \(\pi/12\), while the analysis uses the admissible nearby \(\alpha\).

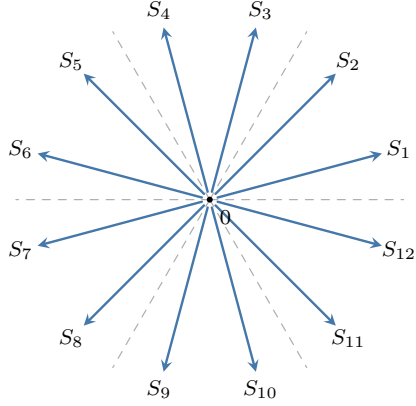
\begin{figure}[htbp]
	\centering
	\begin{tikzpicture}[scale=0.90,>=stealth,line cap=round,line join=round]
\definecolor{SKPBlue}{HTML}{4477AA}
\definecolor{SKPGray}{HTML}{999999}

\tikzset{
	modelray/.style={
		SKPBlue,
		line width=0.9pt
	},
	oldray/.style={
		SKPGray!80,
		dashed,
		line width=0.45pt
	},
	raylabel/.style={
		font=\scriptsize
	}
}
		
		\foreach \ang in {0,60,120,180,240,300}{
			\draw[oldray] (0,0)--(\ang:2.85);
		}
		
		\foreach \j in {1,...,12}{
			\pgfmathsetmacro{\ang}{15+30*(\j-1)}
			\draw[modelray,->] (\ang:0.12)--(\ang:2.62);
			\node[raylabel] at (\ang:2.88) {$S_{\j}$};
		}
		
		\fill (0,0) circle (1.25pt);
		\node[raylabel,below right] at (0,0) {$0$};
	\end{tikzpicture}
	\caption{The twelve outward-oriented rays $\Sigma_j^{\rm P}$ of the Painlev\'e model and the associated Stokes matrices $S_j$.  Dashed rays indicate the original six-ray contour $\Sigma$.}
	\label{fig:SK-Painleve-rays}
\end{figure}

Set \(\sigma=\overline{\widetilde r_1(0)}\) and
\(\tau=\widetilde r_2(0)\). The Stokes matrices are generated by
\[
\begin{gathered}
	S_1=I_3+\sigma E_{21},\quad
	S_2=I_3+\tau E_{32},\quad
	S_3=I_3-\overline\tau E_{23},\quad
	S_4=I_3-\overline\sigma E_{31},\\
	S_{j+4}=\mathcal A^{-1}S_j\mathcal A,\qquad j=1,\ldots,8.
\end{gathered}
\]
They also satisfy $S_{13-j}=\mathcal B\overline{S_j^{-1}}\mathcal B$. Combined with $H(\lambda;s)=s\lambda\Lambda+\frac95\lambda^5\Lambda^2$, the Painlev\'e model jumps are
\[
	V_{j}^{\mathrm{P}}(\lambda ;s)=e^{\mathrm{ad}H(\lambda ;s)}S_j,\qquad \lambda \in \Sigma _{j}^{\mathrm{P}}.
\]
Direct multiplication gives
\[
S_1S_2\cdots S_{12}=I_3\quad \Longleftrightarrow \quad \sigma +\tau +\overline{\sigma }\,\overline{\tau }=0.
\]
Equation~\eqref{4.7} verifies this identity, hence compatibility with the mSK data.

\begin{rhp}
	Find \(M^{\rm P}(\lambda;s)\) such that
	\begin{enumerate}
		\item[(i)] \(M^{\rm P}\) is analytic in \(\mathbb C\setminus\Sigma^{\rm P}\).
		\item[(ii)]
		\(M^{\rm P}_+=M^{\rm P}_-V_j^{\rm P}\) on \(\Sigma_j^{\rm P}\).
		\item[(iii)] \(M^{\rm P}\) is bounded at \(\lambda=0\) and satisfies the inherited reductions.
		\item[(iv)] As $\lambda \rightarrow \infty$,
		\begin{equation}\label{4.12}
			M^{\mathrm{P}}(\lambda ;s)=I+\frac{M_{1}^{\mathrm{P}}(s)}{\lambda}+O(\lambda ^{-2}).
		\end{equation}
	\end{enumerate}
\end{rhp}

\begin{lemma}\label{lem:4.3}
	Let $\Gamma$ be a finite union of rays with pairwise positive angles, and assume
	\[
	h,\partial _sh\in L^2(\Gamma \cap B_{2r_0})\cap L^1(\Gamma \setminus B_{r_0}).
	\]
	Let $M=I_3+Ch$, and let $B_j$ be bounded sectorial analytic matrices with bounded first $s$-derivatives. If $F_j=MB_j$ glue across the rays in $B_{r_0}\setminus\{0\}$, then
	\begin{equation}\label{4.13}
		\mathop {\mathrm{sup}} \limits_{0<r<r_0/2}\int_{|\lambda |=r}{\left( \parallel F(\lambda ;s)\parallel ^2+\parallel \partial _sF(\lambda ;s)\parallel ^2 \right) |d\lambda |}\le C_{\Gamma}(1+\mathcal{H} _{0}^{2}+\mathcal{H} _{1}^{2}),
	\end{equation}
	where
	\[
	\mathcal{H} _m=\parallel \partial _{s}^{m}h\parallel _{L^2(\Gamma \cap B_{2r_0})}+\parallel \partial _{s}^{m}h\parallel _{L^1(\Gamma \setminus B_{r_0})},\qquad m=0,1.
	\]
	The constant depends only on $\Gamma$, $r_0$, and the stated bounds for $B_j$.
\end{lemma}

\begin{proof}
Split \(h=h_{\rm loc}+h_{\rm far}\). The \(L^1\) norm controls the far
Cauchy transform, while angular separation gives uniformly bounded \(L^2\)
Carleman restrictions to \(|\lambda|=r\). Applying these bounds to
\(h,\partial_sh\), and then multiplying by \(B_j\), proves \eqref{4.13}.
Gluing makes \(F\) single-valued.
\end{proof}

We now verify the small-Stokes condition used below. Set
\[
	C_*=\frac{7}{2}+2\sqrt{6}+\sqrt{2}+\frac{2}{\sqrt{3}},\qquad \Phi _M=\frac{4\sqrt{3}}{5}M\left( \frac{M}{9\cos(\pi /12)} \right) ^{1/4},
\]
and
\[
q(W)=\left(\frac{2(\cosh W-1)}{1+2\cosh W}\right)^{1/2}.
\]
Define
\[
\eta _M=\frac{e^{-\Phi _M}}{2C_*},\qquad \varepsilon _{\mathrm{P}}(M)=\mathrm{ar}\cosh \left( \frac{2+\eta _{M}^{2}}{2(1-\eta _{M}^{2})} \right) .
\]

\begin{prop}\label{prop:4.2}
	Let $W_0=\int_{\mathbb R}w_0$. If
	\begin{equation}\label{4.14}
		\kappa_M:=C_*q(W_0)e^{\Phi_M}<1,
	\end{equation}
	then the Painlev\'e RH problem is uniquely solvable for $s\in[-M,M]$, with Beals-Coifman resolvent bounded by $(1-\kappa_M)^{-1}$. The solution, its inverse, and the required
	$s$-derivatives are uniformly bounded on pole-free compact subsets. Moreover,
	\[
	|W_0|\le \varepsilon _{\mathrm{P}}(M)\quad \Longrightarrow \quad \kappa _M\le \frac{1}{2}.
	\]
\end{prop}

\begin{proof}
	Ray deformation preserves \(M_1^{\rm P}\). The phase envelope and \eqref{4.7} give \(\|V^{\rm P}-I_3\|_\infty\le q(W_0)e^{\Phi_M}\), with the required weighted bounds. The Beals-Coifman operator has norm at most \(\kappa_M<1\), so its Neumann series gives the resolvent bound. The Cauchy formula, Lemma~\ref{lem:4.3}, and the Stokes cycle remove the vertex; hence \(\det M^{\rm P}=1\). Differentiation gives the stated bounds, and the definition of \(\eta_M\) proves the last claim.
\end{proof}

For the Beals-Coifman density \(\mu_{\rm BC}\), differentiation of its equation gives, uniformly for \(|s|\le M\),
\begin{equation}\label{4.15}
	\parallel \mu _{\mathrm{BC}}-I_3\parallel _2+\parallel M_{1}^{\mathrm{P}}\parallel \le \frac{C_Mq(W_0)}{1-\kappa _M},\qquad \parallel \partial _s\mu _{\mathrm{BC}}\parallel _2+\parallel \partial _sM_{1}^{\mathrm{P}}\parallel \le \frac{C_Mq(W_0)}{(1-\kappa _M)^2}.
\end{equation}
At distance \(\delta\) from the rays, the corresponding model bounds gain only \(\delta^{-1}\). Set $p(s)=3\bigl(M_1^{\rm P}(s)\bigr)_{13}$, Appendix~\ref{app1} shows that $p$ satisfies the fourth-order Painlev\'e equation used below.

\subsection{Analysis of the mixed \texorpdfstring{$\bar\partial$}{d-bar}-RH problem}
Set
\[
\epsilon =(5t)^{-1/5},\qquad \widehat{\mu }(\lambda ;x,t)=\mu (x,t;\epsilon \lambda ),
\]
and consider
\[
\Omega_1^{\rm P}=\{\lambda=u+iv:u>0,\ 0<v<u\tan\alpha\}.
\]
Writing $\Delta_{\rm P}=\operatorname{diag}(\delta_{{\rm P},1},\delta_{{\rm P},2},\delta_{{\rm P},3})$, the factor $U_1$ in \eqref{4.11} is
\[
U_1(k)=I_3+\Gamma_{1,+}(k)e^{\theta_{21}(k)}E_{21},
\]
where
\begin{equation}\label{4.16}
	\Gamma _{1,+}(k)=\frac{\overline{\widetilde{r}_1(k)}}{d_1(k)}\frac{\pi _1(k)}{\pi _2(k)}\frac{\delta _{\mathrm{P},2,+}(k)}{\delta _{\mathrm{P},1,+}(k)}.
\end{equation}
Since \(\pi_1(0)/\pi_2(0)=1\) and \(\delta_{{\rm P},2,+}(0)/\delta_{{\rm P},1,+}(0)=d_1(0)\), we have \(\Gamma_{1,+}(0)=\overline{\widetilde r_1(0)}=\sigma\).

For $\lambda=re^{i\varphi}\in\Omega_1^{\rm P}$, define
\[
\mathcal{R} _{1,+}(\lambda )=\sigma +\bigl[ \Gamma _{1,+}(\epsilon u)-\sigma \bigr] \cos \!\left( \frac{\pi \varphi}{2\alpha} \right) .
\]
This equals \(\Gamma_{1,+}(\epsilon u)\) on the real ray and \(\sigma\) on the model ray. The endpoint regularity and \eqref{4.10}, \eqref{4.16} give \(|\Gamma_{1,+}(k)-\sigma|+|k|\,|\Gamma_{1,+}'(k)|\le C|k|\), hence
\begin{equation}\label{4.17}
	|\bar{\partial}\mathcal{R} _{1,+}(\lambda )|\le C\epsilon .
\end{equation}

Choose an invariant cutoff \(\chi_{\rm p}\), equal to one on both ray systems and zero near the poles, such that
\[
\bar{\partial}_{\lambda}\chi _{\mathrm{p}}(\epsilon \lambda )=O(\epsilon ),\qquad \epsilon =(5t)^{-1/5}.
\]

In $\Omega_1^{\rm P}$ set
\[
R^{(2)}(\lambda )=I_3-\chi _{\mathrm{p}}(\epsilon \lambda )\mathcal{R} _{1,+}(\lambda )e^{\Theta _{21}(\lambda ;s)}E_{21}.
\]
It equals \(U_1^{-1}\) on \(\mathbb R_+\) and produces on the model ray
\[
(R_{-}^{(2)})^{-1}R_{+}^{(2)}=I_3+\sigma e^{\Theta _{21}}E_{21}=V_{1}^{\mathrm{P}}.
\]
Obtain the other extensions by reduction. With
\[
m^{(2)}=\widehat\mu R^{(2)},
\]
the twelve Painlev\'e jumps replace the original ones, with exponentially small pole circles, and
\[
\bar{\partial}m^{(2)}=m^{(2)}W^{(2)},\qquad W^{(2)}=(R^{(2)})^{-1}\bar{\partial}R^{(2)}.
\]

After scaling and conjugation by \(M^{\rm P}\Delta_{\rm P}(\epsilon\lambda)\),
the pole-circle jumps remain exponentially small. Small-norm theory gives
\begin{equation}\label{4.18}
E^{\rm disc}(\lambda)=I_3+\frac{E_1^{\rm disc}}{\lambda}+O(\lambda^{-2}),
\qquad E_1^{\rm disc},\ \partial_xE_1^{\rm disc}=O(e^{-c't}).
\end{equation}
Set \(m^{(3)}=m^{(2)}(M^{\rm P})^{-1}(E^{\rm disc})^{-1}\). Then the RH jumps
are removed and the remaining solid Cauchy equation has kernel \(W^{(3)}\).
By \eqref{4.17}, the phase decay, and the uniform model bounds,
\[
|W^{(3)}(u+iv)|\le C(5t)^{-1/5}e^{-cu^4v}
\]
on a representative wedge, with the other wedges obtained by symmetry.
Hence the associated solid Cauchy operator is \(O((5t)^{-1/5})\) on
\(L^\infty\), so it is invertible for large \(t\). Expanding at infinity gives
\begin{equation}\label{4.19}
m^{(3)}(\lambda )=I+\frac{m_{1}^{(3)}(x,t)}{\lambda}+O(\lambda ^{-2}),
\end{equation}
with
\begin{equation}\label{4.20}
m_{1}^{(3)}=O\bigl((5t)^{-1/5}\bigr).
\end{equation}
Differentiating the solid Cauchy equation and using \(\partial_xs=(5t)^{-1/5}\)
yields
\begin{equation}\label{4.21}
|\partial_xm_1^{(3)}|=O((5t)^{-2/5}).
\end{equation}

\subsection{Asymptotic analysis in the Painlev\'e region}
We now reconstruct mSK and apply the Miura map.

Undoing the transformations and using \eqref{4.12}, \eqref{4.18}, \eqref{4.19} gives
\begin{equation}\label{4.22}
	\lim_{k\rightarrow \infty} k\,m_{13}(x,t;k)=(5t)^{-1/5}\left[ \bigl( M_{1}^{\mathrm{P}}(s) \bigr) _{13}+\bigl( m_{1}^{(3)}(x,t) \bigr) _{13} \right] +O(e^{-ct}).
\end{equation}
The first coefficients of \(D_{\rm P},\Delta_{\rm P}\) are diagonal, \(R^{(2)}\) contributes nothing at infinity, and \eqref{4.18} controls the pole circles.

The mSK reconstruction formula and \eqref{4.20} give
\begin{equation}\label{4.23}
	w(x,t)=3\lim_{k\rightarrow \infty} k\,m_{13}(x,t;k)=(5t)^{-1/5}p(s)+O\bigl( (5t)^{-2/5} \bigr) .
\end{equation}
Reality of the mSK potential and \eqref{4.23} imply
\begin{equation}\label{4.24}
	p(s)\in\mathbb R,\qquad s\in[-M,M].
\end{equation}
The coefficient reductions give \(a_2=\omega a_3\), \(p=3a_3\). Appendix~\ref{app1} then yields
\begin{equation}\label{4.25}
	p^{(4)}=5p(p')^2+5p'p''+sp+5p^2p''-p^5.
\end{equation}
Thus RH compatibility selects the zero integration constant.

Differentiating \eqref{4.22} and using \eqref{4.21}, \eqref{4.23} gives
\begin{equation}\label{4.26}
	\begin{aligned}
		w_x&=(5t)^{-2/5}p'(s)+O((5t)^{-3/5}),\\
		w^2&=(5t)^{-2/5}p(s)^2+O((5t)^{-3/5}).
	\end{aligned}
\end{equation}
Consequently, the Miura map yields
\[
u(x,t)=\frac16(5t)^{-2/5} \bigl[p'(s)-p(s)^2\bigr]+O((5t)^{-3/5}).
\]
Equations~\eqref{4.23}, \eqref{4.25}, \eqref{4.26} prove Theorem~\ref{th2}{\rm(iii)}, and the Miura map proves Theorem~\ref{th1}{\rm(iii)}. The nonresonant discrete part is exponentially small by \eqref{4.18}.

\section{Long-time asymptotics in Sector IV}\label{sec5}
Fix \(M>1\), set \(J=(-\infty,-M^{-1}]\), \(\mathcal C_J=\{(x,t):t\ge t_0,\ x/t\in J\}\), and \(X=-x\). Then
\begin{equation}\label{5.1}
	\xi=x/t\le-M^{-1},\qquad
	X\ge t/M,\qquad X\le |x|+t\le(M+1)X.
\end{equation}
There are no nonzero stationary points on \(\Sigma\), leaving pole selection and the singular endpoint.

For \(r>0\), put \(d_1(r)=1-|r_1(r)|^2\), \(d_2(r)=1-|r_2(-r)|^2\), \(\gamma_1=d/c^2\), \(\gamma_2=-c/d^2\), and \(\mathfrak a=p_0/c+q_0/d\). Lemma~\ref{lem:2.1} and the cofactor identities give
\[
	d_1(r)=\gamma_1r\{1+\mathfrak a r+O(r^2)\},\qquad
	d_2(r)=\gamma_2r\{1-\mathfrak a r+O(r^2)\}.
\]
Here \(c<0<d\). The Schwartz Volterra recursion gives arbitrary-order endpoint expansions, and \(r_j/d_j=O(r^{-1})\) requires partial indices at \(0\).

Reduce the poles before opening lenses. Equation~\eqref{3.13} gives the common orbit velocity and
\(\operatorname{Re}\Phi_{21}(\xi;\lambda)=\sqrt3\,\operatorname{Im}\lambda\,(\xi-v_\lambda)\). Denote by \(\mathcal Z_{\rm ret}(J)\) the complete orbits with
\(v_\lambda\le-M^{-1}\), including equality. Define
\begin{equation}\label{5.2}
	D_{>J}(k)=
	\prod_{v_{\kappa_j}>-M^{-1}}D_{\kappa_j}(k)
	\prod_{v_{\eta_\ell}>-M^{-1}}D_{\eta_\ell}(k),
\end{equation}
using \eqref{3.15}-\eqref{3.16}. Then \(D_{>J}(0)=I_3\), \(D_{>J}=I_3+D_{>J}^{[1]}/k+O(k^{-2})\), and \(\partial_xD_{>J}^{[1]}=0\). The uniform phase gap for \(v_\lambda>-M^{-1}\) gives
\begin{equation}\label{5.3}
	\sum_{\lambda\notin\mathcal Z_{\rm ret}(J)}\|D_x^q(v_\lambda^{\rm cir}-I_3)\|_{L^1\cap L^2\cap L^\infty}\le Ce^{-\gamma_J(X+t)}.
\end{equation}
Thus the outer problem retains exactly \(\mathcal Z_{\rm ret}(J)\). Write
\[
\Delta_{\rm IV}=\operatorname{diag}(\Delta_{{\rm IV},1}, \Delta_{{\rm IV},2},\Delta_{{\rm IV},3}).
\]
Absorb it and the regular part of \eqref{5.2} once into the retained constants:
\begin{equation}\label{5.4}
	\begin{aligned}
		\widehat c_j(J)
		&=c_j\frac{\Delta_{{\rm IV},1}(\kappa_j)}
		{\Delta_{{\rm IV},2}(\kappa_j)}
		\frac{(D_{>J})_2(\kappa_j)}{(D_{>J})_1(\kappa_j)},
		&&v_{\kappa_j}\in J,\\
		\widehat d_\ell(J)
		&=d_\ell\frac{\Delta_{{\rm IV},2}(\eta_\ell)}
		{\Delta_{{\rm IV},1}(\eta_\ell)}
		\frac{(D_{>J})_1(\eta_\ell)}{(D_{>J})_2(\eta_\ell)},
		&&v_{\eta_\ell}\in J.
	\end{aligned}
\end{equation}
All factors in \eqref{5.4} are evaluated off the rays and pole disks. For the resulting datum \(\widehat\sigma(J)\), let
\[
R_J(k;x,t)=M(k;x,t\mid\widehat\sigma(J))=I_3+\frac{R_{J,1}(x,t)}{k}+O(k^{-2}).
\]
We assume the negative-cone version of the balanced gap:
\[
	\inf_{\substack{t\ge t_0\\x/t\le-M^{-1}}}|\det\mathbb A_J^{\rm bal}(x,t)|\ge d_J>0.
\]
Lemma~\ref{lem:3.1} then uniformly bounds \(R_J^{\pm1}\) and one \(x\)-derivative off the pole disks.

Set \(\alpha_n=(n-1)\pi/3\), \(n=1,\ldots,6\). For \(\operatorname{cyc}(a,b,c)=(b,c,a)\), set
\[
A_1=(-1,1,0),\qquad A_2=(0,1,-1),\qquad A_{n+2}=\operatorname{cyc}A_n,\quad n=1,\ldots,4.
\]
Put \(\mathfrak d_n=d_1\) for odd \(n\), \(\mathfrak d_n=d_2\) for even \(n\), and set
\[
\mathcal{D} _n(r)=\mathrm{diag}\bigl( \mathfrak{d} _n(r)^{A_{n1}},\mathfrak{d} _n(r)^{A_{n2}},\mathfrak{d} _n(r)^{A_{n3}} \bigr) .
\]
The real logarithm is used in
\[
\boldsymbol{\delta }_{\mathrm{IV}}(k)=\frac{1}{2\pi i}\sum_{n=1}^6{A_n\int_0^{\infty}{\frac{\log \mathfrak{d} _n(r)}{r-e^{-i\alpha _n}k}\,dr,\qquad \Delta _{\mathrm{IV}}(k)}}=\mathrm{diag}\mathrm{(}e^{\delta _{\mathrm{IV},1}},e^{\delta _{\mathrm{IV},2}},e^{\delta _{\mathrm{IV},3}}).
\]
The integral converges by integrability of \(\log r\) at \(0\) and Schwartz decay.

The Plemelj construction gives
\begin{equation}\label{5.5}
	(\Delta_{\rm IV})_+
	=(\Delta_{\rm IV})_-\mathcal D_n
	\quad\hbox{on }e^{i\alpha_n}\mathbb R_+,\qquad
	\det\Delta_{\rm IV}=1,\qquad
	\Delta_{\rm IV}(\infty)=I_3,
\end{equation}
and \(\Delta_{\rm IV}\) obeys both reductions. In \(D_n\),
\begin{equation}\label{5.6}
	\Delta _{\mathrm{IV}}(k)=C_nk^{B_n}\left[ I_3+\sum_{m=1}^L{\sum_{q=0}^{Q_m}{C_{n,mq}k^m(\log _nk)^q}}+O\bigl( |k|^{L+1}(1+|\log |k||)^{Q_{L+1}} \bigr) \right] ,
\end{equation}
where \(C_n\) is constant, diagonal, and invertible. The index vectors are
\[
	B_1=(-1,0,1),\qquad B_2=(-1,1,0),\qquad
	B_{n+2}=\operatorname{cyc}B_n,\quad n=1,\ldots,4.
\]
They satisfy
\begin{equation}\label{5.7}
	B_n-B_{n-1}=A_n,\qquad B_0=B_6.
\end{equation}
Here \(k^{B_n}\) uses the branch in \(D_n\). Splitting the Plemelj integrals and using the three displayed sum identities cancels the leading logarithms, while the opposite \(\mathfrak a\)-coefficients cancel \(k\log k\). Transport across ray \(n\) changes the index by \(A_n\), proving \eqref{5.6}-\eqref{5.7}.

For example, the reverse Gauss factorization on the first ray is
\begin{equation}\label{5.8}
	v_1=\left(I_3-\frac{r_1}{d_1}e^{-\theta_{21}}E_{12}\right)
	D_{21}(d_1)
	\left(I_3+\frac{\overline r_1}{d_1}e^{\theta_{21}}E_{21}\right),
\end{equation}
where \(D_{21}(d_1)=\operatorname{diag}(d_1^{-1},d_1,1)\). Reduction gives the other identities. Set
\[
M^{(\Delta)}=M\Delta_{\rm IV}^{-1},\qquad N^{(\Delta)}=\ell^{\rm T}M\Delta_{\rm IV}^{-1}.
\]
Equations~\eqref{5.5}, \eqref{5.8} remove the diagonal factors, and \(B_n\) compensates the \(O(k^{-1})\) amplitudes.

The reconstruction row supplies the partial-index zeros. In \(D_1\), \(q_3M_1^{[3]}=J_-^{[3]}\), \(q_3=d/k+O(1)\), and \(\ell^{\rm T}e=0\) give
\[
N_1,N_2=O(1),\qquad N_3=O(k),\qquad k\to0,\quad k\in D_1.
\]
Thus \(N^{(\Delta)}\) is bounded, with reductions cycling the \(O(k)\) entry. Only the reconstruction row is used at the endpoint, while a full matrix formal cocycle verifies compatibility.

Let
\[
\Omega _j=\left\{ k:\frac{(2j-1)\pi}{12}<\mathrm{arg}k<\frac{(2j+1)\pi}{12} \right\} ,\qquad j=0,...,11.
\]
After absorbing \(C_n\) into the frames, the endpoint factors occur in the cyclic order
\[
U_1,L_2,U_2,L_3,U_3,L_4,U_4,L_5,U_5,L_6,U_6,L_1.
\]
Conjugation by \(H=\operatorname{diag}(\gamma_1^{-1},1,(\gamma_1\gamma_2)^{-1})\) gives the universal leading coefficient tuple.
Thus no small endpoint parameter is assumed.
Transport each rotated, branched ray germ to \(z\), with entries in \(\mathcal R_{\rm for}=\mathbb C((z))[\log z]\). Only phase-free amplitudes are truncated. Distinct rays are not identified radially, and the carrier is exact.

Denote the twelve transported phase-free triangular factors, in cyclic order, by
\(g_1,\ldots,g_{12}\). The sector identities, diagonal conjugation, and reverse
factorization imply the formal compatibility relation
\begin{equation}\label{5.9}
g_1g_2\cdots g_{12}=I_3 \quad\hbox{in }GL_3(\mathcal R_{\rm for}).
\end{equation}
Consequently the partial products
\[
T_j^{\rm IV}=C_{\rm ep}g_1\cdots g_j,\qquad
C_{\rm ep}=\begin{pmatrix}1&0&0\\0&1&0\\a_6&-b_2&1\end{pmatrix},
\]
close cyclically at \(j=12\). Direct division-free multiplication shows that
their off-diagonal entries occur only in the cyclic root spaces dictated by
the adjacent factors, so the transformed reconstruction row is bounded at the
vertex. This argument also covers the degenerate case in which an intermediate
Gauss denominator vanishes.

For \(\lambda_1=\omega,\lambda_2=\omega^2,\lambda_3=1\), every allowed root obeys
\begin{equation}\label{5.10}
\operatorname{Re}\{-e^{i\arg k}(\lambda_p-\lambda_q)\}
\le-c_*,\qquad c_*=\sqrt3\sin\frac{\pi}{12}.
\end{equation}
Set \(\rho_X=X^{-1/2}\). For \(|k|\le2\rho_X\), the fifth-order term is a
relative \(O_M(X^{-2})\) perturbation, hence every allowed carrier is bounded by
\(e^{-c_*X|k|/2}\) for large \(X\). If \(T_j^{[L]}\) denotes the common-order
truncation in \(\Omega_j\), exact adjacent quotients telescope and \eqref{5.9}
closes the cycle. Thus the internal-edge mismatch is
\[
T_{j-1}^{[L]}v_j^{\rm amp}(T_j^{[L]})^{-1}
=I_3+O\bigl(|k|^{L+1}(1+|\log|k||)^{Q_L}\bigr),
\]
with the exact carrier left unexpanded.
A superscript $J$ denotes conjugation by $R_J$, using its boundary values for RH weights. The exact carrier, \eqref{5.10}, and the formal remainder give
\begin{equation}\label{5.11}
	\sum_{q=0}^1{\parallel D_{x}^{q}w_{\mathrm{ep}}^{J}}\parallel _{L^p}\le C\rho _{X}^{L+1+1/p}(\log X)^{Q_L}+CX^{-1/p}e^{-c\sqrt{X}},\qquad p=1,2,\infty .
\end{equation}

Outside \(|k|<\rho_X\), open reverse-factorization lenses satisfying, in every reduced wedge,
\begin{equation}\label{5.12}
	|e^{\pm\theta_{pq}(k)}| \le\exp\{-c(X+t|k|^4)\operatorname{dist}(k,\Sigma)\}.
\end{equation}
Extend each explicit \(k^m(\log_nk)^q\) term in \eqref{5.6} analytically with its branch. For the remainder \(\mathcal R_L\),
\begin{equation}\label{5.13}
	\partial_u^m\mathcal R_L(u)=O\bigl(u^{L+1-m}(1+|\log u|)^{Q_L}\bigr),\qquad 0\le m\le L+1.
\end{equation}
Extend only \(\mathcal R_L\) almost analytically, and join the endpoint and exterior multipliers by their exact quotient. Thus cutoffs never differentiate lower-order terms.

Figure~\ref{fig:SK-sector-IV-endpoint-exterior} summarizes the deformation.

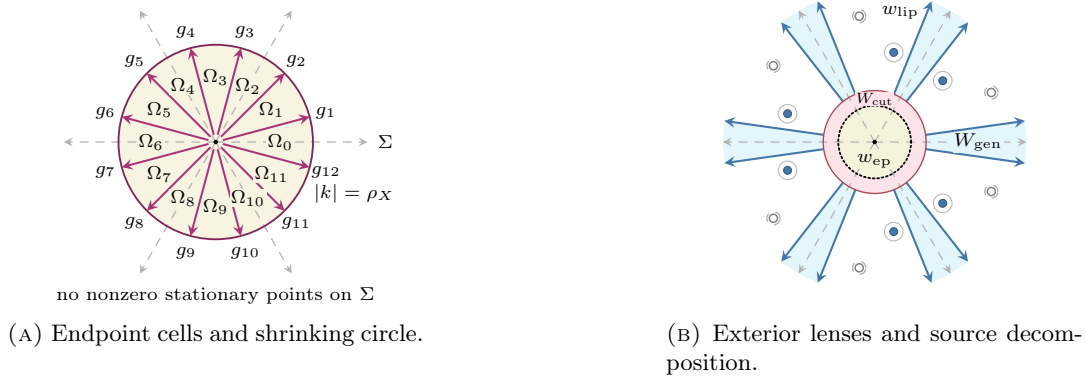
\begin{figure}[htbp]
	\centering
	\definecolor{SKBlue}{HTML}{4477AA}
	\definecolor{SKCoral}{HTML}{EE6677}
	\definecolor{SKGreen}{HTML}{228833}
	\definecolor{SKYellow}{HTML}{CCBB44}
	\definecolor{SKCyan}{HTML}{66CCEE}
	\definecolor{SKPurple}{HTML}{AA3377}
	\definecolor{SKGray}{HTML}{999999}
	
\subfloat[Endpoint cells and shrinking circle.]{%
	\begin{minipage}[t]{0.48\textwidth}
		\centering
		\begin{tikzpicture}[scale=0.83,>=stealth,line cap=round,line join=round]
			\def\rhoep{1.55}
			\def\Router{2.42}
			\path[use as bounding box] (-3.15,-2.67) rectangle (3.15,2.50);
			
			\tikzset{
				oldray/.style={SKGray!75,dashed,line width=0.48pt},
				endpointedge/.style={SKPurple,line width=0.76pt},
				smalllabel/.style={font=\tiny}
			}
			
			\fill[SKYellow!16] (0,0) circle (\rhoep);
			
			\foreach \ang in {0,60,120,180,240,300}{
				\draw[oldray,->] (0,0)--(\ang:\Router);
			}
			
			\foreach \j in {0,...,11}{
				\pgfmathsetmacro{\edgeang}{15+30*\j}
				\pgfmathsetmacro{\cellang}{30*\j}
				\pgfmathtruncatemacro{\gj}{\j+1}
				
				\draw[endpointedge,->]
				(\edgeang:0.13)--(\edgeang:\rhoep);
				
				\node[smalllabel]
				at (\edgeang:1.82)
				{$g_{\gj}$};
				
				\node[smalllabel,fill=SKYellow!16,inner sep=0.4pt]
				at (\cellang:1.03)
				{$\Omega_{\j}$};
			}
			
			\draw[SKPurple!75!black,line width=0.70pt]
			(0,0) circle (\rhoep);
			
			\fill (0,0) circle (1.0pt);
			
			\node[smalllabel,anchor=west,inner sep=0.6pt]
			at (332:1.73) {$|k|=\rho_X$};
			
			\node[smalllabel,right] at (0:\Router) {$\Sigma$};
			
			\node[smalllabel]
			at (0,-2.44)
			{no nonzero stationary points on $\Sigma$};
		\end{tikzpicture}
	\end{minipage}%
}
	\hfill
\subfloat[Exterior lenses and source decomposition.]{%
	\begin{minipage}[t]{0.48\textwidth}
		\centering
		\begin{tikzpicture}[scale=0.83,>=stealth,line cap=round,line join=round]
			\def\rhoin{0.58}
			\def\rhoout{0.82}
			\def\Router{2.42}
			\path[use as bounding box] (-3.15,-2.67) rectangle (3.15,2.50);
			\def\lensang{8}
			
			\tikzset{
				oldray/.style={SKGray!75,dashed,line width=0.48pt},
				lenslip/.style={SKBlue,line width=0.78pt},
				dbarfill/.style={fill=SKCyan!18,draw=none},
				retained/.style={
					circle,
					fill=SKBlue,
					draw=SKBlue!75!black,
					inner sep=1.05pt
				},
				omitted/.style={
					circle,
					fill=white,
					draw=SKGray!85!black,
					line width=0.58pt,
					inner sep=1.08pt
				},
				cutoff/.style={
					circle,
					fill=white,
					draw=SKGray!70,
					line width=0.40pt,
					minimum size=7.0pt,
					inner sep=0pt
				},
				polecircle/.style={
					SKGray!80,
					dashed,
					line width=0.48pt
				},
				smalllabel/.style={font=\tiny}
			}
			
			\foreach \rot in {0,60,120,180,240,300}{
				\begin{scope}[rotate=\rot]
					\fill[dbarfill]
					(\rhoout,0)--(\Router,0)
					--({\Router*cos(\lensang)},{\Router*sin(\lensang)})
					--({\rhoout*cos(\lensang)},{\rhoout*sin(\lensang)})
					--cycle;
					
					\fill[dbarfill]
					(\rhoout,0)--(\Router,0)
					--({\Router*cos(\lensang)},{-\Router*sin(\lensang)})
					--({\rhoout*cos(\lensang)},{-\rhoout*sin(\lensang)})
					--cycle;
					
					\draw[lenslip,->]
					({\rhoout*cos(\lensang)},{\rhoout*sin(\lensang)})
					--({\Router*cos(\lensang)},{\Router*sin(\lensang)});
					
					\draw[lenslip,->]
					({\rhoout*cos(\lensang)},{-\rhoout*sin(\lensang)})
					--({\Router*cos(\lensang)},{-\Router*sin(\lensang)});
				\end{scope}
			}
			
			\fill[SKCoral!18] (0,0) circle (\rhoout);
			\fill[SKYellow!18] (0,0) circle (\rhoin);
			
			\foreach \ang in {0,60,120,180,240,300}{
				\draw[oldray,->] (0,0)--(\ang:\Router);
			}
			
			\draw[densely dotted,line width=0.68pt]
			(0,0) circle (\rhoin);
			
			\draw[SKCoral!70!black,line width=0.50pt]
			(0,0) circle (\rhoout);
			
			\foreach \ang in {42,78,162,198,282,318}{
				\node[cutoff] at (\ang:1.46) {};
				\node[retained] at (\ang:1.46) {};
			}
			
			\foreach \ang in {23,97,143,217,263,337}{
				\draw[polecircle,<-]
				(\ang:2.02) circle (0.105);
				\node[omitted]
				at (\ang:2.02) {};
			}
			
			\fill (0,0) circle (1.0pt);
			
			\node[smalllabel,fill=SKYellow!18,inner sep=0.4pt]
			at (270:0.30) {$w_{\rm ep}$};
			
			\node[smalllabel,fill=SKCoral!18,inner sep=0.3pt,scale=0.80]
			at (90:0.69) {$W_{\rm cut}$};
			
			\node[smalllabel,fill=SKCyan!18,inner sep=0.4pt]
			at (0:1.65) {$W_{\rm gen}$};
			
			\node[smalllabel,anchor=south east,inner sep=1.2pt]
			at (68:2.05) {$w_{\rm lip}$};
			
		\end{tikzpicture}
	\end{minipage}%
}
	
	\caption{Sector IV endpoint-exterior deformation. (a) Endpoint cells $\Omega_j$, cyclic factors $g_j$, and the shrinking circle $|k|=\rho_X$, with $\Sigma$ dashed. (b) The sources $w_{\rm ep}$, $W_{\rm gen}$, $w_{\rm lip}$, and $W_{\rm cut}$ on the endpoint disk, lenses, lips, and annulus. Solid and open sets are complete retained and omitted pole orbits with representative velocities $v_\lambda\le-M^{-1}$ and $v_\lambda>-M^{-1}$, respectively. Dashed clockwise circles mark the exponentially small omitted-pole jumps. All pole disks are $\bar\partial$-free.}
	\label{fig:SK-sector-IV-endpoint-exterior}
\end{figure}

For $L\ge K\ge1$ and $1\le p\le\infty$, the remaining conjugated weights satisfy
\begin{equation}\label{5.14}
	\begin{aligned}
		\sum_{q=0}^1\|D_x^qW_{\rm cut}^J\|_{L^p(\mathbb C)}
		&\le C\rho_X^{L+2/p}(\log X)^{Q_L},\\
		\sum_{q=0}^1\|D_x^qW_{\rm gen}^J\|_{L^p(\mathbb C)}
		&\le C_{K,p}X^{-K-1/p}(\log X)^{Q_K},\\
		\sum_{q=0}^1\|D_x^qw_{\rm lip}^J\|_{L^p}
		&\le CX^{-1/p}e^{-c\sqrt X}.
	\end{aligned}
\end{equation}
Equation~\eqref{5.12}, the remainder bound \eqref{5.13}, annular scaling,
and tangential integration by parts prove
\eqref{5.14}, including one \(x\)-derivative. With \eqref{5.3} and
\eqref{5.11}, the differentiated RH and solid Cauchy equations are invertible
for \(L\ge2N+4\), \(K\ge N+2\). Their first row-error coefficients
\(E_1^{\rm RH},E_1^{(3)}\) satisfy
\[
\sum_{q=0}^1\left(|D_x^qE_1^{\rm RH}|+|D_x^qE_1^{(3)}|\right)\le C_{N,M}X^{-N}.
\]
Undoing the transformations, the local multipliers vanish at infinity, and
\(\Delta_{\rm IV}^{[1]},D_{>J}^{[1]}\) are diagonal and independent of
\(x\). Reconstruction by \eqref{2.10} and \eqref{5.1} therefore gives
\begin{equation}\label{5.15}
u=u_J^{\rm out}+O((|x|+t)^{-N}),\qquad u_J^{\rm out}=-\frac12\partial_x(\ell^{\rm T}R_{J,1}e_3).
\end{equation}
Compatibility identifies \(u_J^{\rm out}\) with the reflectionless SK solution. For scalar simple-ray data \(v_j>0\), so \(R_J=I_3\) and
\[
u(x,t)=O((|x|+t)^{-N}),\qquad \xi<0,\quad M^{-1}\le|\xi|.
\]
This proves Theorem~\ref{th1}{\rm(iv)}.

\section{Soliton resolution and asymptotic stability}\label{sec6}
The rational orbit problem need not satisfy the scalar SK origin completion. Thus \(u_{\rm sol}(x,t\mid\sigma_d(I;\xi))\) denotes a moving algebraic outer profile, while \(u_{\rm sol}(x,t;\sigma^{\rm phys})\) and \(u_{\rm sol}(x,t;\widehat\sigma(I))\) denote fixed reflectionless SK solutions.

\subsection{Asymptotic decomposition of reflectionless \texorpdfstring{$N$}{N}-soliton solutions}
\leavevmode\par
\noindent Let
\[
\sigma ^{\mathrm{phys}}=\{(\eta _j,\beta _j):1\le j\le N\},\qquad 0<\eta _1<\cdots <\eta _N,\quad \beta _j>0,
\]
be the scalar \(P=0\) data of \eqref{3.8}. The scalar closure conditions \eqref{3.9}-\eqref{3.10} give \(u_{\rm sol}=\partial_x^2\log\Delta_N\), with positive \(\Delta_N\) from \eqref{3.7}. Hence the solution is globally regular and Schwartz for fixed \(t\).

Under matrix compatibility, if
\[
R(k;x,t\mid\mathfrak G_N(\sigma^{\rm phys}))=I_3+\frac{R_1(x,t)}k+O(k^{-2}),
\]
then Proposition~\ref{prop:3.2} gives the same solution as \(-\frac12\partial_x(\ell^{\rm T}R_1e_3)\).

We use $p_j$, $v_j$, $y_j$, and $A_{ij}$ introduced before \eqref{3.7}.

For \(N=1\), \(\Delta_1=1+y\) gives the exact solitary wave
\[
u_{\mathrm{sol}}(x,t;\eta ,\beta )=\frac{3\eta ^2}{4}\operatorname{sech}^2\!\left[ \frac{\sqrt{3}\eta}{2}\left( x-9\eta ^4t+\frac{\log \beta}{\sqrt{3}\eta} \right) \right] .
\]
Since \(dv_j/d\eta_j>0\), the scalar simple-ray velocities satisfy \(0<v_1<\cdots<v_N\), so each velocity group is one orbit.

For \(i<j\), use the interaction factors \(A_{ij}\) defined above, and define
\begin{equation}\label{6.1}
	\beta_j^+=\beta_j\prod_{i<j}A_{ij},\qquad \beta_j^-=\beta_j\prod_{i>j}A_{ji},\qquad x_j^\pm=-p_j^{-1}\log\beta_j^\pm.
\end{equation}
Empty products are one, and the one-orbit solutions are
\[
U_j^\pm(x,t)=\frac{p_j^2}{4} \operatorname{sech}^2\!\left(\frac{p_j(x-v_jt-x_j^\pm)}2\right).
\]
Compatibility identifies \(U_j^\pm=-\frac12\partial_x(\ell^{\rm T}R_{j,1}^\pm e_3)\).

\begin{theorem}[Soliton resolution conjecture]\label{th3}
	Under the scalar closure conditions \eqref{3.9}-\eqref{3.10},
	\eqref{3.7} defines a global real Schwartz-in-\(x\) SK \(N\)-soliton. Let \(\beta_j^\pm,x_j^\pm\) be given by \eqref{6.1}, and let \(U_j^\pm\) be the profiles above. For every integer \(q\ge0\), there exist \(c,C_q>0\) such that
	\begin{equation}\label{6.2}
		\sup_{x\in\mathbb R}
		\left|
		\partial_x^q
		\left[
		u_{\mathrm{sol}}(x,t;\sigma^{\mathrm{phys}})
		-\sum_{j=1}^N U_j^\pm(x,t)
		\right]
		\right|
		\le C_qe^{-c|t|},
		\qquad t\to\pm\infty.
	\end{equation}
	The corresponding position shifts are
	\begin{equation}\label{6.3}
		x_j^+-x_j^-
		=\frac{1}{p_j}
		\left(
		\sum_{i>j}\log A_{ji}
		-\sum_{i<j}\log A_{ij}
		\right).
	\end{equation}

Let
	\(r=r(M)\), \(\widehat U_{j,M}^+\), and
	\(\widehat U_{j,I}^+\) be as defined in
	Section~\ref{sec:positive-resolution}. Under the hypotheses of Theorem~\ref{th1},
	\[
	u(x,t)
	=\sum_{j=r+1}^N\widehat U_{j,M}^+(x,t)
	+t^{-1/2}h_M+O(t^{-3/4}),
	\]
	uniformly for \(x/t\ge M\). For the chosen scalar preimage
	of \(\widehat\sigma(I)\),
	\[
	u(x,t)
	=\sum_{j=1}^{N(I)}\widehat U_{j,I}^+(x,t)
	+t^{-1/2}h+O(t^{-3/4}),
	\]
	locally uniformly for \(M^{-1}\le x/t<M\).
\end{theorem}

\begin{re}
	Equations~\eqref{6.2}-\eqref{6.3} describe the asymptotic decomposition
	and position shifts of a reflectionless SK \(N\)-soliton. The sectorial
	formulas add the \(t^{-1/2}\) radiation generated by the continuous
	spectrum. Empty selections yield pure radiation, and unselected solitons
	are exponentially small. If \(v_j=M\), the soliton tube is restricted to
	\(x/t\ge M\).
\end{re}

Set \(m_j=(v_j+v_{j+1})/2\) for \(1\le j<N\), with \(m_0=-\infty\), \(m_N=+\infty\). For \(t>0\), factor \(A(\{1,\ldots,j-1\})\prod_{i<j}y_i\) from \(\Delta_N\) on \(m_{j-1}t\le x\le m_jt\). For \(t<0\), use the reversed cells \(m_jt\le x\le m_{j-1}t\) and factor \(A(\{j+1,\ldots,N\})\prod_{i>j}y_i\). The remaining factor is \(1+\widehat y_j^\pm+\rho_j^\pm\), where
\[
\widehat y_j^\pm=\beta_j^\pm e^{p_j(x-v_jt)},\qquad
\max_{0\le q\le m}\left|\partial_x^q\frac{\rho_j^\pm}{1+\widehat y_j^\pm}\right|\le C_m e^{-c|t|},\quad t\to\pm\infty,
\]
uniformly on the corresponding cells, for every fixed integer \(m\ge0\). This follows by grouping the terms in \eqref{3.7} according to the dominant subset. The positive prefactors have logarithms affine in \(x\), so two logarithmic derivatives remove them. The other soliton profiles are exponentially small on each cell. The cell cover proves \eqref{6.2}, and \eqref{6.1} gives \eqref{6.3}.

\subsection{Soliton resolution in Sectors I and II}\label{sec:positive-resolution}
Fix \(M>0\), and let \(r=r(M)\in\{0,\ldots,N\}\) be uniquely determined by
\[
v_r<M\le v_{r+1},\qquad v_0=-\infty,\qquad v_{N+1}=+\infty.
\]
Thus \(r<j\le N\) are retained, including equality at \(M\). Set
\[
\beta_j^{(M)}=\beta_j\prod_{i=1}^rA_{ij}.
\]
Then, for every fixed \(q\ge0\),
\[
\mathop {\mathrm{sup}} \limits_{x/t\ge M}\left| \partial _{x}^{q}\!\left[ u_{\mathrm{sol}}(x,t;\sigma ^{\mathrm{phys}})-u_{\mathrm{sol}}(x,t;\{(\eta _j,\beta _{j}^{(M)}):j>r\}) \right] \right|\le C_qe^{-c_Mt},
\]
where \(0<c_M<\min_{i\le r}p_i(M-v_i)\), with equality exact when \(r=0\). Factoring \(A(O_M)\prod_{i\in O_M}y_i\) leaves the corrected retained polynomial and an \(O(e^{-c_Mt})\) differentiated remainder.
The positive half-cone restriction is illustrated in Figure~\ref{fig:SK-soliton-selection}.

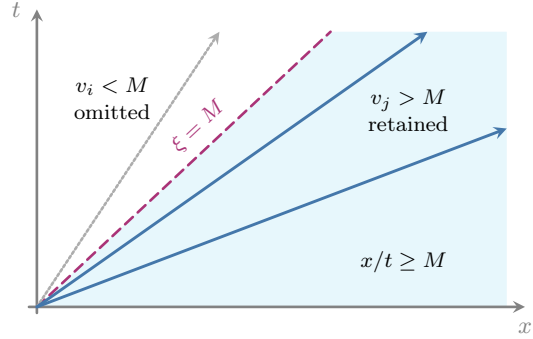
\begin{figure}[htbp]
	\centering
	
	\definecolor{SKBlue}{HTML}{4477AA}
	\definecolor{SKCoral}{HTML}{EE6677}
	\definecolor{SKGreen}{HTML}{228833}
	\definecolor{SKYellow}{HTML}{CCBB44}
	\definecolor{SKCyan}{HTML}{66CCEE}
	\definecolor{SKPurple}{HTML}{AA3377}
	\definecolor{SKGray}{HTML}{999999}
	
	\begin{tikzpicture}[
		x=1.05cm,
		y=0.80cm,
		>=stealth,
		line cap=round,
		line join=round
		]
		\tikzset{
			axis/.style={
				SKGray!85!black,
				line width=1.00pt
			},
			retained/.style={
				SKBlue,
				line width=1.10pt
			},
			omitted/.style={
				SKGray!80,
				densely dotted,
				line width=1.00pt
			},
			boundary/.style={
				SKPurple,
				dash pattern=on 4pt off 2.5pt,
				line width=1.00pt
			},
			smalllabel/.style={
				font=\scriptsize,
				inner sep=1pt
			}
		}
		
		\path[use as bounding box]
		(-0.42,-0.42) rectangle (6.40,5.12);
		
		\fill[SKCyan!16]
		(0,0)--(3.70,4.55)--(5.92,4.55)--(5.92,0)--cycle;
		
		\draw[axis,->]
		(-0.10,0)--(6.15,0)
		node[below=2pt,font=\footnotesize] {$x$};
		
		\draw[axis,->]
		(0,-0.10)--(0,4.92)
		node[left=2pt,font=\footnotesize] {$t$};
		
		\draw[boundary]
		(0,0)--(3.70,4.55)
		node[pos=0.60,sloped,above=3pt,smalllabel]
		{$\xi=M$};
		
		\draw[omitted,->]
		(0,0)--(2.30,4.55);
		
		\draw[retained,->]
		(0,0)--(4.92,4.55);
		
		\draw[retained,->]
		(0,0)--(5.92,2.95);
		
		\node[smalllabel,align=center]
		at (0.95,3.45)
		{$v_i<M$\\omitted};
		
		\node[smalllabel,align=center]
		at (4.68,3.25)
		{$v_j>M$\\retained};
		
		\node[smalllabel]
		at (4.62,0.75)
		{$x/t\ge M$};
		
	\end{tikzpicture}
	
	\caption{Positive half-cone selection in Sector I. For \(v_r<M\le v_{r+1}\), the region \(x/t\ge M\) retains \(j>r\) solitons and omits \(i\le r\) solitons, with modified constants
		\(\beta_j^{(M)}=\beta_j\prod_{i=1}^r A_{ij}\). Constant shifts are suppressed.}
\label{fig:SK-soliton-selection}
\end{figure}

For Theorem~\ref{th1}{\rm(i)}, assume that the complete nonzero discrete datum is of the scalar simple-ray form
\[
\sigma^{\rm phys}=\{(\eta_j,\beta_j):1\le j\le N\},\qquad 0<\eta_1<\cdots<\eta_N,\qquad \beta_j>0.
\]
Let \(r=r(M)\) be the index fixed at the beginning of this subsection, and define the retained moving datum by
\[
\sigma_M(\xi)=\mathfrak G_{N-r}\left(\left\{\left(\eta_j,\,\beta_j\tau_j(\xi)\prod_{i=1}^{r}A_{ij}\right):r<j\le N\right\}\right),\qquad \xi\ge M.
\]
Here \(\tau_j>0\) contains only the continuous-spectrum transmission factor, and \(\prod_{i=1}^rA_{ij}\) is the discrete restriction correction.

More precisely, the Sector~I part of Assumption~\ref{ass:1.1}{\rm(b)} requires that, for every \(H>\max\{M,V_*\}\), after the same pole switches
	\[
		\sigma_M(\xi)=\sigma_d(I_{M,H};\xi),\qquad I_{M,H}=[M,2H],\qquad M\le\xi\le2H.
	\]
	Thus \(\sigma_M\) is the continuation of the transformed compact-cone datum. It also requires
	\begin{equation}\label{6.4}
		\sup_{\substack{\xi\ge M\\r<j\le N}}\left\{\tau_j(\xi)+\tau_j(\xi)^{-1}+|\partial_\xi\tau_j(\xi)|\right\}<\infty.
	\end{equation}
The transmission freezing in \eqref{3.14} is scalar-to-matrix compatible, and the moving and frozen balanced residue systems lie in finitely many common single-valued \(C^1\) charts, with uniformly bounded resolvents and first \(x\)-derivatives. Empty retained sets correspond to the identity outer model.
Thus \eqref{3.12}, with \(N\) replaced by \(N-r\), defines
\[
u_{\rm sol}(x,t;\widehat\sigma_M)=u_{\rm sol}(x,t;\widehat\sigma_M^{\rm phys}).
\]

For \(r<j\le N\), define
\[
\widehat\beta_{j,M}^{+}=\widehat\beta_{j,M}\prod_{r<i<j}A_{ij}=\beta_j\tau_j(v_j)\prod_{i<j}A_{ij},\qquad\widehat x_{j,M}^{+}=-\frac{\log\widehat\beta_{j,M}^{+}}{p_j},
\]
and
\[
\widehat U_{j,M}^{+}(x,t)=\frac{p_j^2}{4}\operatorname{sech}^2\!\left(\frac{p_j(x-v_jt-\widehat x_{j,M}^{+})}{2}\right).
\]
When \(r=N\), the retained datum and the corresponding sum are empty.

\begin{lemma}\label{lem:6.1}
	Under the Sector~I conditions above, let
	\[
	\begin{aligned}
		M_M^{\rm mov}(k;x,t)
		&=I_3+\frac{M_{M,1}^{\rm mov}(x,t)}{k}+O(k^{-2}),\\
		M_M^{\rm fix}(k;x,t)
		&=I_3+\frac{M_{M,1}^{\rm fix}(x,t)}{k}+O(k^{-2}),
	\end{aligned}
	\]
	be the rational outer models with data \(\sigma_M(x/t)\) and \(\widehat\sigma_M\), respectively.  For all sufficiently large \(t\), both models are uniquely solvable and
	\begin{equation}\label{6.5}
		\sup_{\xi\ge M}\sum_{q=0}^{1}\left\|D_x^q\bigl(M_{M,1}^{\rm mov}-M_{M,1}^{\rm fix}\bigr)\right\|\le C_M\frac{\log(2+t)}{t},
	\end{equation}
	where \(D_x\) is the fixed-\(t\) total derivative.  At the six moving stationary points,
	\begin{equation}\label{6.6}
		\begin{aligned}
			& \underset{\xi \ge M}{\sup} \sum_{\lambda \in \mathcal{K}_0} \sum_{q=0}^1
			\left\{
			\| D_{x}^{q}[M_{M}^{\mathrm{mov}}(\lambda) - M_{M}^{\mathrm{fix}}(\lambda)] \|
			\right. \\
			&\qquad \left.
			+ \| D_{x}^{q}[(M_{M}^{\mathrm{mov}}(\lambda))^{-1} - (M_{M}^{\mathrm{fix}}(\lambda))^{-1}] \|
			\right\} \le C_M \frac{\log(2+t)}{t}.
		\end{aligned}
	\end{equation}
	Here evaluation at \(\lambda=\lambda(x/t)\) precedes \(D_x\).
	
	Every orbit with \(v_\lambda<M\) can be replaced by a decaying pole-circle jump satisfying
	\begin{equation}\label{6.7}
		\sum_{\lambda :\,v_{\lambda}<M}{\parallel \partial _{x}^{q}(v_{\lambda}^{\mathrm{cir}}}-I_3)\parallel _{L^p(\partial D_{\lambda})}\le C_Me^{-\gamma _Mt},\qquad p=1,2,\infty ,\quad q=0,1,
	\end{equation}
	where
	\[
	\gamma_M=\frac{\sqrt3}{2}\min_{\lambda:\,v_\lambda<M}|\operatorname{Im}\lambda|\,(M-v_\lambda)>0.
	\]
	If no orbit is discarded, set \(\gamma_M=+\infty\).  Consequently, uniformly for \(x/t\ge M\),
	\begin{equation}\label{6.8}
		-\frac{1}{2}D_x\bigl( \ell ^{\mathrm{T}}M_{M,1}^{\mathrm{mov}}e_3 \bigr) =u_{\mathrm{sol}}(x,t;\widehat{\sigma }_M)+O_M\!\left( \frac{\log\mathrm{(}2+t)}{t}+e^{-\gamma _Mt} \right) ,
	\end{equation}
	with the exponential term omitted when no orbit is discarded.
\end{lemma}

\begin{proof}
	The phase identity gives \eqref{6.7}, retaining equality. In logarithmic tubes, \(\xi=V_m+O(\log(2+t)/t)\), while inactive residues are \(O(t^{-2})\) outside. Equation~\eqref{6.4}, common balanced charts, and the resolvent identity give \eqref{6.5}-\eqref{6.6}. Reconstruction and compatibility give \eqref{6.8}.
\end{proof}

Proposition~\ref{prop:3.5} and \eqref{6.8} replace the moving model by \(\widehat\sigma_M\) with error \(O(t^{-3/4})\), proving \eqref{3.1}.

On a compact cone \(I=[\xi_-,\xi_+]\Subset(0,\infty)\), set
\[
\begin{aligned}
	D_{>I}(k)&:=D_I(k)=\mathrm{diag}\mathrm{(}D_{>I,1},D_{>I,2},D_{>I,3}),\\
	\mathfrak{t} _1(k;\xi )&=\frac{\widetilde{\delta }_{v1}(k;\xi )}{\delta _1(k;\xi )^2}, \qquad \mathfrak{t} _4(k;\xi )=\frac{\widetilde{\delta }_{v4}(k;\xi )}{\delta _4(k;\xi )^2}.\\
\end{aligned}
\]
For each retained representative, freeze the transmission factor at its orbit velocity
\[
\begin{aligned}
	\widehat{c}_j(I)&=c_j\mathfrak{t} _1(\kappa _j;v_{\kappa _j})\frac{D_{>I,2}(\kappa _j)}{D_{>I,1}(\kappa _j)}, \quad v_{\kappa _j}\in I,\\
	\widehat{d}_{\ell}(I)&=d_{\ell}\mathfrak{t} _4(\eta _{\ell};v_{\eta _{\ell}})\frac{D_{>I,1}(\eta _{\ell})}{D_{>I,2}(\eta _{\ell})}, \quad v_{\eta _{\ell}}\in I.\\
\end{aligned}
\]
Generate the remaining constants by \eqref{2.6} and define
\[
	\widehat\sigma(I)=\{(\kappa_j,\widehat c_j(I)):v_{\kappa_j}\in I\}\cup\{(\eta_\ell,\widehat d_\ell(I)):v_{\eta_\ell}\in I\}.
\]
Assume that \(\widehat\sigma(I)\) is SK-realizable and choose an ordered scalar preimage
\[
\sigma_I^{\rm phys}=\{(\eta_{j,I},\beta_{j,I}):1\le j\le N(I)\},\qquad 0<\eta_{1,I}<\cdots<\eta_{N(I),I},\qquad\beta_{j,I}>0,
\]
with \(\widehat\sigma(I)=\mathfrak G_{N(I)}(\sigma_I^{\rm phys})\). Put
\[
p_{j,I}=\sqrt3\,\eta_{j,I},\qquad v_{j,I}=p_{j,I}^{4}\in I.
\]
Let \(A_{ij}^I\) be the same interaction factor with \(p_i,p_j\) replaced by \(p_{i,I},p_{j,I}\), and set
\[
\widehat\beta_{j,I}^{+}=\beta_{j,I}\prod_{i<j}A_{ij}^{I},\;\widehat x_{j,I}^{+}=-\frac{\log\widehat\beta_{j,I}^{+}}{p_{j,I}},\; \widehat U_{j,I}^{+}(x,t)=\frac{p_{j,I}^{2}}{4}\operatorname{sech}^2\!\left(\frac{p_{j,I}(x-v_{j,I}t-\widehat x_{j,I}^{+})}{2}\right).
\]
For empty data take \(N(I)=0\). Proposition~\ref{prop:3.2} gives
\[
u_{\rm sol}(x,t;\widehat\sigma(I))=-\frac12\partial_x\bigl(\ell^{\rm T}R_1(x,t\mid\widehat\sigma(I))e_3\bigr).
\]
For every balanced chart associated with \(\widehat\sigma(I)\), let \(\mathbf r_{\rm fix}^{\rm bal}\) denote its solution vector and assume
\[
\sup_{\substack{t\ge t_0\\x/t\in I}}\sum_{q=0}^1\|D_x^q\mathbf r_{\rm fix}^{\rm bal}(x,t)\|<\infty.
\]
This is the fixed-model counterpart of the regularity condition \eqref{3.22} for the moving \(\xi\)-dependent residue system.

Lemma~\ref{lem:6.1} and the common balanced charts give
\[
u_{\rm sol}(x,t\mid\sigma_d(I;\xi))=u_{\rm sol}(x,t;\widehat\sigma(I))+O\!\left(\frac{\log t}{t}\right),\qquad \xi=\frac xt\in I.
\]
The same bound, including one total \(x\)-derivative, holds at
\(\lambda\in\mathcal K_0\) for the outer matrices and their inverses, with
\(R_I^{\rm fix}=R(k;x,t\mid\widehat\sigma(I))\). Evaluation precedes
\(D_x\), and the inverse bound follows from the resolvent identity.

The freezing estimate and \eqref{3.33} give
\[
u(x,t)=u_{\rm sol}(x,t;\widehat\sigma(I))+t^{-1/2}h+O(t^{-3/4}),
\]
uniformly on \(I\), with identity outer model for empty data. Taking \(I=[M^{-1},M]\) proves Theorem~\ref{th1}{\rm(ii)}. Combining this estimate and the Sector~I estimate with the reflectionless decomposition \eqref{6.2} proves the soliton-resolution assertions in Theorem~\ref{th3}, with unselected solitons exponentially small.
\subsection{Asymptotic stability}
An \(L^{1,1}\) perturbation controls persistent data away from \(0\). Extra near-zero poles may have nonuniform decay, so constants may depend on the perturbed datum. Use \(\mathcal Y_a\) from \eqref{2.7}.

Let $\sigma_*=\{(\eta_j^*,\beta_j^*)\}_{j=1}^N$, with $0<\eta_1^*<\cdots<\eta_N^*$ and $\beta_j^*>0$, be reference scalar data. Set $u_N^*=u_{\rm sol}(\cdot,\cdot;\sigma_*)$, $v_j^*=9(\eta_j^*)^4$, and let $\beta_j^{*,+}$ be the corresponding outgoing constants. Fix $0<\xi_0<v_1^*$ and write $\varepsilon=\|u_0-u_N^*(\cdot,0)\|_{L^{1,1}}$.

Assume the reference datum is SK-realizable, has exactly \(N\) admissible representative orbits, and has no spectral singularity away from \(0\). Choose \(r_0\) with no reference zero in \(|k|\le2r_0\), \(45r_0^4<\xi_0/2\), and put
\[
H_*=2+\max\{1,\xi_0,v_N^*+1\},\qquad I_*=[\xi_0,2H_*].
\]
Here \(H_*\) is perturbation-independent. The case \(\varepsilon=0\) follows from Theorem~\ref{th3}, so assume \(0<\varepsilon\le\varepsilon_0\) and call the \(N\) continued orbits persistent.
The perturbed datum is called \emph{admissible relative to \((\xi_0,\sigma_*)\)} if the following conditions hold.
	
	\emph{(i)} The scattering data satisfy the endpoint condition in Assumption~\ref{ass:1.1}{\rm(a)}, have no spectral singularity on \(\Sigma\cap\{0<|k|<r_0\}\), and their additional representative zeros in \(0<|k|<r_0\) are normal-simple and exclusive and generate a finite pairwise-disjoint family of complete orbits.
	
	\emph{(ii)} In the representative coordinates labeled by the reference zero disks, there exist open neighborhoods
	\[
	\begin{aligned}
		&(\eta_1^*,\ldots,\eta_N^*,\log\beta_1^*,\ldots,
		\log\beta_N^*)\in\mathcal V_*\Subset
		\{(\eta_1,\ldots,\eta_N,b_1,\ldots,b_N):
		0<\eta_1<\cdots<\eta_N\},\\
		&\mathfrak G_N(\sigma_*)\in\mathcal U_*,\qquad
		\Psi:\mathcal U_*\to\mathcal V_* \quad\hbox{locally Lipschitz},
	\end{aligned}
	\]
	such that, with
	\[
\begin{aligned}
	\mathfrak{G} _{N}^{\log}(\eta _1,...,\eta _N,b_1,...,b_N):=\mathfrak{G} _N\bigl\{ (\eta _j,e^{b_j}):1\le j\le N\} \bigr) ,
\\
\mathfrak{G} _{N}^{\log}(\mathcal{V} _*)\subset \mathcal{U} _*,\qquad \Psi \circ \mathfrak{G} _{N}^{\log}=\mathrm{id}_{\mathcal{V} _*}.
	\end{aligned}
\]
	For data \(\sigma=\{(\lambda_j,q_j)\}_{j=1}^N\) and \(\sigma'=\{(\lambda_j',q_j')\}_{j=1}^N\), labeled by these disks in the corresponding \(c\)- or \(d\)-coordinates, write
	\(\|\sigma-\sigma'\|_{\rm rep}:=\max_j(|\lambda_j-\lambda_j'|+|q_j-q_j'|)\). The datum \(\widehat\sigma(I_*)\) belongs to \(\mathfrak G_N^{\log}(\mathcal V_*)\), and its selected scalar preimage is defined by
	\[
	\begin{aligned}
		(\eta_1,\ldots,\eta_N,
		\log\widetilde\beta_1,\ldots,\log\widetilde\beta_N):=\Psi(\widehat\sigma(I_*)),\\
		\widehat\sigma(I_*)=\mathfrak G_N\bigl(\{(\eta_j,\widetilde\beta_j):1\le j\le N\}\bigr),
		\quad 0<\eta_1<\cdots<\eta_N.
	\end{aligned}
	\]
	
	\emph{(iii)} Put \(R^{\rm fix}(k;x,t)=R(k;x,t\mid\widehat\sigma(I_*))\). Proposition~\ref{prop:3.2} and the determinant argument following RH Problem~\ref{rhp6} give unique solvability and \(\det R^{\rm fix}=1\). There is a fixed union \(\mathcal D_*\) of mutually disjoint pole disks, disjoint from \(\Sigma\cup\{0\}\), such that
	\[
	\sup_{\substack{t\ge1,\ x/t\in I_*\\
			k\in\mathbb C\setminus\mathcal D_*}}
	\sum_{q=0}^{1}\|\partial_x^qR^{\rm fix}(k;x,t)\|<\infty .
	\]
	Here \(\partial_x\) fixes \(t,k\). The adjugate and inverse differentiation formulas give the same bounds for \((R^{\rm fix})^{-1}\).

\begin{lemma}\label{lem:6.2}
	Let \(I=[\xi_-,\xi_+]\Subset(0,\infty)\). Let
	\(\widehat\sigma\) be a fixed SK-realizable rational datum formed by finitely many disjoint complete normal-simple pole orbits with distinct velocities
	\(v_j\in\operatorname{int}I\). Assume that the corresponding fixed outer model \(R^{\rm fix}\) is uniquely solvable and uniformly bounded away from the pole disks \(\mathcal D\), where \(\mathcal D\cap(\Sigma\cup\{0\})=\emptyset\).
	Let the moving datum have the same pole set, with representative constants satisfying
	\begin{equation}\label{6.9}
		q_j^{\mathrm{mov}}(\xi)=\widehat q_jg_j(\xi),\qquad
		g_j(v_j)=1,\qquad
		\sup_{\xi\in I}\max_j
		\left(
		|g_j(\xi)|+|g_j(\xi)|^{-1}+|\partial_\xi g_j(\xi)|
		\right)<\infty .
	\end{equation}
	Writing
	\[
	M^{\rm mov}=I_3+k^{-1}M_1^{\rm mov}+O(k^{-2}),
	\]
	and setting
	\[
	\Delta_0=M^{\rm mov}-R^{\rm fix},\qquad
	\Delta_{-1}=(M^{\rm mov})^{-1}-(R^{\rm fix})^{-1},
	\]
	there exists \(T_I\ge1\) such that the moving problem is uniquely solvable for
	\(t\ge T_I\), \(x/t\in I\), and
	\[
	\sup_{\substack{x/t\in I\\k\in\mathbb C\setminus\mathcal D}}
	\sum_{q=0}^{1}
	\left(
	\|D_x^q\Delta_0(k)\|
	+\|D_x^q\Delta_{-1}(k)\|
	\right)
	+
	\sup_{x/t\in I}\sum_{q=0}^{1}
	\|D_x^q(M_1^{\rm mov}-R_1^{\rm fix})\|
	\le C_I\frac{\log(2+t)}{t}.
	\]
	The same estimate holds along bounded \(C^1\) evaluation curves, including the moving stationary points.
\end{lemma}

\begin{proof}
Use common pole circles and balanced charts. In logarithmic tubes,
\eqref{6.9} gives \(g_{\lambda,j}^{\pm1}-1=O_I(\log(2+t)/t)\) and
\(D_xg_{\lambda,j}=O_I(t^{-1})\); outside, the balanced coefficients decay
exponentially. Hence the circle jumps differ by \(O_I(\log(2+t)/t)\), including
one derivative. Conjugation by the frozen solution and the differentiated
Beals-Coifman equation give the same bound for \(E^{\pm1}-I_3\), \(E_1\),
and their \(k\)-derivatives off the pole disks. Reversing the transformations
proves the claim, also along curves since
\(D_x\lambda_\nu(x/t)=t^{-1}\lambda_\nu'(x/t)\).
\end{proof}

Define
\begin{equation}\label{6.10}
	p_j=\sqrt{3}\eta _j,\qquad v_j=p_{j}^{4},\qquad \beta _{j}^{+}=\widetilde{\beta }_j\prod_{i<j}{A_{ij},\qquad x_{j}^{+}}=-p_{j}^{-1}\log \beta _{j}^{+},
\end{equation}
and
\begin{equation}\label{6.11}
	U_j^+(x,t)=\frac{p_j^2}{4}
	\operatorname{sech}^2\!\left(\frac{p_j(x-v_jt-x_j^+)}2\right).
\end{equation}
Compatibility makes \eqref{6.11} the persistent outgoing solitons, with omitted-orbit corrections contained in \(\widetilde\beta_j\).

\begin{theorem}[Asymptotic stability]
	\label{th4}
	Let $\sigma_*$ and $\xi_0$ be as above. There exists
	\(\varepsilon_0=\varepsilon_0(\xi_0,\sigma_*)>0\) such that if
	\[
	\varepsilon =\parallel u_0-u_{N}^{*}(\cdot ,0)\parallel _{L^{1,1}}\le \varepsilon _0,
	\]
	and, if \(\varepsilon>0\), the scattering data of \(u_0\) are admissible relative to \((\xi_0,\sigma_*)\) in the sense specified above, then the following conclusions hold.
	
	Let \(U_j^+\) denote the outgoing profiles in \eqref{6.10}-\eqref{6.11}, with parameters
	\((\eta_j,\widetilde\beta_j,\beta_j^+)\) ordered by velocity. For \(\varepsilon=0\), these parameters are set equal to \((\eta_j^*,\beta_j^*,\beta_j^{*,+})\). Then there exist
	\(C_{\xi_0}(u_0),T_{\xi_0}(u_0)>0\) such that
	\[
	\sup_{x/t\ge\xi_0}\left| u(x,t)-\sum_{j=1}^N{U_{j}^{+}(x,t)} \right|\le C_{\xi _0}(u_0)\left( \varepsilon t^{-1/2}+t^{-3/4} \right) ,\qquad t\ge T_{\xi _0}(u_0).
	\]
	Moreover,
	\[
	\max_{1\le j\le N} \left( |\eta _j-\eta _{j}^{*}|+|\log \beta _{j}^{+}-\log \beta _{j}^{*,+}| \right) \le C_{\xi _0,\sigma _*}\varepsilon .
	\]
\end{theorem}

By \eqref{3.13}, after decreasing \(\varepsilon_0\), every additional orbit
has velocity in \((-\xi_0/2,\xi_0/2)\), hence its switched circle is
exponentially small on \(x/t\ge\xi_0\). The persistent velocities remain
distinct in \((\xi_0,H_*)\), and the selected datum on \(I_*\) contains exactly
these \(N\) orbits. With \(a_0=\tfrac12(\xi_0/45)^{1/4}\), the saddles stay
outside \(|k|\le a_0\), and \eqref{2.8} gives
\begin{equation}\label{6.12}
\|r\|_{\mathcal Y_{a_0}}\le C_{\xi_0}\varepsilon.
\end{equation}
The transmission factors are \(1+O(\varepsilon^2)\) on \(I_*\). Defining the
moving-to-frozen ratios \(g_j\) by evaluating each retained transmission factor
at \(\xi\) and at its orbit velocity, separated Cauchy formulas verify
\eqref{6.9}. Lemma~\ref{lem:6.2} and Proposition~\ref{prop:3.4} then yield
\[
u=u_{\rm sol}(x,t;\widehat\sigma(I_*))+t^{-1/2}h_*
+O_{\xi_0,u_0}\!\left(t^{-3/4}+\frac{\log(2+t)}t+e^{-ct}\right).
\]
By \eqref{6.12}, \(\sup|h_*|\le C_{\xi_0}(u_0)\varepsilon\), and
Theorem~\ref{th3} replaces the reflectionless term by \(\sum_jU_j^+\) with
exponentially small error. Proposition~\ref{prop:3.5} gives the same bound on
the far half-cone. Finally, scattering Lipschitz continuity,
\(\mathfrak t_{\alpha(j)}(\lambda_j;v_j)=1+O(\varepsilon^2)\), and the locally
Lipschitz inverse map \(\Psi\) give the stated parameter estimate.

\section*{Acknowledgments}\label{s:8} 	
    This work was supported by the National Natural Science Foundation of China under Grant No. 12371255, the Fundamental Research Funds for the Central Universities of CUMT under Grant No. 2024ZDPYJQ1003, and the Postgraduate Research \& Practice Program of Education \& Teaching Reform of CUMT under Grant No. 2025YJSJG031.\\

  \textbf{Data availibility}: The data which supports the findings of this study is available within the article.\\

\textbf{Conflict of interest}: The authors declare no conflict of interest.

\appendix
\section{The model RH problem \texorpdfstring{$M^{X_{A,B}}$}{M X(A,B)}}\label{app:PC-model}
Let
\[
\mathsf X_j=e^{(2j-1)\pi i/4}\mathbb R_+, \quad j=1,2,3,4, \qquad \mathsf X=\bigcup_{j=1}^4\mathsf X_j,
\]
where the four rays are oriented from \(0\) to \(\infty\). For \(q\in\mathbb C\) with \(|q|<1\), set
\[
d(q)=1-|q|^2, \qquad \nu(q)=-\frac{1}{2\pi}\log d(q).
\]
The contours and branch cuts are shown in Figure~\ref{fig:PC-cross}.

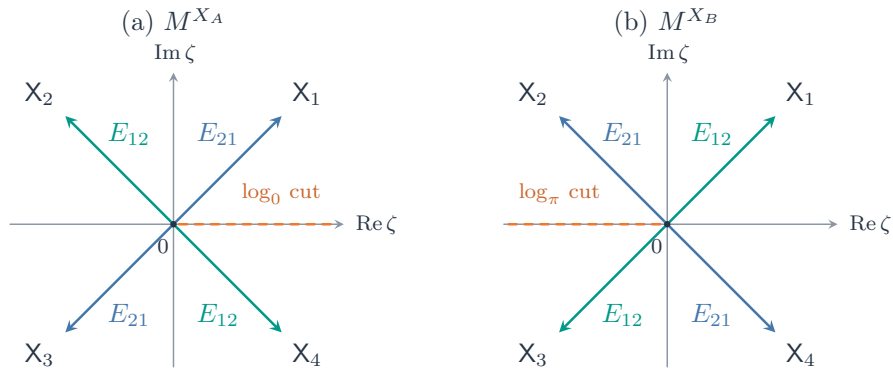
\begin{figure}[htbp]
\centering
\definecolor{SKPCBlue}{HTML}{4477AA}
\definecolor{SKPCTeal}{HTML}{009988}
\definecolor{SKPCOrange}{HTML}{EE7733}
\definecolor{SKPCAxis}{HTML}{8A96A4}
\definecolor{SKPCInk}{HTML}{263445}
\begin{tikzpicture}[x=0.92cm,y=0.92cm,>=stealth,
  line cap=round,line join=round,font=\small,text=SKPCInk]
  \begin{scope}
    \node at (0,2.90) {(a) $M^{X_A}$};
    \draw[SKPCAxis,->,line width=0.55pt] (-2.35,0)--(2.45,0);
    \node[anchor=north west,font=\scriptsize,xshift=2pt,yshift=-2pt]
      at (2.40,0.35) {$\operatorname{Re}\zeta$};
    \draw[SKPCAxis,->,line width=0.55pt] (0,-2.05)--(0,2.18);
    \node[anchor=south west,font=\scriptsize,xshift=2pt,yshift=2pt]
      at (-0.5,2.10) {$\operatorname{Im}\zeta$};
    \draw[SKPCOrange,dashed,line width=0.95pt] (0.08,0)--(2.28,0);
    \node[above,font=\scriptsize,text=SKPCOrange!85!black]
      at (1.55,0.15) {$\log_0$ cut};
    \draw[SKPCBlue,->,line width=0.95pt] (0,0)--(45:2.20)
      node[above right,text=SKPCInk] {$\mathsf X_1$};
    \draw[SKPCTeal,->,line width=0.95pt] (0,0)--(135:2.20)
      node[above left,text=SKPCInk] {$\mathsf X_2$};
    \draw[SKPCBlue,->,line width=0.95pt] (0,0)--(225:2.20)
      node[below left,text=SKPCInk] {$\mathsf X_3$};
    \draw[SKPCTeal,->,line width=0.95pt] (0,0)--(315:2.20)
      node[below right,text=SKPCInk] {$\mathsf X_4$};
    \node[text=SKPCBlue] at (0.64,1.30) {$E_{21}$};
    \node[text=SKPCTeal] at (-0.64,1.30) {$E_{12}$};
    \node[text=SKPCBlue] at (-0.64,-1.30) {$E_{21}$};
    \node[text=SKPCTeal] at (0.64,-1.30) {$E_{12}$};
    \fill[SKPCInk] (0,0) circle (1.2pt);
    \node[below left,font=\scriptsize,xshift=2pt,yshift=-2pt]
      at (0,0) {$0$};
  \end{scope}
  \begin{scope}[shift={(7.10,0)}]
    \node at (0,2.90) {(b) $M^{X_B}$};
    \draw[SKPCAxis,->,line width=0.55pt] (-2.35,0)--(2.45,0);
    \node[anchor=north west,font=\scriptsize,xshift=2pt,yshift=-2pt]
      at (2.40,0.35) {$\operatorname{Re}\zeta$};
    \draw[SKPCAxis,->,line width=0.55pt] (0,-2.05)--(0,2.18);
    \node[anchor=south west,font=\scriptsize,xshift=2pt,yshift=2pt]
      at (-0.5,2.10) {$\operatorname{Im}\zeta$};
    \draw[SKPCOrange,dashed,line width=0.95pt] (-2.28,0)--(-0.08,0);
    \node[above,font=\scriptsize,text=SKPCOrange!85!black]
      at (-1.55,0.15) {$\log_\pi$ cut};
    \draw[SKPCTeal,->,line width=0.95pt] (0,0)--(45:2.20)
      node[above right,text=SKPCInk] {$\mathsf X_1$};
    \draw[SKPCBlue,->,line width=0.95pt] (0,0)--(135:2.20)
      node[above left,text=SKPCInk] {$\mathsf X_2$};
    \draw[SKPCTeal,->,line width=0.95pt] (0,0)--(225:2.20)
      node[below left,text=SKPCInk] {$\mathsf X_3$};
    \draw[SKPCBlue,->,line width=0.95pt] (0,0)--(315:2.20)
      node[below right,text=SKPCInk] {$\mathsf X_4$};
    \node[text=SKPCTeal] at (0.64,1.30) {$E_{12}$};
    \node[text=SKPCBlue] at (-0.64,1.30) {$E_{21}$};
    \node[text=SKPCTeal] at (-0.64,-1.30) {$E_{12}$};
    \node[text=SKPCBlue] at (0.64,-1.30) {$E_{21}$};
    \fill[SKPCInk] (0,0) circle (1.2pt);
    \node[below left,font=\scriptsize,xshift=2pt,yshift=-2pt]
      at (0,0) {$0$};
  \end{scope}
\end{tikzpicture}
\caption{Contours and branch cuts for $M^{X_A}$ and $M^{X_B}$.}
\label{fig:PC-cross}
\end{figure}

The branches \(\log_0\) and \(\log_\pi\) are cut along \(\mathbb R_+\) and \(\mathbb R_-\), respectively. For \(\nu=\nu(q)\), define
\[
	v_A(\zeta;q)=\begin{cases}
		I_3+\dfrac{\overline{q}}{d(q)}e^{-2i\nu\log_0\zeta}e^{i\zeta^2/2}E_{21},& \zeta\in\mathsf X_1,\\
		I_3+q e^{2i\nu\log_0\zeta}e^{-i\zeta^2/2}E_{12},& \zeta\in\mathsf X_2,\\
		I_3-\overline{q}e^{-2i\nu\log_0\zeta}e^{i\zeta^2/2}E_{21},& \zeta\in\mathsf X_3,\\
		I_3-\dfrac{q}{d(q)}e^{2i\nu\log_0\zeta}e^{-i\zeta^2/2}E_{12},& \zeta\in\mathsf X_4,
	\end{cases},
\]
and
\[
	v_B(\zeta;q)=\begin{cases}
		I_3+\overline{q}e^{-2i\nu\log_\pi\zeta}e^{i\zeta^2/2}E_{12},& \zeta\in\mathsf X_1,\\
		I_3+\dfrac{q}{d(q)}e^{2i\nu\log_\pi\zeta}e^{-i\zeta^2/2}E_{21},& \zeta\in\mathsf X_2,\\
		I_3-\dfrac{\overline{q}}{d(q)}e^{-2i\nu\log_\pi\zeta}e^{i\zeta^2/2}E_{12},& \zeta\in\mathsf X_3,\\
		I_3-q e^{2i\nu\log_\pi\zeta}e^{-i\zeta^2/2}E_{21},& \zeta\in\mathsf X_4.
	\end{cases}
\]

\begin{rhpa}\label{rhp:PC-model}
	For \(\diamond\in\{A,B\}\), find \(M^{X_\diamond}(\zeta;q)\) satisfying:
	\begin{enumerate}
		\item[(i)] \(M^{X_\diamond}\) is analytic in \(\mathbb C\setminus\mathsf X\) and has continuous boundary values on \(\mathsf X\setminus\{0\}\).
		\item[(ii)] The boundary values satisfy
		\[
		M^{X_\diamond}_+(\zeta;q)=M^{X_\diamond}_-(\zeta;q)v_\diamond(\zeta;q),\qquad \zeta\in\mathsf X\setminus\{0\}.
		\]
		\item[(iii)] \(M^{X_\diamond}\) is bounded at \(\zeta=0\) in each sector.
		\item[(iv)] As \(\zeta\to\infty\),
		\[
		M^{X_\diamond}(\zeta;q)=I_3+O(\zeta^{-1}).
		\]
	\end{enumerate}
\end{rhpa}

For every \(|q|<1\), RH Problem~\ref{rhp:PC-model} has a unique solution, and
\[
	\det M^{X_A}(\zeta;q)=\det M^{X_B}(\zeta;q)=1.
\]
As \(\zeta\to\infty\),
\begin{equation}\label{A.1}
	M^{X_A}(\zeta;q)=I_3+\frac{M_1^{X_A}(q)}{\zeta}+O(\zeta^{-2}),
	\qquad
	M^{X_B}(\zeta;q)=I_3+\frac{M_1^{X_B}(q)}{\zeta}+O(\zeta^{-2}),
\end{equation}
where
\[
	M_1^{X_A}(q)=
	\begin{pmatrix}
		0&\beta_{12}^A(q)&0\\
		\beta_{21}^A(q)&0&0\\
		0&0&0
	\end{pmatrix},
	\qquad
	M_1^{X_B}(q)=
	\begin{pmatrix}
		0&\beta_{12}^B(q)&0\\
		\beta_{21}^B(q)&0&0\\
		0&0&0
	\end{pmatrix},
\]
with
\begin{equation}\label{A.2}
	\beta_{12}^A(q)=\frac{\sqrt{2\pi}e^{-\pi i/4}e^{-5\pi\nu(q)/2}}{\overline{q}\,\Gamma(-i\nu(q))},
	\qquad
	\beta_{21}^A(q)=\frac{\sqrt{2\pi}e^{\pi i/4}e^{3\pi\nu(q)/2}}{q\,\Gamma(i\nu(q))},
\end{equation}
and
\begin{equation}\label{A.3}
	\beta_{12}^B(q)=\frac{\sqrt{2\pi}e^{\pi i/4}e^{-\pi\nu(q)/2}}{q\,\Gamma(i\nu(q))},
	\qquad
	\beta_{21}^B(q)=\frac{\sqrt{2\pi}e^{-\pi i/4}e^{-\pi\nu(q)/2}}{\overline{q}\,\Gamma(-i\nu(q))}.
\end{equation}
At \(q=0\), the four coefficients are defined by continuity as zero and
\(M^{X_A}(\zeta;0)=M^{X_B}(\zeta;0)=I_3\). Moreover, for every fixed
\(0<q_*<1\), there is a constant \(C(q_*)>0\) such that, for
\(\diamond\in\{A,B\}\),
\[
	\sup_{|q|\le q_*}\sup_{\zeta\in\mathbb C\setminus\mathsf X}
	\left(\|M^{X_\diamond}(\zeta;q)\|+
	\|(M^{X_\diamond}(\zeta;q))^{-1}\|\right)\le C(q_*),
\]
and the remainders in \eqref{A.1} are uniform for \(|q|\le q_*\), including the one-sided limits on \(\mathsf X\).

The embedded parabolic-cylinder construction gives existence, \eqref{A.1}-\eqref{A.3}, and uniform bounds for \(|q|\le q_*<1\); see \cite{d3} and Appendix A of \cite{w3}. Unit-determinant jumps, removability at zero, and Liouville's theorem give determinant one and uniqueness. Continuity at \(q=0\) follows from \(\nu=|q|^2/(2\pi)+O(|q|^4)\) and \(1/\Gamma(\pm i\nu)=\pm i\nu+O(\nu^2)\). The cofactor formula gives the inverse bound.

\section{The Painlev\'e RH model}\label{app1}
By Lemma~\ref{lem:4.2}, deform the Painlev\'e contour to the canonical rays
\[
\vartheta_j=\frac{(2j-1)\pi}{12},\qquad j=1,\ldots,12.
\]
Let \(\Theta_j\) be the phase multiplying \(S_j-I_3\). In counterclockwise order,
\[
\begin{gathered}
	(\Theta_1,\ldots,\Theta_6)=(\Theta_{21},\Theta_{32},\Theta_{23},
	\Theta_{31},\Theta_{13},\Theta_{21}),\\
	\Theta_{j+6}=-\Theta_j,\qquad j=1,\ldots,6,
\end{gathered}
\]
and \(e^{\operatorname{ad}H}(S_j-I_3)=e^{\Theta_j}(S_j-I_3)\).

For \(\sigma=\overline{\widetilde r_1(0)}\), \(\tau=\widetilde r_2(0)\), multiplication gives
\begin{equation}\label{B.1}
	S_1S_2\cdots S_{12}=I_3 \quad\Longleftrightarrow\quad\sigma+\tau+\overline\sigma\,\overline\tau=0.
\end{equation}
Equation~\eqref{4.7} verifies the latter identity.

Uniformly for \(s\in K\Subset\mathbb R\) and all rays,
\[
\mathrm{Re}\Theta _j(re^{i\vartheta _j};s)=-\frac{9\sqrt{3}}{5}\cos \frac{\pi}{12}\,r^5+O_K(r)\le -cr^5\quad (r\ge R),\quad V_{j}^{\mathrm{P}}-I_3=O(e^{-c|\lambda |^5}),
\]
uniformly on the model rays.

Write
\[
	M^{\mathrm{P}}(\lambda ;s)=I_3+\sum_{j=1}^4{\frac{M_{j}^{\mathrm{P}}(s)}{\lambda ^j}}+O(\lambda ^{-5}),\qquad \lambda \rightarrow \infty ,
\]
with
\[
	M^{\rm P}(\lambda;s)=\mathcal A M^{\rm P}(\omega\lambda;s)\mathcal A^{-1}=\mathcal B\overline{M^{\rm P}(\bar\lambda;s)}\mathcal B.
\]
Set \(M_j=M_j^{\rm P}\).

Assume \eqref{B.1}, put \(h=\mu_{\rm BC}(V^{\rm P}-I_3)\), \(\Psi=M^{\rm P}e^H\), and \(K_j=(S_1\cdots S_j)^{-1}\). Then \(\Psi K_j\) glue on a punctured disk. Lemma~\ref{lem:4.3} excludes negative Laurent terms in \(M^{\rm P}=I_3+Ch\), removing the vertex. Unipotence and Liouville give
\begin{equation}\label{B.2}
	\det M^{\rm P}(\lambda;s)\equiv1.
\end{equation}
Consequently, \((M^{\rm P})^{-1}\) has the same analytic domain.

Under \eqref{4.14}, Proposition~\ref{prop:4.2} and the differentiated Beals-Coifman equation give \eqref{4.15} and uniform matching-circle bounds. Lemma~\ref{lem:4.3} controls the vertex, localized differentiation controls the other sectors, and \eqref{B.2} controls the inverse.

\noindent\textbf{Symmetries of the expansion.}
The first reduction gives
\[
M_j=\omega ^{-j}\mathcal{A} M_j\mathcal{A} ^{-1},\qquad j\ge 1.
\]
For \(\rho^3=1\), set
\[
\mathcal C_\rho(x_1,x_2,x_3)=
\begin{pmatrix}
	x_1&x_2&x_3\\
	\rho x_3&\rho x_1&\rho x_2\\
	\rho^2x_2&\rho^2x_3&\rho^2x_1
\end{pmatrix}.
\]
Then
\begin{equation}\label{B.3}
	\begin{aligned}
		M_1&=\mathcal C_{\omega^2}(a_1,a_2,a_3),&
		M_2&=\mathcal C_{\omega}(b_1,b_2,b_3),\\
		M_3&=\mathcal C_{1}(c_1,c_2,c_3),&
		M_4&=\mathcal C_{\omega^2}(d_1,d_2,d_3).
	\end{aligned}
\end{equation}
The Schwarz reduction gives
\[
(x_1,x_2,x_3)=\omega ^j(\overline{x_1},\overline{x_3},\overline{x_2}),
\]
for
\[
(j;x_1,x_2,x_3)\in \{(1;a_1,a_2,a_3),(2;b_1,b_2,b_3),(3;c_1,c_2,c_3),(4;d_1,d_2,d_3)\}.
\]
On the real mSK branch \eqref{4.24},
\[
a_2=\omega a_3,\qquad p(s)=3a_3(s)=3(M_1(s))_{13}\in\mathbb R .
\]
We henceforth use this branch.

\noindent\textbf{The Lax pair.}
For $\Psi=M^{\rm P}e^H$, the jumps are the constant matrices $S_j$. The cyclic identity makes $F=\Psi K_j$ and $\partial_sF$ single-valued and holomorphic at zero, with $\det F=1$. Hence
\[
\mathcal U:=\Psi_s\Psi^{-1}=F_sF^{-1},\qquad
\mathcal V:=\Psi_\lambda\Psi^{-1}=F_\lambda F^{-1}
\]
are entire.
Let \(\widehat M_j\) be the coefficients of \((M^{\rm P})^{-1}\) at infinity, determined by \(M^{\rm P}(M^{\rm P})^{-1}=I_3\). Liouville's theorem gives
\[
\mathcal U=\lambda\Lambda+[M_1,\Lambda],\qquad
\mathcal V=\sum_{q=0}^4\lambda^q\mathcal V_q,
\]
where, with \(M_0=\widehat M_0=I_3\),
\[
\mathcal V_q=\delta_{q0}s\Lambda+
9\!\!\sum_{\substack{a,b\ge0\\a+b=4-q}}M_a\Lambda^2\widehat M_b,\qquad 0\le q\le4.
\]
The \(s\)-equation gives
\[
M_j'=[\Lambda,M_{j+1}]+[M_1,\Lambda]M_j,\qquad j=1,2,3.
\]
With \eqref{B.3}, this yields
\begin{equation}\label{B.4}
\begin{aligned}
a_1'&=3\omega^2a_3^2,\\
b_3&=\omega a_1a_3+\frac{a_3'}{\omega-1}
+\frac{\omega a_3^2}{\omega+1},\\
b_2&=a_1a_3-\frac{a_3'}{\omega-1}-a_3^2.
\end{aligned}
\end{equation}
The same recurrence eliminates the remaining coefficients. In the compatibility relation
\[
\mathcal V_s-\mathcal U_\lambda+[\mathcal V,\mathcal U]=0,
\]
the constant coefficient of the \((1,3)\) entry, simplified using \eqref{B.4}, gives
\[
a_3^{(4)}-45a_3^2a_3''-45a_3(a_3')^2-15a_3'a_3''+81a_3^5-sa_3=0.
\]
With \(p=3a_3\), this is \eqref{4.25}.

	\bibliographystyle{plain}

\end{document}